\documentclass[runningheads]{lmcs}
\pdfoutput=1
\usepackage[utf8]{inputenc}

\usepackage{lastpage}
\lmcsdoi{22}{1}{10}
\lmcsheading{}{\pageref{LastPage}}{}{}%
{Jan.~22,~2025}{Feb.~27,~2026}{}

\input{imports}

\definecolor{orange}{RGB}{255,145,0}

\DeclareMathOperator{\inj}{\hookrightarrow}

\DeclareMathOperator{\nv}{nv}
\DeclareMathOperator{\sv}{sv}

\DeclareMathOperator{\shift}{Shift}
\DeclareMathOperator{\leftshift}{Left}

\DeclareMathOperator{\repair}{Rep}

\DeclareMathOperator{\violation}{Imp}

\DeclareMathOperator{\id}{id}
\DeclareMathOperator{\track}{tr}
\DeclareMathOperator{\der}{der}

\DeclareMathOperator{\pre}{Pre}
\DeclareMathOperator{\post}{Post}
\DeclareMathOperator{\Over}{OL}
\DeclareMathOperator{\OverEq}{OL_{\sim}}

\DeclareMathOperator{\OverCon}{OL_{ind}}
\DeclareMathOperator{\OverPre}{OL_{dep}}
\DeclareMathOperator{\chk}{check}

\newcommand{\rle}[5]{#1 \overset{#2}{\longleftarrow\joinrel\rhook} #3 \overset{#4}{\lhook\joinrel\longrightarrow}#5 }
\newcommand{\completeRle}{\rle{L}{l}{K}{r}{R}}

\newcommand{\true}{\textsf{true}}
\newcommand{\false}{\textsf{false}}
\newcommand{\OverEQ}{\Over_{\sim}}

\begin{document}
\title[Using weakest application conditions to rank graph transformations]{Using weakest application conditions to rank graph transformations for graph repair}


\author[L.~Fritsche]{Lars Fritsche\lmcsorcid{0000-0003-4996-4639}}[c]
\author[A.~Lauer]{Alexander Lauer\lmcsorcid{0009-0001-9077-9817}}[b]
\author[M.~Kratz]{Maximilian Kratz \lmcsorcid{0000-0001-7396-7763}}[a]
\author[A.~Schürr]{Andy Schürr\lmcsorcid{0000-0001-8100-1109}}[a]
\author[G.~Taentzer]{Gabriele Taentzer\lmcsorcid{0000-0002-3975-5238}}[b]


\address{Technical University Darmstadt, Darmstadt, Germany}
\email{\{maximilian.kratz, andy.schuerr\}@es.tu-darmstadt.de}

\address{Philipps-Universität Marburg, Marburg, Germany}
\email{alexander.lauer@uni-marburg.de, taentzer@mathematik.uni-marburg.de}

\address{Faktor Zehn GmbH, Germany}
\email{lars.fritsche@faktorzehn.de}

\begin{abstract}
	When using graphs and graph transformations to model systems, consistency is an important concern.
While consistency has primarily been viewed as a binary property, i.e., a graph is consistent or inconsistent with respect to a set of constraints, recent work has presented an approach to consistency as a graduated property.
This allows living with inconsistencies for a while and repairing them when necessary.
For repairing inconsistencies in a graph, we use graph transformation rules with so-called {\em impairment-indicating and repair-indicating application conditions} to understand how much repair gain certain rule applications would bring.
Both types of conditions can be derived from given graph constraints.
Our main theorem shows that the difference between the number of actual constraint violations before and after a graph transformation step can be characterised by the difference between the numbers of violated impairment-indicating and repair-indicating application conditions. 
This theory forms the basis for algorithms with look-ahead that rank graph transformations according to their potential for graph repair.
An evaluation shows that graph repair can be well-supported by rules with these new types of application conditions in terms of effectiveness and scalability.
This paper provides further theoretical results and an extended evaluation of the results presented in~\cite{fritsche2024using}.

\end{abstract}

\maketitle

\section{Introduction}
\label{chapter:introduction}
Graph transformation has proven to be a versatile approach for specifying and validating software engineering problems~\cite{HT20}.
This is true because graphs are an appropriate means for representing complex structures of interest, the constant change of structures can be specified by graph transformations, and there is a strong theory of graph transformation~\cite{EEPT06} that has been used to validate software engineering problems. 
When applying graph transformations, it is typically important that the processed graphs are consistent with respect to a given set of constraints.
Ensuring graph consistency involves two tasks.
First, to specify what a consistent graph is and to check whether a graph is indeed consistent, and second, to ensure that graph transformations preserve or even improve consistency.

Throughout this paper, we consider the \ac{CRA} problem as running example~\cite{BBL10}.
It concerns an optimal assignment of features (i.e., methods and attributes) to classes.
The constraints enforcing that each feature belongs to one and only one class are invariants for all operations modifying the class model.
To validate the quality of a class model, coupling and cohesion metrics are often used.
Related design guidelines, such as ``minimize dependencies across class boundaries", can also be formulated as constraints.
This example shows that constraints can serve different purposes; some are considered so essential that they must always be satisfied, while others are used for optimization; they may be violated to a certain extent, but the number of violations should be kept as small as possible. 

Nested graph constraints~\cite{HabelP09} provide a viable means of formulating graph properties.
The related notion of constraint consistency introduced in~\cite{HabelP09} is binary: a graph is either consistent or inconsistent.
Since graph repair is often only gradual, it is also interesting to consider graph consistency as a graduated property, as was done in~\cite{KosiolSTZ22}.
To support gradual repair, graph transformations are analyzed in~\cite{KosiolSTZ22} with respect to their potential to improve (sustain) the consistency level of a processed graph, i.e., to strictly reduce (preserve) the number of constraint violations in a graph.
A {\em static analysis approach} is presented in~\cite{KosiolSTZ22} for checking whether or not a graph transformation rule is always consistency-sustaining or -improving.
That approach does not yet support graph constraints with different priorities or propositional logic operators.
In addition, it does not support scenarios where a rule either improves or degrades graph consistency depending on the context of the matched and rewritten subgraph.

To mitigate these problems, we introduce a {\em new dynamic analysis approach that ranks rule matches and related rule applications according to their effect on improving graph consistency} as follows:
\begin{enumerate}
	\item Graph transformation rules are equipped with {\em impairment-indicating and repair-in\-di\-ca\-ting application conditions}.
	These conditions no longer block rule applications, but count the number of constraint violations introduced or removed by a given graph transformation step.
	They are derived from nested graph constraints based on the constructions presented in~\cite{HabelP09}.
	\item Our main theorem shows that {\em the number of additional constraint violations caused by a rule application can be characterized by the difference between the numbers of violations of the associated impairment-indicating and repair-indicating application conditions}.
	\item This theory forms the basis for {\em graph repair algorithms with a look-ahead} for the graph-consistency-improving potential of selectable rule applications.
	Based on a prototypical implementation of a greedy algorithm, we show in an initial evaluation that our approach is effective in the sense that it can reduce the number of violations, and it scales well with the available number of rule applications for each consistency-improving transformation step.
\end{enumerate}

This paper extends the approach presented in \cite{fritsche2024using} as follows:
\begin{enumerate}
	\item In Section~\ref{chapter:appl_conds} we present a new method for computing \emph{repair-} and \emph{impairment-indicating application conditions} that (a) reduces the total number of overlaps of the left-hand side of a rule and the \emph{premise} of a constraint that need to be considered, (b) results in a reduction of the total number of application conditions, (c) provides a simpler way to compute the change in consistency during a transformation through repair- and impairment-indicating application conditions, and (d) relies solely on \emph{nested graph condition}, whereas in the theory presented in~\cite{fritsche2024using}, additional information about the overlap with which an application condition was computed must be considered to correctly determine the change in consistency.
	\item We compare \emph{repair-} and \emph{impairment-indicating application conditions} with \emph{(direct) consistency-sustaining} and \emph{(direct) consistency-improving} transformations as introduced in~\cite{KosiolSTZ22} in Section~\ref{sec:comparisson} and show that \emph{repair-} and \emph{impairment-indicating application conditions} can be used to construct weakest direct consistency-sustaining and weakest direct consistency-improving application conditions, i.e., we obtain application conditions that are satisfied if and only if the transformation is direct consistency-sustaining (-improving).
	\item We further extend our evaluation using weak constraints, which model good patterns for class structures, and evaluate them on the well-known test data of the TTC16~\cite{DBLP:conf/staf/FleckTW16} \ac{CRA} case study and compared it to another tool using the \ac{GIPS} tool~\cite{gipsGCM2022}, which uses \ac{ILP} to compute the optimal solution for a given input graph.
\end{enumerate}

\section{Running Example}
\label{chapter:example}
To illustrate the problem addressed, we consider a variant of the \ac{CRA} problem~\cite{BBL10, DBLP:conf/staf/FleckTW16} for our running example. 
The \ac{CRA} aims to provide a high-quality design for object-oriented class structures.
Features, i.e., methods and attributes, with dependencies between them are assigned to classes so that the class design achieves high cohesion within classes and low coupling between classes.
We slightly adapt the \ac{CRA} problem by starting with a predefined class diagram and refactoring it by moving features between classes.
These refactoring steps aim to reduce dependencies between methods and attributes as well as methods and methods across different classes.
They can also group methods with similar attribute dependencies within the same class.

The amount of coupling and cohesion is well reflected by the \ac{CRA}-index, which is the cohesion ratio minus the coupling ratio within a class diagram~\cite{masoud2014clustering}. 
This means that dependencies should be maximised within classes, while minimising dependencies between different classes.

\subsubsection*{Class diagrams and feature dependencies}
\Cref{fig:example:example} shows a class diagram to be refactored on the left, consisting of three classes that model a snippet of an online shopping session.
On the right, we see a graph-like representation of the class diagram, in which the classes, methods, and attributes are represented as graph nodes.
The names of the classes, methods and attributes are used as node identifiers. 
There are also edges between methods and attributes, which model access dependencies of methods on attributes.

\begin{figure}
	\centering
	\includegraphics[scale = 0.9]{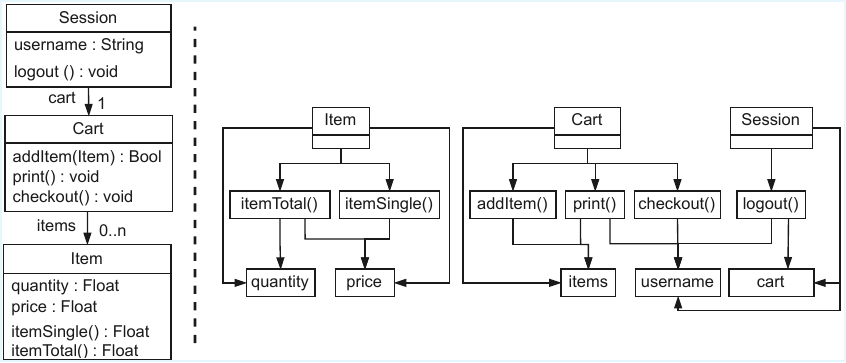}
	\caption{Class diagram (left) and feature dependencies (right)} 
	\label{fig:example:example}
\end{figure}

\subsubsection*{Rules and constraints}
Next, we define two simple refactoring rules and the constraints that they must obey, as shown in \cref{fig:example:rule+constraints}.
The \textsf{moveAttribute} rule moves an attribute from one class to another one, while the \textsf{moveMethod} rule does the same for methods.
Graph elements annotated with ``- -'' are to be deleted, while graph elements annotated with ``++'' are to be created.
From the set of language constraints on class structures, we select two constraints, $\textsf{h}_1$ and $\textsf{h}_2$, which impose two basic properties on class models.
We consider these constraints to be {\em hard}, i.e., constraints that must not be violated.
These constraints state that methods and attributes must not be contained in more than one class.

To model the \ac{CRA}-index as accurately as possible, we must represent the desired and undesired dependencies as graph constraints that state the following properties:
\begin{enumerate}
	\item For each pair of a method and an attribute contained within the same class, the method uses the attribute.
	
	\item A method cannot use an attribute contained in another class.
	
	\item Methods contained in the same class must depend on each other.
	\item A method cannot call a method contained in another class.
\end{enumerate}

We will formalise these properties as \emph{weak constraints} that can be violated, but which must be satisfied to the greatest possible extent.
We start with the constraint $\textsf{w}_2$, since the constraint $\textsf{w}_1$ is special and will be explained later.
In \cref{fig:example:rule+constraints}, the constraint $\textsf{w}_2$ says that methods and attributes contained in different classes must not be dependent on each other (reflecting the property described in (2) above).
Constraint $\textsf{w}_3$ says that methods contained in the same class must depend on each other (this reflects the property described in (3) above), and constraint $\textsf{w}_4$ says that methods contained in different classes must not depend on each other (this reflects the property described in (4) above).

A constraint that describes the property stated in (1) would have the same structure as $\textsf{w}_3$.
Such a constraint could be achieved by replacing the method $\textsf{m}_2'$ in $\textsf{w}_3$ with an attribute and deleting the right pattern within the brackets. 
However, to illustrate the concepts presented in the theory section as much as possible, we have selected a slightly different constraint $\textsf{w}_1$.
This constraint states that each pair of methods contained in the same class must have at least one common dependency on an attribute within the same class.
We use this constraint because it leads to application conditions that can illustrate each part of our construction nicely, in contrast to 
the application conditions for constraints $\textsf{w}_2$, $\textsf{w}_3$, and $\textsf{w}_4$.

Please note that we will only use $\textsf{w}_1$ and $\textsf{w}_2$ in the examples throughout the theoretical part of the paper.
Constraints $\textsf{w}_3$ and $\textsf{w}_4$ are only used in our extended evaluation. 

\begin{figure}
	\centering
	\includegraphics[scale = 0.8]{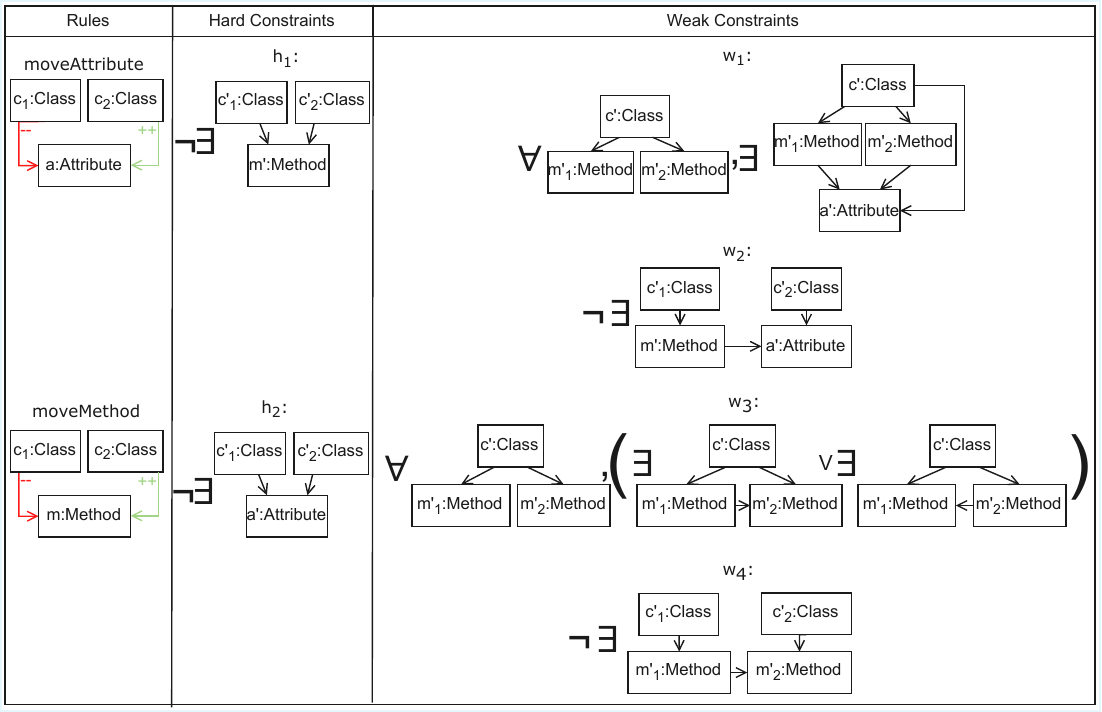}
	\caption{Refactoring rules \textsf{moveAttribute} and \textsf{moveMethod} and constraints: \textsf{h}$_1$: 
	A method must not be contained in more than one class. \textsf{h}$_2$: An attribute must not be contained in more than one class. \textsf{w}$_1$: Two methods within the same class $\textsf{c}'$ should have at least one dependency on a common attribute within the same class $\textsf{c}'$. \textsf{w}$_2$: A method should not depend on an attribute from another class. \textsf{w}$_3$: Methods contained within the same class must have at least one dependency on each other. \textsf{w}$_4$: A method should not depend on a method of another class.
	} 
	\label{fig:example:rule+constraints}
\end{figure}
Obviously, neither rule can violate the hard constraints, but moving features between classes can remove or add violations of the weak constraints, i.e., repair or impair the consistency of the graph under consideration.
For example, moving the \textsf{print()} method from the \textsf{Cart} to the \textsf{Session} class would repair a violation of $\textsf{w}_2$ because the \textsf{print()} depends on the \textsf{username} attribute.
However, this would introduce a new violation of the same constraint due to its dependency on the \textsf{items} attribute in the \textsf{Cart} class.

\subsubsection*{Ranking of graph transformations for graph repair}
We will now examine the impact of graph transformations on the consistency of the graph shown in \cref{fig:example:example} with respect to the weak constraints $\textsf{w}_1$ and $\textsf{w}_2$ by applying the  \textsf{moveMethod} and \textsf{moveAttribute} rules.
\Cref{table:ranking_rule_applications} shows the actions of the applied rules and the number of impairments and repairs for the constraints $\textsf{w}_1$ and $\textsf{w}_2$. 

The first rule application uses the rule $\textsf{moveMethod}$ to move \textsf{checkout()} from \textsf{Cart} to \textsf{Session}. 
As \textsf{checkout()} does not share a common attribute with \textsf{print()} and \textsf{addItem()}, which are contained within the same class, the rule application will perform four repairs for the constraint $\textsf{w}_1$ (due to the symmetry of $\textsf{w}_1$, the pairs (\textsf{checkout(), print()}) and (\textsf{checkout(), addItem()}) are counted twice). 
Additionally, one repair is performed for constraint $\textsf{w}_2$ because \textsf{checkout()} has a dependency on the \textsf{username}  attribute, which is contained in \textsf{Session}.

In the second transformation, \textsf{moveAttribute} is used to move \textsf{username} from \textsf{Session} to \textsf{Cart}.
This transformation makes two repairs to $\textsf{w}_1$. 
Once the transformation has been performed, both methods \textsf{print()} and \textsf{checkout()} have a dependency on \textsf{username}, which is now contained in \textsf{Cart}. 
Additionally, there are two repairs of $\textsf{w}_2$:
After the transformation, neither \textsf{print()} nor \textsf{checkout()} have a dependency on an attribute not contained in \textsf{Cart}. 
However, there is also an impairment of $\textsf{w}_2$ since \textsf{logout()} (contained in \textsf{Session}) has a dependency on \textsf{username} (which is now in \textsf{Cart}).

In the third transformation the method \textsf{print()} is moved from \textsf{Cart} to \textsf{Session}. 
This transformation performs two repairs of $\textsf{w}_1$. 
These are exactly the same repairs as described for the second transformation. 
Additionally, there is a repair and an impairment of $\textsf{w}_2$. 
The repair arises because both \textsf{print()} and \textsf{username} are contained in \textsf{Session} after the transformation. 
However, \textsf{items} (on which \textsf{print()} depends) remains in \textsf{Cart}.
Therefore, there is also an impairment of $\textsf{w}_2$.

In the fourth transformation, the \textsf{addItem()} method is moved from \textsf{Cart} to \textsf{Session}. 
This transformation also performs two repairs of $\textsf{w}_1$. 
After the transformation, \textsf{addItem()} and \textsf{checkout()} (which do not share an attribute) are no longer contained in the same class. 
Since \textsf{Session} will contain \textsf{addItem()} and \textsf{logout()}, which do not share a common attribute, this transformation also performs two impairments. 
Additionally, there is an impairment of $\textsf{w}_2$, since \textsf{items} and \textsf{addItem()} will be contained in different classes, even though \textsf{addItem()} has a dependency on \textsf{items}.

During a repair process, we would apply the $\textsf{moveMethod}$ rule to move \textsf{checkout()} from \textsf{Cart} to \textsf{Session}, as this maximises the consistency of $\textsf{w}_1$ and $\textsf{w}_2$. 

Based on this ranking information of rules, various optimisation algorithms can be implemented, including a greedy algorithm that always selects refactoring rule applications with the greatest possible consistency gain.
We will report on a greedy implementation of the \ac{CRA} in Section~\ref{chapter:evaluation}.

\begin{table*}[t]
	\centering
	\caption{Ranking of rule applications of \textsf{moveMethod} and \textsf{moveAttribute} w.r.t.~the consistency of the weak constraints $\textsf{w}_1$ and $\textsf{w}_2$}
	\label{table:ranking_rule_applications}
	\begin{tabular}{ccccccccccc}
		\toprule
		\multicolumn{1}{c}{Rule} & \multicolumn{3}{c}{Action} & \multicolumn{2}{c}{constraint $\textsf{w}_1$} & \multicolumn{2}{c}{constraint $\textsf{w}_2$}  \\
		\cmidrule(lr){1-1}
		\cmidrule(lr){2-4}
		\cmidrule(lr){5-6}    
		\cmidrule(lr){7-8}
		
		& Element &From & To & rep. & imp.  & rep. & imp. \\
		\midrule
		\textsf{moveMethod} & \textsf{checkout()} & \textsf{Cart} & \textsf{Session} & 4 & 0 & 1 & 0 \\
		\textsf{moveAttribute} & \textsf{username} & \textsf{Session} & \textsf{Cart} & 2 & 0 & 2 & 1 \\
		\textsf{moveMethod} & \textsf{print()} & \textsf{Cart} & \textsf{Session} & 2 & 0 & 1 & 1 \\
		\textsf{moveMethod} & \textsf{addItem()} & \textsf{Cart} & \textsf{Session} & 2 & 2 & 0 & 1 \\
		\bottomrule   
	\end{tabular}
\end{table*}

\subsubsection*{Repair-indicating application conditions}
As we have seen, there are many ways to apply the \textsf{moveMethod} and \textsf{moveAttribute} rules, even in this small example.
Therefore, it is not practically feasible to apply each possible transformation in order to find the one that increases consistency the most.
This is why we will present a new methodology that derives application conditions for each refactoring rule.
Based on the application conditions, we can calculate a look-ahead for each refactoring transformation (i.e., rule application) to determine how many repairs and impairments will be caused. 
These application conditions will not prevent rule applications, but will indicate which parts of a graph will be repaired or impaired. 
More specifically, a rule can have repair- and impairment-indicating application conditions. 
The set of violated (i.e., not satisfying) occurrences of repair-indicating application conditions shows where the graph would be repaired by the rule match under consideration.
Similarly, the set of violated (i.e., not holding) occurrences of impairment-indicating application conditions shows where the graph would be impaired.

\begin{figure}
	\includegraphics[scale = 0.8]{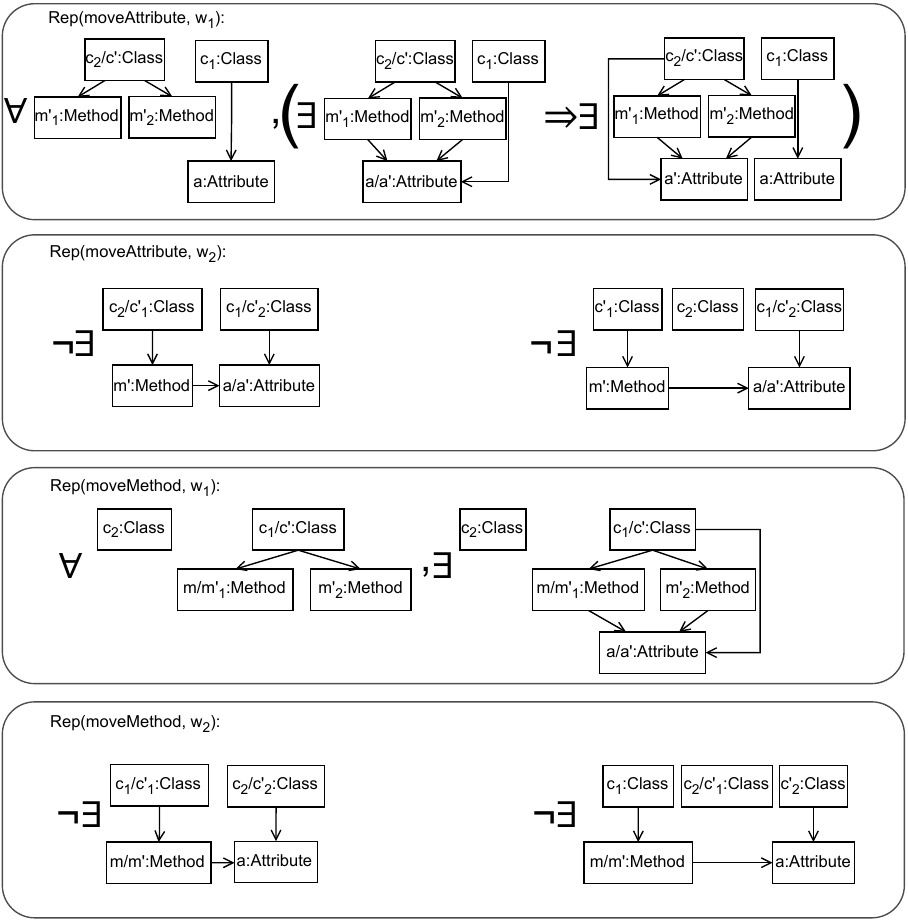}
	\caption{Repair-indicating application conditions for the rules \textsf{moveAttribute} and \textsf{moveMethod} w.r.t.~the constraints $\textsf{w}_1$ and $\textsf{w}_2$. The node labels implicitly denote the mappings of the LHS of the rule and the premise of the constraint. For example, the label $c_2/c'$ indicates that the node $c_2$ from the LHS of \textsf{moveAttribute} and $c'$ from the premise of $\textsf{w}_1$ are mapped to the node $c_2/c'$.}\label{fig:resulting_application_conditions}
\end{figure}

\Cref{fig:resulting_application_conditions} shows the repair-indicating application conditions for the rules \textsf{moveAttribute} and \textsf{moveMethod} with respect to the constraints $\textsf{w}_1$ and $\textsf{w}_2$.
The set of repair-indicating application conditions for \textsf{moveAttribute} and $\textsf{w}_1$, $\repair(\textsf{moveAttribute},\textsf{w}_1)$, contains one application condition. 
This condition checks that, for each pair of methods contained within the class to which the attribute is moved, they share a common attribute after the transformation has been performed and they do not share a common attribute before the transformation.
This can only happen if they  both have a dependency on the moved attribute (as modelled by the premise of the implication) and if they do not both have a dependency on another attribute contained in the same class (as modelled by the conclusion of the implication). 
Therefore,  a repair of $\textsf{w}_1$ will be found if this implication is not satisfied, i.e., if the premise of the implication is satisfied but the conclusion is not.

The set $\repair(\textsf{moveAttribute},\textsf{w}_2)$  consists of two repair-indicating application conditions.
The left application condition has the same form as the constraint $\textsf{w}_2$.
This application condition checks whether the moved attribute has been moved to a class containing a method that depends on it.
If so, a repair of $\textsf{w}_2$ will be found.
The application condition on the right checks whether a method that depends on the moved attribute exists that is not contained within either the attribute's original class nor its new class.
An occurrence of this pattern in a transformation does not affect the consistency of $\textsf{w}_2$ w.r.t.~that method and the moved attribute.
We will see later on that this application condition is also included in the set of impairment-indicating application conditions for \textsf{moveAttribute} and $\textsf{w}_2$, which implies that we can discard it.

The set $\repair(\textsf{moveMethod},\textsf{w}_1)$ contains one application condition. 
This application condition checks whether the moved method shares a common attribute with every other method that is contained in the same class.
If not,  a repair is found, since $\textsf{w}_1$ states that methods contained in the same class must share a common attribute.

The set $\repair(\textsf{moveMethod},\textsf{w}_2)$ consists of two application conditions. 
Similar to the application conditions contained in the set $\repair(\textsf{moveAttribute},\textsf{w}_2)$, the left application condition checks whether the method has been moved to a class containing an attribute on which it depends. 
If so,  a repair of $\textsf{w}_2$ will be found. 
The application condition on the right checks whether the moved method depends on an attribute that is not contained in either the method's original class nor its new class.
An occurrence of this pattern within a transformation will not affect the consistency of $\textsf{w}_2$ w.r.t.~this attribute and the moved method. 
We will see later that this application condition is also included in the set of impairment-indicating application conditions for \textsf{moveMethod} and $\textsf{w}_2$, which implies that we can discard it.

\section{Preliminaries}
\label{chapter:preliminaries}
In this section, we recall key notions that are used throughout this paper. 
Our theory for the construction of graph repair algorithms is based on typed graphs, as introduced for graph transformations in~\cite{Ehrig78, EEPT06}. 

\subsubsection*{Typed graph and graph morphism}
A graph consists of nodes and edges, where edges connect nodes. 
In our running example, classes, methods, and attributes are represented as nodes.
Object references, such as the attribute dependencies between a method and an attribute, are represented as edges. 
\emph{Graph morphisms} are mappings between graphs that preserve the structural properties of their domain graph.
In this paper, we will only consider injective graph morphisms.
Intuitively, the existence of an injective graph morphism from a graph $A$ to a graph $B$ means that $A$ is a subgraph of $B$.

\begin{defi}[Graph and graph morphism]
	A \emph{graph} $G = (G_V, G_E, s_G, t_G)$ consists of a set $G_V$ of nodes, a set $G_E$ of edges and two mappings $s_G \colon G_E \to G_V$ and $t_G \colon G_E \to G_V$ that assign the source and target nodes for each edge of $G$.
	If a tuple as above is not given explicitly, the set of nodes (edges) is denoted by $G_V$ ($G_E$) and the source (target) mapping is denoted by $s_G$ ($t_G$).
	
	A \emph{graph morphism} $f \colon G \to H$ between two graphs $G$ and $H$ consists of two mappings $f_V \colon G_V \to H_V$ and $f_E \colon G_E \to H_E$ that preserve the source and target mappings, i.e., $f_V \circ s_G = s_H \circ f_E$ and $f_C \circ t_G = t_H \circ f_E$. 
	A graph morphism is called \emph{injective} (\emph{surjective}) if both mappings $f_V$ and $f_E$ are injective (surjective).
	An injective morphism $f$ from $G$ to $H$ is denoted by $f \colon G \inj H$.
	If $f \colon G \inj H$ is injective, we call $f$ an \emph{occurrence of $G$ in $H$}.
	A graph morphism $f \colon G \to H$ is called an \emph{isomorphism} if there is a graph morphism $f^{-1} \colon H \to G$ such that $\id_G = f^{-1} \circ f$ and $\id_H = f \circ f^{-1}$.	
\end{defi}
Note that, for each graph $G$, there is a unique morphism $f \colon \emptyset \to G$ originating in the empty graph $\emptyset$.

\begin{figure}
	\includegraphics[scale=0.8]{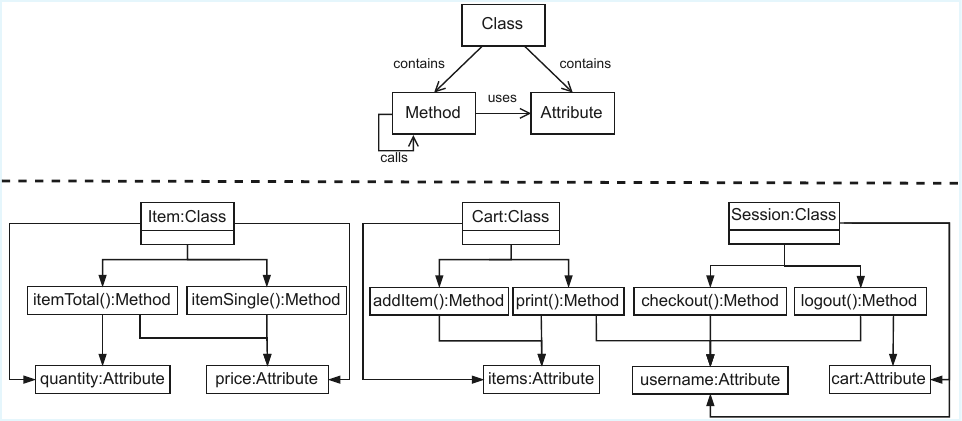}
	\caption{Type graph (top) and typed graph (bottom)}
	\label{fig:typed_graph}
\end{figure}

\begin{exa}
	For the graph at the bottom in \cref{fig:typed_graph}, the set of nodes includes \textsf{Item}, \textsf{Cart}, and \textsf{Session}, as well as any nodes representing features.
	The set of edges contains the arrows in between them. The source and target nodes of each edge are determined by the source and target of the arrow. 
	
	If we consider the graph that contains a class \textsf{c'} that is connected to the methods $\textsf{m'}_1$ and $\textsf{m'}_2$ as shown in the constraint $\textsf{w}_1$ in \cref{fig:example:rule+constraints}, there are several injective graph morphisms from this graph into the graph at the bottom of \cref{fig:typed_graph}.
	For example, the node \textsf{c'} could be mapped to \textsf{Item}, $\textsf{m'}_1$ to \textsf{itemTotal()}, and $\textsf{m'}_2$ to \textsf{itemSingle()}.
	The edge from \textsf{c'} to $\textsf{m'}_1$ is mapped to the edge starting in \textsf{Item} and ending in \textsf{itemTotal()}, and the edge from \textsf{c'} to $\textsf{m'}_2$ is mapped to the edge starting in \textsf{Item} and ending in \textsf{itemTotal()}.
\end{exa}

Throughout the paper, we assume that a graph is always finite, i.e., the set of nodes and the set of edges are finite.

Typing provides a way to assign meaning to graph elements.
For example, we can assign each node and edge of the graph shown at the top of \cref{fig:typed_graph} their respective roles in a class diagram.
I.e., the nodes \textsf{Item}, \textsf{Cart}, and \textsf{Session} are typed as classes, the nodes on the second layer \textsf{itemTotal()}, \textsf{itemSingle()}, \ldots are typed as methods, and the nodes \textsf{quantity}, \textsf{price}, \textsf{item}, \textsf{username}, and \textsf{cart} are typed as attributes.
An edge from a class to a method signals that the method is contained in the class, and an edge from a method to an attribute signals that the method uses the attribute.
This can be formalised by \emph{typed graphs} and \emph{typed graph morphisms}.

\begin{defi}[Typed graph, typed graph morphism]
	Given a graph $TG$, called the \emph{type graph}, a \emph{typed graph} over $TG$ is a tuple $(G, type)$ consisting of a graph $G$ and a graph morphism $type \colon G \to TG$.
	Given two typed graphs $G = (G', type_G)$ and $H = (H', type_H)$, a \emph{typed graph morphism} $f \colon G \to H$ is a graph morphism $f \colon G' \to H'$ such that $type_H \circ f = type_G$.
	The typed graph morphism $f \colon G \to H$ is \emph{injective (surjective)} if $f \colon G' \to H'$ is injective (surjective).
	It is called an \emph{isomorphism} if there is a graph morphism $f^{-1} \colon H \to G$ such that $\id_G = f^{-1} \circ f$ and $\id_H = f \circ f^{-1}$.
\end{defi}

\begin{exa}
	The graph $G$ shown at the bottom of \cref{fig:typed_graph} is typed over the graph $TG$ shown at the top of \cref{fig:typed_graph}.
	The typed graph morphism maps the nodes \textsf{Item}, \textsf{Cart}, and \textsf{Session} to \textsf{Class}.
	The nodes \textsf{itemTotal()}, \textsf{itemSingle()}, \textsf{addItem()}, \textsf{print()}, \textsf{checkout()}, and \textsf{logout()} are mapped to \textsf{Method} and \textsf{quantity}, \textsf{price}, \textsf{items}, \textsf{username}, and \textsf{cart} are mapped to \textsf{Attribute}.
	The edges are mapped accordingly.
\end{exa}

\emph{Throughout the paper, we assume that all graphs are typed graphs.
	Additionally, all graphs used in our examples and evaluation are typed over the graph shown at the top of \cref{fig:typed_graph}. 
	For a more compact representation of graphs, we will not specify edge types in the figures for the rest of the paper if they can be identified unambiguously.
}

\subsubsection*{Nested graph constraints}
To verify that a graph satisfies a desired structural property, a concept of formulating these properties is needed.
An example property for class diagrams is that a method should not be contained in two classes.
In particular, we want to derive graphs that satisfy these properties or are more consistent with respect to these properties. 
For example, considering the property described above, a class diagram $H$ is more consistent with respect to this property than another class diagram $G$ if $H$ contains fewer methods that are contained in two classes.
To formulate such properties, we use \emph{nested graph conditions} introduced by Habel and Pennemann~\cite{HabelP09}.
Rensink has shown that the class of nested graph conditions is equivalent to first-order logic~\cite{Rensink04}.
In addition, almost all \ac{OCL} formulas can be translated into nested graph constraints~\cite{RadkeABHT18}.
Nested graph constraints, or constraints for short, are nested graph conditions that can be evaluated directly on a given graph, whereas graph conditions are generally evaluated with respect to graph morphisms.

\begin{defi}[Nested graph conditions]
	A \emph{nested graph condition} over a graph $P$ is of the form
	\begin{itemize}
		\item $\true$, or
		\item $\exists(e \colon P \inj Q, d)$ where $d$ is a condition over $Q$, or
		\item $d_1 \vee d_2$ or $\neg d_1$ where $d_1$ and $d_2$ are conditions over $P$.
	\end{itemize}
	A condition over the empty graph $\emptyset$ is called \emph{constraint}. We use the abbreviations $\false := \neg \true$, $d_1 \wedge d_2 := \neg(\neg d_1 \vee \neg d_2)$, $d_1 \implies d_2 := \neg d_1 \vee d_2$, and $\forall(e \colon P \inj Q,d) := \neg \exists(e \colon P \inj Q, \neg d).$
	When $e$ is of the form $e \colon \emptyset \inj Q$, we use the short notations $\exists(Q,d)$ and $\forall(Q,d)$. 
	For a condition $c = \forall(Q,d)$, we call $Q$ the \emph{premise} and $d$ the \emph{conclusion} of $c$. 
\end{defi}

Throughout the paper, we assume that each condition is finite, i.e., the graphs used and the number of nesting levels are finite. 

\begin{exa}
	\Cref{fig:example:rule+constraints} shows the constraints $\textsf{h}_1, \textsf{h}_2, \textsf{w}_1$, and $\textsf{w}_2$. The constraints $\textsf{h}_1, \textsf{h}_2$, and $\textsf{w}_2$ specify forbidden patterns where
	\begin{itemize}
		\item $\textsf{h}_1$ states that a method cannot be contained in two different classes,
		\item $\textsf{h}_2$ states that an attribute cannot be contained in two different classes, and
		\item $\textsf{w}_2$ states that if a method uses an attribute, they cannot be contained in different classes.
	\end{itemize}
	The constraint $\textsf{w}_1$ states that each pair of methods contained in the same class must use at least one common attribute that is also contained in that class.
	The graph morphism from the first graph of $\textsf{w}_1$ to the second graph is implicitly given by mapping nodes with the same identifiers and the edges are mapped accordingly.
\end{exa}

\begin{defi}[Semantics of condition]
	Given a graph morphism $p \colon P \inj G$ and a condition $c$ over $P$, then $p$ satisfies $c$, denoted by $p \models d$, if 
	\begin{itemize}
		\item $c = \true$, or
		\item $c = \exists (e \colon P \inj P, d)$ and there is a morphism $q \colon Q \inj G$ with $p = q \circ e$ and $q \models d$, or
		\item $c = d_1 \vee d_2$ and $p \models d_1$ or $p \models d_2$, or
		\item $c = \neg d$ and $p \not \models d$. 
	\end{itemize}
	A graph $G$ satisfies a constraint $c$, denoted by $G \models c$ if the unique morphism $p \colon \emptyset \inj G$ satisfies $c$. 
	Two conditions $c_1$ and $c_2$ over a graph $P$ are equivalent, denoted by $c_1 \equiv c_2$, if for each morphism $p \colon P \inj G$ we have $p \models c_1 \iff p \models c_2$. 
\end{defi}

\begin{exa}\label{ex:constraint_satisfaction}
	The graph shown in \cref{fig:example:example} satisfies the hard constraints $\textsf{h}_1$ and $\textsf{h}_2$.
	There is no attribute and no method in two classes.
	However, the graph does not satisfy the weak constraint $\textsf{w}_1$ shown in \cref{fig:example:rule+constraints}.
	The methods \textsf{addItem()} and \textsf{checkout()} are contained in the same class \textsf{Cart}, but \textsf{Cart} does not contain an attribute used by both methods. 
	The graph also does not satisfy the weak constraint $\textsf{w}_2$.
	This is because \textsf{checkout()} uses the attribute \textsf{username}, which is contained in the class \textsf{Session}. 
\end{exa}

\subsubsection*{Graph rules and transformations}
Graph transformation rules are used to specify state-changing operations for a system of interest.
They can be used to specify graph elements that are to be deleted or created when the rule is applied to a graph.
In the case of our \ac{CRA} example, rules are used to define the set of all available refactoring operations. 
Their associated application conditions are unsatisfied whenever a rule application changes the consistency state of a rewritten graph.
In the following, we recall the formal definitions of \emph{graph transformation rules} and \emph{graph transformations}~\cite{EEPT06,Ehrig78}.

\begin{defi}[Graph transformation rule and application condition]
	A \emph{graph transformation rule}, or $\emph{rule}$ for short, $\rho = \completeRle$ consists of graph $L$, called the \emph{\ac{LHS}}, $K$, called the $\emph{context}$, $R$, called the \emph{\ac{RHS}}, and injective morphisms $l \colon K \inj L$ and $r \colon K \inj R$.
	The \emph{inverse rule} of $\rho$, denoted by $\rho^{-1}$ is defined as $\rho^{-1} := \rle{R}{r}{K}{l}{L}$.
	An \emph{application condition} for a rule is a nested condition over its \ac{LHS}. 
\end{defi}

\begin{exa}\label{ex:rule}
	\Cref{fig:example:rule+constraints} shows the rules \textsf{moveAttribute} and \textsf{moveMethod} that move an attribute or a method from one class to another. 
	Elements to be deleted are coloured red, elements to be preserved are coloured black, and elements to be created are coloured green. 
	The \ac{LHS} of the rules contains the black and red elements, the context contains the black elements, and the \ac{RHS} contains the black and green elements.
	When the \textsf{moveAttribute} rule is applied to a graph, the red edge is deleted and the green edge is created.
	
	The condition contained in $\repair(\textsf{moveAttribute},  \textsf{w}_1)$ (\cref{fig:resulting_application_conditions}) is an application condition for the rule \textsf{moveAttribute}.
	The embedding of the \ac{LHS} in the first graph of the application condition is implicitly given by the node identifiers, i.e., $\textsf{c}_1$ is mapped to $\textsf{c}_1$, $\textsf{c}_2$ is mapped to $\textsf{c}_2/\textsf{c}'$, and $\textsf{a}$ is mapped to $\textsf{a}$.
\end{exa}

For simplicity, we present the more constructive, set-theoretic definition of graph transformation, which has been shown to be equivalent to the commonly used double-pushout approach, based on category theory~\cite{EEPT06}.
A brief introduction to the double-pushout approach is given in Section~\ref{app:dpo}, as we need some of its properties for our proofs.
Intuitively, a graph transformation consists of two steps: First, elements that are exclusively contained in the \ac{LHS} of the rule are deleted.
The second step is to create elements exclusively contained in the \ac{RHS} of the rule.
Note that a rule is only applicable to a graph if deleting a node does not lead to a dangling edge, i.e., every edge in the resulting graph has a source and a target node.

\begin{figure}
	\centering
	\begin{tikzpicture}[scale = 1]
		\node(A) at (0,0) {$L$};
		\node(B) at (2,0) {$K$};
		\node(C) at (0,-2) {$G$};
		\node(D) at (2,-2) {$D$};		
		\node(E) at (4,0) {$R$};
		\node(F) at (4,-2) {$H$};
		\node(label1) at (1,-1) {$(1)$};
		\node(label2) at (3,-1) {$(2)$};
		
		\draw [left hook-stealth] (B) edge node  [above]{$l$} (A);
		\draw [left hook-stealth] (A) edge node  [left]{$m$} (C);
		\draw [left hook-stealth] (D) edge node  [above]{$g$} (C);
		\draw [left hook-stealth] (B) edge node  {} (D);
		\draw [right hook-stealth] (B) edge node  [above]{$r$} (E);
		\draw [right hook-stealth] (D) edge node  [above]{$h$} (F);
		\draw [left hook-stealth] (E) edge node  [right]{$n$} (F);
	\end{tikzpicture}
	\caption{Graph transformation}
	\label{fig:graph_transformation1}
\end{figure}

\begin{defi}[Graph transformation and derived rule]
	Given a graph $G$, a rule $\rho = \completeRle$, and an injective morphism $m \colon L \inj G$ (see \cref{fig:graph_transformation1}), a \emph{graph transformation} $t$, denoted by $t \colon G \Longrightarrow_{\rho, m} H$, via $\rho$ at $m$ can be constructed by (a) deleting all nodes and edges of $L$ that do not have a preimage in $K$, i.e., construct the graph $D = G \setminus m(L \setminus(l(K))$ and (b) adding all nodes and edges of $R$ that do not have a preimage in $K$, i.e., construct the graph $H = D \dot{\cup} R \setminus r(K)$, where $\dot{\cup}$ denotes the disjoint union.	
	
	The rule $\rho$ is applicable at $m$ if and only if $D$ is a graph, i.e., if it does not contain any dangling edges.
	In this case, $m$ is called \emph{match} and the newly created morphism $n \colon R \inj H$ is called \emph{comatch}.
	We call $G$ the \emph{original graph}, $D$ the \emph{interface} and $H$ the \emph{result graph} of the transformation $t$.
	The \emph{derived rule} of a transformation  $t \colon G \Longrightarrow_{\rho, m} H$, denoted by $\der(t)$, is defined as $\der(t) =  \rle{G}{g}{D}{h}{H}$.
\end{defi}

\begin{exa}\label{exa:transformation}
	For the rule \textsf{moveMethod}, there is a match $m$ from the \ac{LHS} in the graph $G$ shown in \cref{fig:example:example}.
	We can map $\textsf{c}_1$ to \textsf{Cart}, $\textsf{c}_2$ to \textsf{Session}, and \textsf{m} to \textsf{checkout()}.
	When \textsf{moveMethod} is applied to $m$, we get the graph $H$ shown at the bottom of \cref{fig:typed_graph}.
	In the resulting graph, \textsf{checkout()} has been moved from \textsf{Cart} to \textsf{Session}, i.e., first the edge from \textsf{Cart} to \textsf{checkout()} is deleted and then a new edge from \textsf{Session} to \textsf{checkout()} is created.
\end{exa}

In addition, we need a method for tracking graph elements throughout a transformation, i.e., we want to know which elements in the original graph of the transformation correspond to which elements in the resulting graph of the transformations, and which elements have been deleted or created.
In particular, we want to decide whether occurrences of graphs in other graphs are deleted or created by a transformation.

\begin{defi}[Track morphism~\cite{Plump05}]
	The \emph{track morphism} of a transformation $t\colon G\Longrightarrow H$ (see \cref{fig:graph_transformation1}), denoted by $\track_t \colon G \dasharrow H$, is a partial morphism, defined as
	$$\track_t(e) := \begin{cases}
		h(g^{-1}(e)) \quad \text{if } e \in g(D) \\
		\text{undefined} \quad \text{otherwise}.
	\end{cases}$$
\end{defi}

\begin{exa}
	If we consider the transformation $t \colon G \Longrightarrow H$ described in \cref{exa:transformation}, we have $\track(e) = e$ for each element that is preserved (i.e., for every element except the edge from \textsf{Cart} to \textsf{checkout()}). 
	Also, there is no element $e \in G$, so that the track morphism maps $e$ to the newly created edge from \textsf{Session} to \textsf{checkout()} in $H$.
\end{exa}

\section{Counting Constraint Violations}
\label{sec:counting_constraint_violation}
\label{chapter:constraints}
When large graphs are transformed using small rules, such as \textsf{MoveMethod} or \textsf{MoveAttribute}, it is unlikely that the consistency of a constraint will be fully restored by a single transformation.
However, consistency can be increased in the sense that the original graph of the transformation contains more violations than the resulting graph.

For example, if we consider the transformation described in \cref{exa:transformation}, where the method \textsf{checkout} is moved from \textsf{Card} to \textsf{Session} (the original graph is shown in \cref{fig:example:example}; the resulting graph is shown in \cref{fig:typed_graph}) and the weak constraint $\textsf{w}_2$ (shown in \cref{fig:example:rule+constraints}), the consistency  increases even though the resulting graph still does not satisfy $\textsf{w}_2$.
The original graph contains two methods \textsf{print()} and \textsf{checkout()} that use an attribute that is not in the same class, whereas the resulting graph contains only one such method (which is \textsf{print()}).
Therefore,  the resulting graph is more consistent w.r.t.~$\textsf{w}_2$ than the original graph of this transformation.

To rank rule applications in terms of their potential to increase the consistency of a constraint, we need a method to evaluate this increase or decrease in the consistency of weak constraints.
Here, we rely on the \emph{number of violations} introduced by Kosiol et al. in~\cite{KosiolSTZ22}.
This notion enables an increase in consistency to be detected even if it is not fully recovered.
For example, it can detect whether an occurrence of the premise that does not satisfy the conclusion has been deleted or extended so that it satisfies the conclusion in the result graph of the transformation.
Although this notion was originally introduced for constraints in so-called \emph{alternating quantifier normal form}, i.e., the set of all nested constraints that do not use conjunctions or disjunctions, we have extended it to support universally bounded nested conditions that may use these operators in the conclusion.

\emph{We allow hard constraints and weak constraints to be arbitrary nested graph conditions.}

Although our theory is developed for universal weak constraints (i.e., weak constraints of the form $\forall(P,d)$), this does not limit the range of supported constraints:
\emph{Every condition $c$ over a graph $P$ is equivalent to the condition $\forall(\id_P \colon P \inj P, c)$. }
This means that any nested graph condition can be moulded into the required form.
For example, every existential constraint $c = \exists(e \colon \emptyset \inj P,d)$ is equivalent to the condition $\forall(id_{\emptyset} \colon \emptyset \inj \emptyset, c)$.
In this case, the number of violations acts as a binary property, i.e., it is equal to $1$ if the constraint is unsatisfied and equal to $0$ if the constraint is satisfied.
This is because there is only one morphism $p \colon \emptyset \inj G$ into any graph $G$ that either satisfies or does not satisfy the constraint $c$.
This effect complies with the semantics of existential graph conditions, since only one repair is sufficient for the resulting graph of a transformation to satisfy the constraint.
Therefore,  inserting occurrences of $P$ that do not satisfy $d$ does not decrease the consistency of a graph, since only one of these occurrences needs to be repaired, or a new occurrence of $P$ that satisfies $d$ needs to be inserted.

\begin{figure}
	\centering
	\begin{tikzpicture}[scale = 1]
		\node(A) at (0,0) {$C$};
		\node(B) at (2,0) {$P$};
		\node(C) at (1,-2) {$G$};

		\draw [right hook-stealth] (A) edge node  [above]{$e$} (B);
		\draw [left hook-stealth] (A) edge node  [left]{$p$} (C);
		\draw [left hook-stealth] (B) edge node  [right]{$q$} (C);
	\end{tikzpicture}
	\caption{Morphisms used to evaluate the \emph{set of violations}}
	\label{fig:set_of_vios}
\end{figure}

\begin{defi}[Set and number of violations]\label{def:counting_method}
	Given a condition $c = \forall(e \colon C \inj P,d)$ and a graph morphism $p \colon C \inj G$ (\cref{fig:set_of_vios}), the \emph{set of violations of $c$ in $p$}, denoted by $\sv_{p}(c)$, is defined as
	$$\sv_{p}(c) := \{q \colon P \inj G \mid p = q \circ e \text{ and } q \not \models d\}.$$
	The set of violations of a constraint $c = \forall(e \colon \emptyset \inj P,d)$ in a graph $G$, denoted by $\sv_G(c)$, is defined as $\sv_G(c) := \sv_{p}(c)$ where $p \colon \emptyset \inj G$ is the unique morphism into $G$.

	The \emph{number of violations of a condition $c$ in $p$}, denoted by $\nv_{p}(c)$, is defined as $\nv_{p}(c) := |\sv_{p}(c)|.$
	The \emph{number of violations of a constraint $c$ in a graph $G$}, denoted with $\nv_G(c)$, is defined as $\nv_G(c) := |\sv_G(c)|.$
\end{defi}

\begin{exa}\label{ex:num_violations}
	As discussed in \cref{ex:constraint_satisfaction}, the graph $G$ in \cref{fig:example:example} does not satisfy $\textsf{w}_1$ (which states that two methods contained within the same class must use a common attribute which is also contained in the same class) and $\textsf{w}_2$ (which states that a method and its used attributes cannot be contained in different classes).
	For $\textsf{w}_1$, there are two pairs of methods that do not use a common attribute in the same class (\textsf{addItem()} and \textsf{checkout()}; \textsf{print()} and \textsf{checkout()}).
	Each of these pairs, accompanied with the class they are contained in, is detected twice by the premise of the constraint. For \textsf{addItem()} and \textsf{checkout()}, we either map $\textsf{m}_1'$ to \textsf{addItem()} and $\textsf{m}_2'$ to \textsf{checkout()} or vice versa.
	So the total number of violations of $\textsf{w}_1$ in $G$ is $4$.

	When considering $\textsf{w}_2$, the methods \textsf{print()} and \textsf{checkout()} contained in the class \textsf{Cart} use the attribute \textsf{username}, which is not contained in \textsf{Cart}.
	Therefore, the total number of violations of $\textsf{w}_2$ in $G$ is $2$.
\end{exa}

\section{Application Conditions for Consistency Monitoring}
\label{chapter:application_conditions}
\label{chapter:appl_conds}
The main idea of our approach is to predict the change in consistency induced by the application of a rule.
This allows us to apply the most consistency-increasing rule without using a trial-and-error approach.
To obtain this prediction, we use application conditions that do not block the application of constraint-violating rules, but instead annotate rule matches with the number of constraint impairments and/or repairs caused by the associated transformation; thus, the conditions monitor the consistency change.

This section begins by recalling some preliminaries from the theory of graph transformation that are necessary for our new theory.
To derive application conditions from a constraint, the first main step is to construct all overlaps of a rule and a constraint.
The set of application conditions should be concise, so we identify equivalence classes of overlaps and work with representatives of these classes.
This allows us to simplify the construction and complexity of application conditions compared to the construction presented in~\cite{fritsche2024using}.
We use the well-known techniques introduced by Habel and Pennemann to construct so-called \emph{consistency-preserving} and \emph{consistency-guaranteeing} application conditions~\cite{HabelP09} to construct the application conditions.
To detect all situations in which a violation of a constraint $c= \forall(P,d)$ is repaired by applying a rule $\rho$, we need to consider each overlap of the \ac{LHS} of the rule $\rho$ and the premise $P$ of the constraint. 
To find repairs of the premise $P$, we consider overlaps where the occurrence of $P$ is destroyed by the transformation.
To find repairs of the conclusion $d$, we consider overlaps where the occurrence of $P$ is preserved by the transformation.
If this occurrence of $P$ satisfies $d$ in the resulting graph of the transformation but does not satisfy $d$ in the original graph of the transformation, then this occurrence is a repair of $d$.

For the construction of impairment-indicating application conditions, we use the observation that a rule introduces a violation of a constraint $c$ if and only if the inverse rule repairs a violation of $c$.
Thus, repair-indicating application conditions for a rule $\rho$ are impairment-indicating application conditions for its inverse rule $\rho^{-1}$.
Our main theorem shows that {\em the number of additional constraint violations and repairs caused by a rule application can be characterized by the difference between the numbers of violations of the associated impairment-indicating and repair-indicating application conditions}.
We call this difference the \emph{gain in consistency}.
Constraints can be assigned weights to prioritise certain ones over others, if desired.
Finally, transformations with repair-indicating and impairment-indicating application conditions are compared with consistency-sustaining and consistency-improving transformations presented in~\cite{KosiolSTZ22}.
We found out that a transformation is consistency-sustaining if the gain in consistency is not negative,
it is consistency-improving if the gain is positive.
The proofs of the results presented in this section are included in Section~\ref{app:proofs}.

\subsection{Preliminaries}\label{sec:preliminaries}

In the following, we will briefly introduce the preliminaries that are needed for our construction of application conditions.

An \emph{overlap} $ o = (i_G, i_H, GH)$ of graphs $G$ and $H$ consists of jointly surjective morphisms\footnote{Two morphisms $p \colon P \inj G$ and $q \colon Q \inj G$ into the same graph $G$ are called \emph{jointly surjective} if for each element $e \in G$ either there is an element $e' \in P$ with $p(e')=e$ or there is an element $e' \in Q$ with $q(e') = e$.} $i_G \colon G \inj GH$ and $i_H \colon H \inj GH$, called the \emph{overlap morphisms}, and the common codomain $GH$, called the \emph{overlap graph}.
The set of all overlaps of $G$ and $H$ is denoted by $\Over(G,H)$.

The \emph{shift along morphism} operator allows shifting a condition $c$ over a graph $P$ along a morphism $i_P \colon P \inj PL$ to extend the graph $P$ by further graph elements.
This construction results in a condition over $PL$ that is equivalent in the sense that each morphism $p' \colon PL \inj G$ into some graph $G$ satisfies the shifted condition if and only if $p' \circ i_P$ satisfies $c$.
Given a condition $c$ over a graph $P$ and a morphism $i_P \colon P \inj PL$, the \emph{shift along $i_P$}~\cite{HabelP09}, denoted by $\shift(i_P,c)$, is defined as follows: 

\begin{minipage}{.7 \textwidth}
	\begin{itemize}
		\item if $c = \true$, $\shift(i_P,c) = \true$, and
		\item if $c = \exists(e \colon P \inj Q, d)$, $\shift(i_P,c) = \bigvee_{(e', i_Q)} \exists(e' \colon PL \inj QL, \allowbreak \shift(i_Q,d))$, where $(e', i_Q, QL)$ is an overlap of $PL$ and $Q$ such that $e' \circ i_P = i_Q \circ e$, i.e., the square on the right is commutative, and
		\item if $c = c_1 \vee c_2$, $\shift(i_P,c) = \shift(i_P,c_1) \vee \shift(i_P,c_2) $, and
		\item if $c = \neg c_1$, $\shift(i_P,c) = \neg \shift(i_P, c_1)$.
	\end{itemize}
\end{minipage}
\hspace{.03\textwidth}
\begin{minipage}{.25 \textwidth}
	\begin{tikzpicture}
		\node (P2) at (0,0) {$QL$};
		\node (P1) at (2,0) {$Q$};
		\node (C2) at (0,2) {$PL$};
		\node (C1) at (2,2) {$P$};

		\draw[left hook-stealth] (P1) edge node [above] {$i_Q$} (P2);
		\draw[right hook-stealth] (C1) edge node [right] {$e$} (P1);
		\draw[left hook-stealth] (C1) edge node [above] {$i_P$} (C2);
		\draw[right hook-stealth] (C2) edge node [left] {$e'$} (P2);
	\end{tikzpicture}
\end{minipage}

Note that, when shifting a constraint $\forall(P, d)$ over a morphism $p \colon \emptyset \inj P'$, the first step is to construct all overlaps $\Over(P,P')$.

The \emph{shift over rule} operator transforms a right-application condition into an equivalent left-application condition and vice versa.
Given a rule $\rho = \rle{L}{l}{K}{r}{R}$ and a condition $c$ over $R$, the \emph{shift of $c$ over $\rho$}~\cite{HabelP09}, denoted by $\leftshift(\rho,c)$ is defined as follows:
\begin{itemize}
	\item If $c = \true$, $\leftshift(\rho, c) = \true$, and
	\item if $c = \exists(e \colon R \inj P, d)$, $\leftshift(\rho,c) := \exists(n \colon L \inj P', \leftshift(\der(t)^{-1}, d))$, where $P'$ is the resulting graph, $n$ the comatch and $\der(t)$ the derived rule of the transformation $t \colon P \Longrightarrow_{\rho^{-1},e} P'$. If there is no such transformation, we set $\leftshift(\rho, c) = \false$.
	\item If $c = c_1 \vee c_2$, $\leftshift(\rho,c) := \leftshift(\rho,c_1) \vee \leftshift(\rho,c_1)$.
	\item If $c = \neg c_1$, $\leftshift(\rho,c) := \neg \leftshift(\rho,c_1)$.
\end{itemize}

\begin{figure}
	\centering
	\begin{tikzpicture}[scale = 1]
		\node (G) at (0,0){$G$};
		\node (D1) at (-3,0){$D_1$};
		\node (D2) at (3,0){$D_2$};
		\node (H1) at (-5,0){$H_1$};
		\node (H2) at (5,0){$H_2$};

		\draw[right hook-stealth](D1) edge node [above] {$g_1$}(G);
		\draw[left hook-stealth](D1) edge node [above] {}(H1);
		\draw[left hook-stealth](D2) edge node [above] {$g_2$}(G);
		\draw[right hook-stealth](D2) edge node [above] {}(H2);

		\node (K1) at (-3,2) {$K_1$};
		\node (L1) at (-1,2) {$L_1$};
		\node (R1) at (-5,2) {$R_1$};

		\draw[right hook-stealth](K1) edge node [above] {$l_1$}(L1);
		\draw[left hook-stealth](K1) edge node [above] {$r_1$}(R1);

		\node (K2) at (3,2) {$K_2$};
		\node (L2) at (1,2) {$L_2$};
		\node (R2) at (5,2) {$R_2$};

		\draw[left hook-stealth](K2) edge node [above] {$l_2$}(L2);
		\draw[right hook-stealth](K2) edge node [above] {$r_2$}(R2);

		\draw[left hook-stealth, bend right=10] (L2) edge node [label={[xshift=-3em, yshift=-1.7em]$x_2$}]{}(D1);
		\draw[right hook-stealth, bend left = 10] (L1) edge node [label={[xshift=3em, yshift=-1.7em]$x_1$}]{}(D2);

		\draw[left hook-stealth] (R1) edge node [right]{}(H1);
		\draw[left hook-stealth] (R2) edge node [right]{}(H2);
		\draw[left hook-stealth] (K1) edge node [right]{}(D1);
		\draw[left hook-stealth] (K2) edge node [right]{}(D2);

		\draw[right hook-stealth] (L1) edge node [fill=white]{$m_1$}(G);
		\draw[left hook-stealth] (L2) edge node [fill=white]{$m_2$}(G);
	\end{tikzpicture}
	\caption{Parallel independent transformations}
	\label{fig:parallel_independence}
\end{figure}

\begin{defi}[Parallel independent graph transformations]\label{def:parallel_dependence}
	Two graph transformations $t_1 \colon G \Longrightarrow_{m_1, \rho_1} H_1$ and $t_2 \colon G \Longrightarrow_{m_2, \rho_2} H_2$ via rules $\rho_1 = \rle{L_1}{l_1}{K_1}{r_1}{R_1}$ and $\rho_2 = \rle{L_2}{l_2}{K_2}{r_2}{R_2}$ are \emph{parallel independent} if there are morphisms $x_1 \colon L_1 \inj D_2$ and $x_2 \colon L_2 \inj D_1$ such that $g_2 \circ x_1 = m_2$ and $g_1 \circ x_2 = m_1$.
	The transformations and morphisms are shown in \cref{fig:parallel_independence}.
	If a pair of transformations is not parallel independent, it is called \emph{parallel dependent}.
\end{defi}

\begin{figure}
	\centering
	\begin{tikzpicture}[scale = 1]
		\node (G) at (0,0){$PL$};
		\node (D1) at (-3,0){$D$};
		\node (D2) at (3,0){$PL$};
		\node (H1) at (-5,0){$PR$};
		\node (H2) at (5,0){$PL$};
		\node (GT) at (0, -2){$G$};
		\node (DT) at (-3, -2){$D$};
		\node (HT) at (-5, -2){$H$};

		\draw[right hook-stealth](D1) edge node [above] {$g$}(G);
		\draw[right hook-stealth](DT) edge node [above] {}(GT);
		\draw[left hook-stealth](D1) edge node [above] {$h$}(H1);
		\draw[left hook-stealth](DT) edge node [above] {}(HT);
		\draw[left hook-stealth](G) edge node [right] {$p$}(GT);
		\draw[left hook-stealth](D1) edge node [right] {}(DT);
		\draw[left hook-stealth](H1) edge node [right] {}(HT);
		\draw[left hook-stealth](D2) edge node [above] {}(G);
		\draw[right hook-stealth](D2) edge node [above] {}(H2);

		\node (K1) at (-3,2) {$K$};
		\node (L1) at (-1,2) {$L$};
		\node (R1) at (-5,2) {$R$};

		\draw[right hook-stealth](K1) edge node [above] {$l$}(L1);
		\draw[left hook-stealth](K1) edge node [above] {$r$}(R1);

		\node (K2) at (3,2) {$P$};
		\node (L2) at (1,2) {$P$};
		\node (R2) at (5,2) {$P$};

		\draw[left hook-stealth](K2) edge node [above] {$\id_P$}(L2);
		\draw[right hook-stealth](K2) edge node [above] {$\id_P$}(R2);

		\draw[left hook-stealth, bend right=10] (L2) edge node [label={[xshift=-3em, yshift=-1.7em]$x$}]{}(D1);
		\draw[left hook-stealth, bend right=15] (L1) edge node [label={[xshift=0.0em, yshift=-3.3em]$m$}]{}(GT);
		\draw[dashed, -stealth, bend left](GT) edge node [below]{$\track_t$}(HT);

		\draw[left hook-stealth] (R1) edge node [right]{}(H1);
		\draw[left hook-stealth] (R2) edge node [right]{}(H2);
		\draw[left hook-stealth] (K1) edge node [right]{}(D1);
		\draw[left hook-stealth] (K2) edge node [right]{}(D2);

		\draw[right hook-stealth] (L1) edge node [fill=white]{$i_L$}(G);
		\draw[left hook-stealth] (L2) edge node [fill=white]{$i_P$}(G);
	\end{tikzpicture}
	\caption{Induced transformations of an overlap $(i_L, i_P, PL)$ of a rule $\rho = \completeRle$ and a graph $P$ and application to a graph $G$}
	\label{fig:equivalence_classes}
\end{figure}

\subsection{Equivalence of Overlaps}\label{sec_equiv_overlaps}

To ensure that each occurrence of the premise of a constraint is detected by the application conditions, so that each violation or repair of the constraint during the transformation can be detected, we need to consider each overlap of the \acs{LHS} of the rule $\rho$ that is to be applied and the premise of the constraint that is to be considered, called \emph{overlaps of rule and graph}.
Additionally, we require that the overlap has \emph{induced transformations}, i.e., that the embedding $i_L \colon L \inj PL$ of the \ac{LHS} in the overlap graph is a match for the rule and that the rule is applicable at this embedding (see \cref{fig:equivalence_classes}).
This means that applying the rule at the overlap does not create any dangling edges.
We do not need to consider overlaps without \emph{induced transformations} since they cannot occur in a transformation $t$ (which does not produce any dangling edges) so that the match $m$ factors\footnote{A morphism $e \colon G \to H$ factors through morphisms $f \colon G \to X$ and $g \colon X \to G$ if $e = g \circ f$.} through $i_L$ and the occurrence of $PL$ as this factorisation is needed for the number of violations (see \cref{def:counting_method}).

\begin{defi}[Overlap induced transformation, Overlaps of rule and graph]
	Given a rule $\rho := \rle{L}{l}{K}{r}{R}$ and a graph P, the \emph{set of overlaps of $\rho$ and $P$}, denoted by $\Over(\rho,P)$ is defined as
	$$\Over(\rho, P) := \{(i_L, i_P, PL) \in \Over(L,P) \mid \text{$\rho$ is applicable at $i_L$}\}.$$
	The \emph{induced transformations of an overlap} $(i_L, i_P)\in \Over(\rho,P)$ are given by $t_1 \colon PL \Longrightarrow_{\rho, i_L} PR$ and $ t_2 \colon PL \Longrightarrow_{\chk_{PL}, i_P} PL$ where $\chk_{PL} = \rle{PL}{\id_{PL}}{PL}{\id_{PL}}{PL}$ does not create or delete any elements.
\end{defi}\label{def:repaired_morphism}

Note that the induced transformation $t_2 \colon PL \Longrightarrow_{\chk_P, i_P} PL$ of an overlap $(i_L, i_P, PL)$ always exists because no elements are deleted or created.
However, we will use this transformation to determine which type of repair (impairment) is detected by an application condition constructed using this overlap: If a rule match can be extended by the overlap graph, we have found an occurrence of the premise.
If this occurrence is also a violation of the constraint, it can be repaired either by deleting it or by extending it so that it satisfies the conclusion of the constraint in the resulting graph of the transformation.

\begin{figure}
	\includegraphics[scale=0.8]{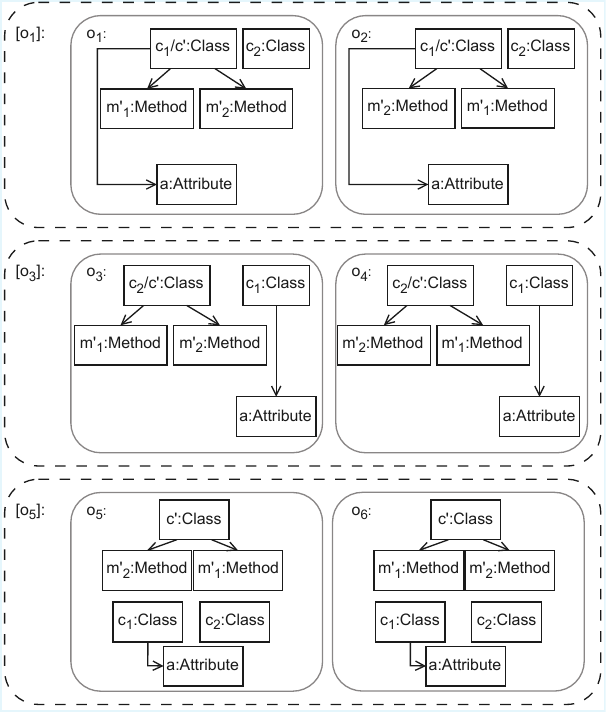}
	\caption{Overlaps $o_1, \ldots, o_6$ of the \ac{LHS} of \textsf{moveAttribute} and the premise of $\textsf{w}_1$ and their respective equivalence classes $[o_1],[o_3]$ and $[o_5]$ (marked with the dashed boxes)}\label{fig:overlaps_move_attribute}
\end{figure}

\begin{exa}
	All overlaps (and their equivalence classes) of the rule \textsf{moveAttribute} and constraint $\textsf{w}_1$ are shown in \cref{fig:overlaps_move_attribute}. The embeddings of the \ac{LHS} and the premise in the graphs are implicitly given by the node identifiers, i.e, if we consider the overlap $o_1$ and the \ac{LHS} of \textsf{moveAttribute}, the node $\textsf{c}_1$ is mapped to $\textsf{c}_1/\textsf{c}'$, $\textsf{c}_2$ is mapped to $\textsf{c}_2$, and $\textsf{a}$ is mapped to $\textsf{a}$.
	For the premise of $\textsf{w}_1$ and $o_1$, $\textsf{c}'$ is mapped to $\textsf{c}_1/\textsf{c}'$, $\textsf{m}'_1$ is mapped to $\textsf{m}'_1$, and $\textsf{m}'_2$ is mapped to $\textsf{m}'_2$.
	All of the overlaps have induced transformations. Applying the rule \textsf{moveAttribute} at an overlap will never produce a dangling edge, since \textsf{moveAttribute} does not delete any nodes.
\end{exa}

When constructing overlaps of graphs, some of them are redundant in the sense that they will detect the same occurrences of the graph they were constructed with, i.e., if for two overlaps $(i_G,i_H, GH)$ and $(i_G',i_H', GH')$ of the same graphs $G$ and $H$ we have: If there is a morphism $p \colon GH \inj X$ into some graph $X$, there is also a morphism $p' \colon GH' \inj X$, so that $p \circ i_G = p' \circ i'_G$ and $p \circ i_H = p' \circ i'_H$.
For example, the overlaps $o_1$ and $o_2$ of the \ac{LHS} of \textsf{moveAttribute} and the premise of $\textsf{w}_1$ (see \cref{fig:overlaps_move_attribute}) will detect the same occurrences of the premise of $\textsf{w}_1$ (i.e., methods contained within the same class), as they only differ in the mapping of the method nodes.
When evaluating application conditions for the number of repairs (impairs), this will lead to overcounting.
In the graph shown in \cref{fig:typed_graph}, there are two occurrences of the premise of $\textsf{w}_1$ involving the nodes \textsf{Item}, \textsf{itemTotal}, and \textsf{itemSingle} (\textsf{$m_1'$} can be mapped to \textsf{itemTotal}, \textsf{$m_2'$} can be mapped to \textsf{itemSingle}, and vice versa).
However, both overlaps $o_1$ and $o_2$ will detect two occurrences of the premise, which results in a total of four detected occurrences of the premise.
In other words, if these occurrences do not satisfy the conclusion of $\textsf{w}_1$ and are repaired during a transformation, two violations will be repaired, but the application condition will detect four repairs of the constraint.
This will lead to a false prediction of consistency gain if this overcounting is not corrected in some way.

To overcome this, we introduce a notion of \emph{equivalence for overlaps} and will then compute only one application condition for each of these equivalence classes, so that repairs (impairments) are only detected once.
This leads to a reduction in the total number of application conditions that need to be computed, and a simpler method of calculating the total gain in consistency using application conditions.
Intuitively, two overlaps of the same graphs are \emph{equivalent} if there is an isomorphism between the overlap graphs that preserves the embeddings of the graphs with which these overlaps were computed.

\begin{figure}
	\centering
	\begin{tikzpicture}[scale = 1]
		\node(GH1) at (0,0){$GH$};
		\node(GH2) at (0,3){$GH'$};

		\node(G) at (-3,1.5){$G$};
		\node(H) at (3,1.5){$H$};

		\draw[right hook-stealth](G) edge node [fill=white] {$i_{G}$} (GH1);
		\draw[right hook-stealth](G) edge node [fill=white] {$i'_{G}$} (GH2);

		\draw[left hook-stealth](H) edge node [fill = white] {$i_{H}$} (GH1);
		\draw[left hook-stealth](H) edge node [fill = white] {$i'_{H}$} (GH2);

		\draw[-stealth](GH1) edge node [left]{$\sim$}(GH2);
	\end{tikzpicture}
	\caption{Isomorphism for the equivalence of overlaps}
	\label{fig:equivalence_of_overlaps}
\end{figure}

\begin{defi}[Equivalence of overlaps]\label{def:equivalence_classes}
	Two overlaps $o = (i_G \colon G \inj GH, i_H \colon H \inj GH)$ and $o' = (i'_G \colon G \inj GH', i'_H \colon H \inj GH')$ over the graphs $G$ and $H$ are \emph{equivalent}, denoted by $o \cong o'$, if there is an isomorphism $\sim \colon GH \inj GH'$ so that $i'_G = \sim \circ i_G$ and $i'_H = \sim \circ i_H$, i.e., the diagram shown in \cref{fig:equivalence_of_overlaps} commutes.
	The \emph{equivalence class of an overlap $o$} is given by $$[o]:= \{ o' \mid o \cong o' \}.$$
\end{defi}

\begin{figure}
	\includegraphics[scale=0.8]{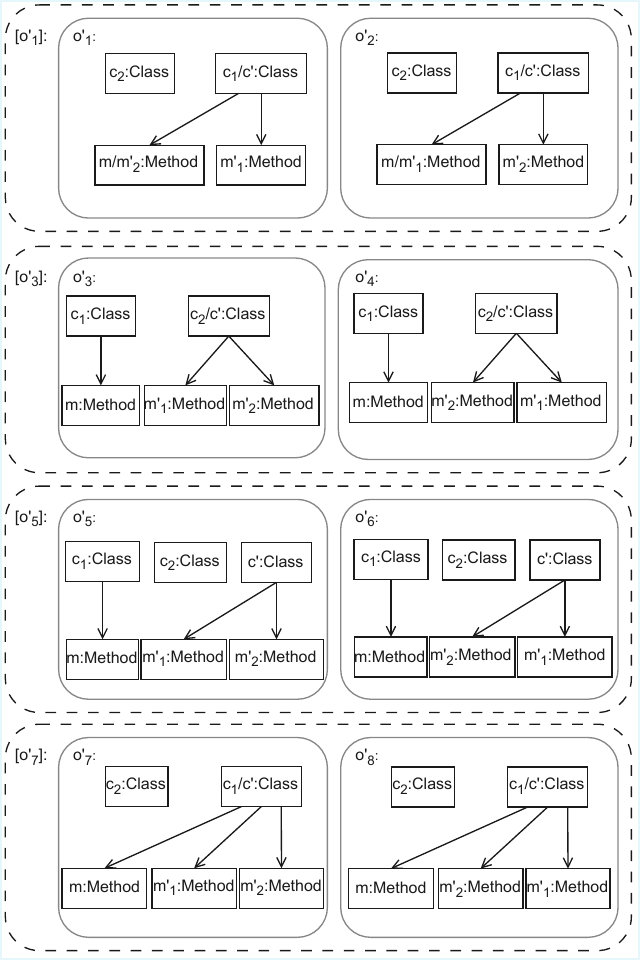}
	\caption{Overlaps $o'_1, \ldots, o'_8$ of the \ac{LHS} of \textsf{moveMethod} and the premise of $\textsf{w}_1$ and their respective equivalence classes $[o_1'], [o_3'],[o'_5]$ and $[o_7']$ (marked with the dashed boxes)}\label{fig_equivalence_classes_dep}
\end{figure}

\begin{exa}\label{ex:equivalence_classes_1}
	The overlap equivalence classes of \textsf{moveAttribute} and the premise of $\textsf{w}_1$ are shown in \cref{fig:overlaps_move_attribute}.
	The node labels implicitly denote the mappings of the LHS of the rule and the overlapped graph. 
	For example, the label $c_1/c'$ indicates that the node $c_1$ from the LHS of \textsf{moveAttribute} and $c'$ from the premise of $\textsf{w}_1$ are mapped to $c_1/c'$, i.e., $i_L(c_1) = c_1/c'$ and $i_P(c') = c_1/c'$.
	 There are three equivalence classes, which are marked by the dashed boxes. 
	\Cref{fig_equivalence_classes_dep} shows the equivalence classes of \textsf{moveMethod} and the premise of $\textsf{w}_1$. Again, we obtain three equivalence classes.
\end{exa}

As the following lemma shows, two overlaps are equivalent if and only if they can detect the same occurrences of the graph $P$.
It ensures that the notion of \emph{equivalence of overlaps} has the desired effect as described above.

\begin{figure}
	\centering
	\begin{tikzpicture}[scale = 1]
		\node(GH1) at (0,0){$GH'$};
		\node(GH2) at (0,3){$GH$};
		\node(X) at (0,1.5){$X$};

		\node(G) at (-3,1.5){$G$};
		\node(H) at (3,1.5){$H$};

		\draw[right hook-stealth](G) edge node [fill=white] {$i'_{G}$} (GH1);
		\draw[right hook-stealth](G) edge node [fill=white] {$i_{G}$} (GH2);

		\draw[left hook-stealth](H) edge node [fill = white] {$i'_{H}$} (GH1);
		\draw[left hook-stealth](H) edge node [fill = white] {$i_{H}$} (GH2);

		\draw[left hook-stealth](GH1) edge node [left]{$p'$}(X);
		\draw[right hook-stealth](GH2) edge node [left]{$p$}(X);
	\end{tikzpicture}
	\caption{Diagram for the characterisation of equivalence classes of overlaps}
	\label{fig:construction_isomorphism}
\end{figure}

\begin{lem}[Characterisation of equivalence classes]\label{lem:equiv_iso_mor}
	Given two graphs $G$ and $H$, two overlaps $(i_G, i_H), (i'_G, i'_H) \in \Over(G,H)$ with overlap graphs $GH$ and $GH'$, and a morphism $p \colon GH \inj X$ into some graph $X$.
	Then $(i_G, i_H) \cong (i'_G, i'_H)$ if and only if there is a morphism $p' \colon GH' \inj X$ so that the diagram shown in \cref{fig:construction_isomorphism} commutes, i.e.,
	\begin{align*}
		&p' \circ i'_G = p \circ i_G \text{ and }\\
		&p' \circ i'_H = p \circ i_H.
	\end{align*}
\end{lem}

In particular, this is a crucial observation for the correctness of our approach: If we consider a constraint $\forall(P,d)$ and evaluate an application condition constructed over some equivalence class $[(i_L, i_P, PL)]$ with respect to some transformation $t \colon G \Longrightarrow_{\rho,m} H$ using some rule $\rho = \completeRle$ at the match $m$, we consider only those occurrences $p \colon PL \inj G$ of the overlap graph, so that the match $m$ factors through $p$ and $i_L$ (i.e., $m = p \circ i_L$).
\Cref{lem:equiv_iso_mor} implies that $p \circ i_P \neq p' \circ i_P'$ (both morphisms $p$ and $p'$ are occurrences of the premise in $G$) if $[(i_L', i_P', PL')]$ is another equivalence class and $p' \colon PL' \inj G$ is an occurrence of $PL'$, so that $m = p' \circ i_L'$.
Therefore, each repair (impairment) of the constraint is detected by exactly one application condition, and we only need to consider the set of \emph{equivalence classes of overlaps of a rule and premise of a constraint} to compute the necessary application conditions.

\subsection{Construction of Application Conditions}\label{sec:construction_appl_cond}

When constructing impairment- and repair-indicating application conditions, all cases in which an impairment can be introduced, or a violation can be repaired must be considered.
A transformation $t \colon G \Longrightarrow H$ can repair violations of a constraint $c = \forall(P,d)$ in the following two ways:
\begin{enumerate}
	\item \emph{Repair of the premise:} An occurrence of $P$ in $G$ that does not satisfy $d$ is deleted.
	\item \emph{Repair of the conclusion:} There exists an occurrence of $P$ which satisfies $d$ in $H$ but does not satisfy $d$ in $G$.
\end{enumerate}
Similarly, the transformation $t$ can introduce impairments in the following ways:
\begin{enumerate}
	\item \emph{Impairment of the premise}: A new occurrence of $P$ is introduced that does not satisfy $d$.
	\item \emph{Impairment of the conclusion}: An occurrence of $P$ satisfies $d$ in $G$ and not in $H$.
\end{enumerate}
The following notions of \emph{repaired} and \emph{impaired morphisms} characterise occurrences of the premise of the constraint that are either repaired or impaired.

\begin{defi}[Repaired morphism]\label{def:repair_morphism}
	Given a constraint $c = \forall(e \colon \emptyset \inj P, d)$, a morphism $p\colon P \inj G$ is called a \emph{repaired morphism w.r.t.~a transformation} $t \colon G \Longrightarrow_{\rho,m} H$ if $p \not \models d$ and
	\begin{enumerate}
		\item $p$ is a \emph{repair of the premise}, if $\track_t \circ p$ is not total, i.e., $p$ is destroyed by the transformation, or
		\item $p$ is a \emph{repair of the conclusion}, if $\track_t \circ p$ is total and $\track_t \circ p \models d$, i.e., $p$ satisfies the conclusion after the transformation.
	\end{enumerate}
\end{defi}

\begin{defi}[Impaired morphism]\label{def:impaired_morphism}
	Given a constraint $c = \forall(e \colon \emptyset \inj P, d)$, a morphism $p\colon P \inj H$ is called an \emph{impaired morphism w.r.t.~a transformation} $t \colon G \Longrightarrow_{\rho,m} H$ if $p \not \models d$ and
	\begin{enumerate}
		\item $p$ is an \emph{impair of the premise}, if there is no morphism $p' \colon P \inj G$ with $p = \track_t \circ p'$, i.e., $p$ is created by the transformation, or
		\item $p$ is an \emph{impair of the conclusion}, if there is a morphism $p'\colon P \inj G$ with $p = \track_t \circ p'$ and $p' \models d$, i.e., $p$ satisfied the conclusion in the original graph of the transformation.
	\end{enumerate}
\end{defi}

The notions of \emph{repaired} and \emph{impaired} morphisms of a transformation $t$ are closely related, since a repaired morphism w.r.t.~$t$ and some constraint $c$ is an impaired morphism w.r.t.~the inverse transformation $t^{-1}$ and $c$.
This observation also implies that we only need an overlap-based construction for repair-indicating application conditions, since impair-indicating application conditions can then be constructed by computing the repair-indicating application conditions of the inverse rule and shifting them to the \ac{LHS}.

\begin{lem}\label{lem:relation_impair-repair}
	Given a constraint $c = \forall(e \colon \emptyset \inj P,d)$, a morphism $p \colon P \inj H$ is a repaired morphism w.r.t.~a transformation $t \colon G \Longrightarrow_{\rho,m} H $ if and only if $p$ is an impaired morphism w.r.t.~the inverse transformation $t^{-1} \colon H \Longrightarrow_{\rho^{-1},n} G $, where $n$ is the comatch of $t$.
\end{lem}

In order to construct application conditions that are able to detect both, repairs (impairments) of the premise and repairs (impairments) of the conclusion, different overlaps of the \ac{LHS}, and the premise of the constraints must be considered.
For repairs of the premise during a transformation $t \colon G \Longrightarrow_{\rho,m} H$, we want to detect those occurrences of the premise that are destroyed by a transformation, i.e., we are searching for an occurrence $p \colon PL \inj G$ of an overlap $(i_L, i_P, PL) \in \Over(\rho, P)$ (where $P$ is the premise of the constraint to be considered) so that $m = p \circ i_L$ and the occurrence $p \circ i_P$ of the premise is destroyed.
This means that we have to consider each overlap, so that the induced transformations are parallel dependent, since the application of $\rho$ at $m$ will destroy $i_P$ (this implies, in particular, that a transformation of $\rho$ at $i_L$ destroys the occurrence $i_P$).

For repairs of the conclusion we want to detect those occurrences of the premise that are preserved by a transformation, i.e., we are searching for occurrences $p \colon PL \inj G$ of an overlap $(i_L,i_P, PL) \in \Over(\rho,P)$ so that $m = p \circ i_L$ and the occurrence $p \circ i_P$ is preserved.
This means that we have to consider every overlap, so that the induced transformations are parallel independent.
Therefore, we decompose the set of all equivalence classes into a set containing the classes whose induced transformations are parallel independent and a set containing the equivalence classes whose induced transformations are parallel dependent.
Note that if the induced transformations of an overlap $o$ are parallel dependent (independent), then the same is true for each overlap $o'$ that is equivalent to $o$.

\begin{defi}
	Given a rule $\rho = \completeRle$ and a graph $P$, the set of \emph{equivalence classes of overlaps of $\rho$ and $P$}, denoted by $\OverEq(\rho, P)$, is defined as the disjoint union $\OverEq(\rho,P) = \OverPre(\rho,P) \dot{\cup} \OverCon(\rho,P)$ where
	\begin{align*}
		&\OverPre := \{[o] \in \Over(\rho,P) \mid \text{The induced transformations of $o$ are parallel dependent}\} \text{ and}\\
		&\OverCon := \{[o] \in \Over(\rho,P) \mid \text{The induced transformations of $o$ are parallel independent}\}.
	\end{align*}
\end{defi}

\begin{exa}\label{ex:equivalence_classes_2}
	The overlap equivalence classes contained in $\OverEq(\textsf{moveAttribute},P)$, where $P$ is the premise of $\textsf{w}_1$, are shown in \cref{fig:overlaps_move_attribute}.
	There are three equivalence classes (marked by the dashed boxes).
	We have $\OverEq(\textsf{moveAttribute},P) = \OverCon(\textsf{moveAttribute},P)$, because the induced transformations are parallel independent for each of these classes: When applying \textsf{moveAttribute}, the edges from nodes which contain the identifier $\textsf{c}_1$ to attribute \textsf{a} are deleted.
	This edge is not mapped by the premise of $\textsf{w}_1$.
	\Cref{fig_equivalence_classes_dep} shows the equivalence classes of \textsf{moveMethod} and the premise of $\textsf{w}_1$.
	Here, we have $\OverPre(\textsf{moveMethod},P) = \{[o'_1]\}$ (applying \textsf{moveMethod} would delete the edges from $\textsf{c}_1/\textsf{c'}$ to $\textsf{m/m}_1$ and $\textsf{m/m}_2$ respectively; these are mapped by the premise of $\textsf{w}_1$) and $\OverCon(\textsf{moveMethod},P) = \{[o'_3],[o'_5], [o'_7]\}$ (here, the deleted edges from classes to methods are not mapped by the premise of $\textsf{w}_1$).
\end{exa}

To check whether an occurrence of the premise of a constraint satisfies the conclusion before and after a transformation (i.e., in the original and in the resulting graph of a transformation), we introduce \emph{overlap-induced pre- and post-conditions}.
Given an overlap $(i_L,i_P,PL)$ of a constraint with the \ac{LHS} of a given rule $\rho$, the overlap-induced pre-condition checks whether an occurrence of the premise satisfies the conclusion of the constraint.
I.e., an occurrence $p \colon PL \inj G$ (see \cref{fig:equivalence_classes}) satisfies the overlap-induced pre-condition if and only if $p \circ i_P$ satisfies the conclusion of the constraint.
The overlap-induced post-condition gives a look-ahead on whether an occurrence of the premise will satisfy the conclusion of the constraint after some transformation.
In particular, an occurrence $p \colon PL \inj G$ satisfies the overlap induced post-condition if and only if $\track_t \circ p \circ i_P$ satisfies the conclusion of the constraint where $t \colon G \Longrightarrow_{\rho, p \circ i_L} H$.
Note that we only need the overlap-induced post-condition if the occurrence $p \circ i_P \colon P \inj G$ is not destroyed during the transformation, i.e., if the induced transformations of the overlap are parallel independent.
Otherwise, if the induced transformations of the overlap are parallel dependent, the overlap-induced post-condition is equal to $\true$.

\begin{defi}[Overlap-induced pre- and post-conditions]
	Given a rule $\rho = \rle{L}{l}{K}{r}{R}$, an overlap $o = (i_L,i_P, PL) \in \Over(\rho,P)$, and a condition $d$ over $P$, the \emph{overlap-induced pre-} and \emph{post-condition} of $\rho$ w.r.t.~$o$ and $d$, denoted by $\pre_{\rho}(o,d)$ and $\post_{\rho}(o,d)$ are defined as
	\begin{enumerate}
		\item $\pre_{\rho}(o,d) := \shift(i_P,d)$, and
		\item
		$\post_{\rho}(o,d) := \begin{cases}
			\leftshift(\rho', \shift(h \circ x,d)) &\text{if } [o] \in \OverCon(\rho,P) \\
			\true &\text{if } [o] \in \OverPre(\rho,P),
		\end{cases}$
	\end{enumerate}
	where $x,h$ are the morphisms and $\rho' = \rle{H}{h}{D}{g}{PL}$ is the rule as shown in \cref{fig:equivalence_classes}.
\end{defi}

\begin{figure}
	\includegraphics[scale = 0.8]{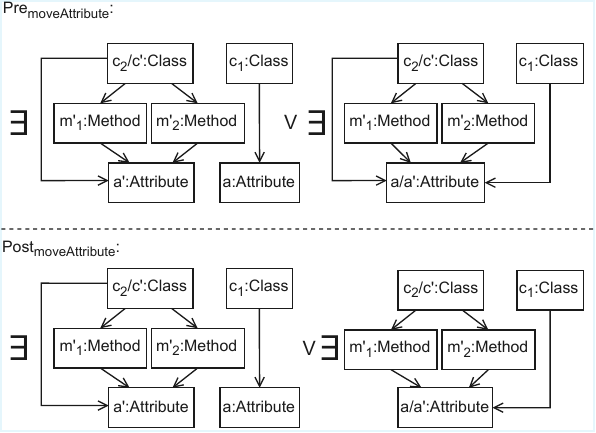}
	\caption{Induced pre- and post-conditions of \textsf{moveAttribute} w.r.t.~$[o_3]$}\label{fig:induced_conditions}
\end{figure}

\begin{exa}\label{ex:postconditions}
	If we consider the rule \textsf{moveAttribute}, constraint $\textsf{w}_1$ (\cref{fig:example:rule+constraints}) and the equivalence class $[o_3]$ shown in \cref{fig:overlaps_move_attribute}, the induced pre- and post-conditions of \textsf{moveAttribute} w.r.t.~$[o_3]$ are shown in \cref{fig:induced_conditions}.
	If the overlaps of $[o_3]$ occur in a transformation via \textsf{moveAttribute}, the attribute \textsf{a} is moved from class $\textsf{c}_1$ to class $\textsf{c}_2$.
	The pre-condition checks whether each pair of methods contained in the class $\textsf{c}_2$ shares a common attribute before the transformation is performed.
	There are two possibilities for this: Either an additional attribute \textsf{a'} contained in $\textsf{c}_2$  is shared by both methods (as in the left graph of the pre-condition), or they share the attribute \textsf{a}, which is to be moved, and \textsf{a} is also contained in \textsf{$c_2$} (as in the right graph of the pre-condition).

	The post-condition predicts whether each pair of methods contained in the class $\textsf{c}_2$ shares a common attribute after the transformation has been performed.
	Again, there are two possibilities for this: Similar to the pre-condition, either an additional attribute \textsf{a'} contained in $\textsf{c}_2$  can be shared by both methods (as in the left graph of the post-condition) or they share the attribute \textsf{a}, which is to be moved (as in the right graph of the pre-condition).
	Note that, for the post-condition \textsf{a} must not be contained in $c_2$ before the transformation, as it will be contained in $c_2$ after the transformation has been performed.
\end{exa}

Before continuing with the construction of the application condition, we present two simplifications for conditions that will reduce the number of graphs contained in the application conditions.
The first lemma enables us to simplify the application conditions by exploiting the fact that hard constraints are always satisfied.
If we consider a hard constraint of the form $c = \neg \exists(P)$ (which is satisfied by a graph $G$ if $P$ is not a subgraph of $G$) and a condition $\exists(e \colon P' \inj Q, d)$, a morphism $p \colon P' \inj G$ cannot satisfy $\exists(e \colon P' \inj Q,d)$ if $Q \not \models c$, otherwise $P$ is a subgraph of $G$.
Therefore, when simplifying an application condition, we can replace any condition of the form $\exists(e \colon P' \inj Q,d)$ with $Q \not \models c$ by $\false$.
Since $\forall(e \colon P' \inj Q, d) = \neg \exists(e \colon P' \inj Q, \neg d)$, conditions of the form $\forall(e \colon P' \inj Q, d)$ can be replaced by $\true$ if $Q \not \models c$.
In our running example, this simplification drastically reduces the number and complexity of derived impairment- and repair-indicating application conditions.

\begin{lem}[Simplification of graph conditions using hard constraints]\label{lemma:naive_filter}
	Let a constraint $c = \forall(e \colon \emptyset \inj P, \false)$ and a condition $c' = \exists(e' \colon P' \inj Q', d)$ with $Q' \not\models c$ be given, then
	$G \models c \implies p \not \models c',$
	for each morphism $p \colon P' \inj G$.
\end{lem}
\begin{exa}\label{ex:first_simplification}
	If we consider the induced pre-condition of \textsf{moveAttribute} w.r.t.~$[o_3]$ shown in \cref{fig:induced_conditions}, the second graph does not satisfy the hard constraint $\textsf{h}_2$, because the attribute \textsf{a} is contained in two classes.
	Therefore, we can replace this graph with $\false$ and the induced pre-condition contains only the existentially bound graph on the left of \cref{fig:induced_conditions}.
\end{exa}

The second simplification of conditions is based on the propositional logical equivalence that $(a \vee b) \implies (a \vee c)$ is equivalent to $b \implies (a \vee c) $.
This allows conditions to be removed if they occur in both parts of the implication.
This again reduces the number of graphs contained in the application condition.
\begin{lem}[Simplification of graph conditions using propositional logic]\label{cor:naive_filter}
	Given a condition $c = (\bigvee_{i \in \mathcal{I}} c_i) \implies (\bigvee_{j \in \mathcal{J}} d_j)$, with index sets $\mathcal{I}$ and $\mathcal{J}$, then this condition is equivalent to the condition $c' = (\bigvee_{i \in \mathcal{I}\setminus\{i_0\}} c_i) \implies (\bigvee_{j \in \mathcal{J}} d_j)$ for some $i_0 \in \mathcal{I}$ if there is a $j_0 \in \mathcal{J}$ with $c_{i_0} \equiv d_{j_0}$.
\end{lem}

\begin{exa}
	Again, consider the induced pre- and post-conditions of \textsf{moveAttribute} w.r.t.~$[o_3]$ as described in \cref{ex:postconditions}. As described in \cref{ex:first_simplification}, the pre-condition can be simplified so that it contains only the left graph shown in \cref{fig:induced_conditions}.
	This means that the condition $\post_{\textsf{moveAttribute}} \implies \pre_{\textsf{moveAttribute}}$ is of the form $(\exists Q_1 \vee \exists Q_2) \implies \exists Q_1$, where $Q_1$ is the graph on the bottom left and $Q_2$ is the graph on the bottom right of \cref{fig:induced_conditions}, respectively.
	Using \cref{cor:naive_filter}, we can remove the occurrence of $\exists Q_1$ from the left side of the implication, since it also appears on the right side of the implication.
	This simplification leads to the equivalent condition $\exists Q_2 \implies \exists Q_1$, i.e., $\post_{\textsf{moveAttribute}} \implies \pre_{\textsf{moveAttribute}} \equiv \exists Q_2 \implies \exists Q_1$.
\end{exa}

Now we are ready to present the construction of \emph{repair-indicating application conditions}.
To detect all situations in which a violation of a constraint $c= \forall(P,d)$ is repaired (i.e., to detect each occurrence of the premise), we need to consider each overlap of the \ac{LHS} of the rule $\rho$ and the premise of the constraint, such that this overlap has induced transformations (i.e., we need to consider each overlap $o = (i_L, i_P, PL)$ contained in $\Over(\rho, P)$).
In particular, \cref{lem:equiv_iso_mor} implies that it is sufficient to consider only the equivalence classes contained in $\OverEq(\rho, P)$.
To find repairs of the conclusion w.r.t.~a transformation, we consider overlaps whose induced transformations are parallel independent (i.e., the occurrence of the premise is preserved by the transformation).
We first need to check whether the occurrence of the premise satisfies the conclusion of the transformation in the resulting graph of the transformation using the overlap induced post-condition.
In addition, if this occurrence of the premise does not satisfy the conclusion in the original graph of the transformation, then this occurrence is a repair of the conclusion.
I.e., if an occurrence of the overlap does not satisfy the condition $\post(o,d) \implies \pre(o,d)$, we have found a repair of the conclusion (as then it satisfies the conclusion in the resulting but not in the original graph of the transformation).
\emph{This means, that each violation of the application condition $\forall(PL, \post(o,d) \implies \pre(o,d))$ represents a repair of the conclusion}.
To find repairs of the premise w.r.t.~a transformation, we consider overlaps whose induced transformations are parallel dependent (i.e., the occurrence of the premise is destroyed by the transformation).
Since the overlap-induced post-condition of that overlap is equal to $\true$ and for propositional logic formula it holds that $\true \implies a \equiv a$, we have found a repair of the premise if the overlap does not satisfy the condition $\post(o,d) \implies \pre(o,d)$, i.e., it does not satisfy the conclusion in the original graph of the transformation.
Here, \emph{each violation of the application condition $\forall(PL, \post(o,d) \implies \pre(o,d))$ represents a repair of the conclusion}.
Therefore, we can compute repair-indicating application conditions for both repairs of the conclusion and repairs of the premise via the following construction.

\begin{defi}[Repair-indicating application condition]
	Given a rule $\rho$ and a constraint $c = \forall(e \colon \emptyset \inj P,d)$, the set of \emph{repair-indicating application conditions for $\rho$ w.r.t.~$c$}, denoted by $\repair(\rho,c)$, is defined as 
	\begin{align*}
		\repair(\rho,c) := &\{ \forall(i_L \colon L \inj PL, \post_{\rho}(o,d) \implies \pre_{\rho}(o,d)) \mid \\
		&[o] = [(i_L, i_P,PL)] \in \OverEq(\rho,P) \}.
	\end{align*}
\end{defi}

For the special case, if $c = \forall(e \colon \emptyset \inj P, \false)$ (a graph $G$ satisfies this constraint if $P$ is not a subgraph of $G$, i.e., there is no morphism $p \colon P \inj G$), the constraint can only be repaired (impaired) by repairs (impairments) of the premise.
For an overlap $[o] = [(i_L, i_P, PL)] \in \OverPre(\rho,P)$, i.e., if the induced transformations are parallel dependent and application conditions constructed over this overlap will detect repairs of the premise, we have $\post_{\rho}(o,d) = \true$ and, by using the definition of the $\shift$ operator, $\pre_{\rho}(o,d) = \false$.
This results in the application condition $\forall(i_L \colon L \inj PL, \false)$, which detects each destroyed occurrence of the premise as a repair of the constraint.
For an overlap $[o] = [(i_L, i_P, PL)] \in \OverCon(\rho,P)$, i.e., if the induced transformations are parallel independent and application conditions constructed over this overlap will detect repairs of the conclusion, we have $\post_{\rho}(o,d) = \pre_{\rho}(o,d) = \false$ and we obtain the application condition $\forall(i_L \colon L \inj PL, \true)$ which is always satisfied and reflects the fact that there can be no repair of the conclusion for a constraint of the form $\forall(P, \false)$.

\begin{exa}
	\Cref{fig:resulting_application_conditions} shows the repair-indicating application conditions for the rules \textsf{moveAttribute} and \textsf{moveMethod} w.r.t.~the constraints $\textsf{w}_1$ and $\textsf{w}_2$.
	The application conditions w.r.t.~$\textsf{w}_1$ are constructed for the overlap classes $[o_3]$ and $[o'_1]$ shown in \cref{fig:overlaps_move_attribute} and \cref{fig_equivalence_classes_dep}. The application conditions constructed with the other equivalence classes can be simplified using \cref{lemma:naive_filter} and \cref{cor:naive_filter}, so that it is easy to see that they are equivalent to $\true$ and can therefore be disregarded.
	Note also that the conditions contained in the sets $\repair(\textsf{moveAttribute}, \textsf{w}_2)$, $\repair(\textsf{moveMethod}, \textsf{w}_2)$, and $\repair(\textsf{moveMethod}, \textsf{w}_1)$ will detect repairs of the premise, while the application conditions contained in $\repair(\textsf{moveAttribute}, \textsf{w}_1)$ detect repairs of the conclusion.
\end{exa}

As shown in \cref{lem:relation_impair-repair}, a morphism is a repaired morphism w.r.t.~some transformation $t$ if and only if it is an impaired transformation w.r.t.~the inverse transformation $t^{-1}$, i.e., we can use the same construction to compute impair-indicating application conditions by switching the role of the \ac{LHS} and the \ac{RHS} of the rule.
So the impair-indicating application conditions can be obtained by computing the repair-indicating application conditions of the inverse rules and shifting them to the \ac{LHS}.

\begin{defi}[Impairment-indicating application condition]
	Given a rule $\rho$ and a constraint $c = \forall(e \colon \emptyset \inj P,d)$, the set of \emph{impairment-indicating applications for $\rho$ w.r.t.~$c$}, denoted by $\violation(\rho,c)$, is defined as	$$\violation(\rho,c) := \{\leftshift(\rho, ac) \mid ac \in \repair(\rho^{-1}, c)\}.$$
\end{defi}

\begin{figure}
	\includegraphics[scale = 0.8]{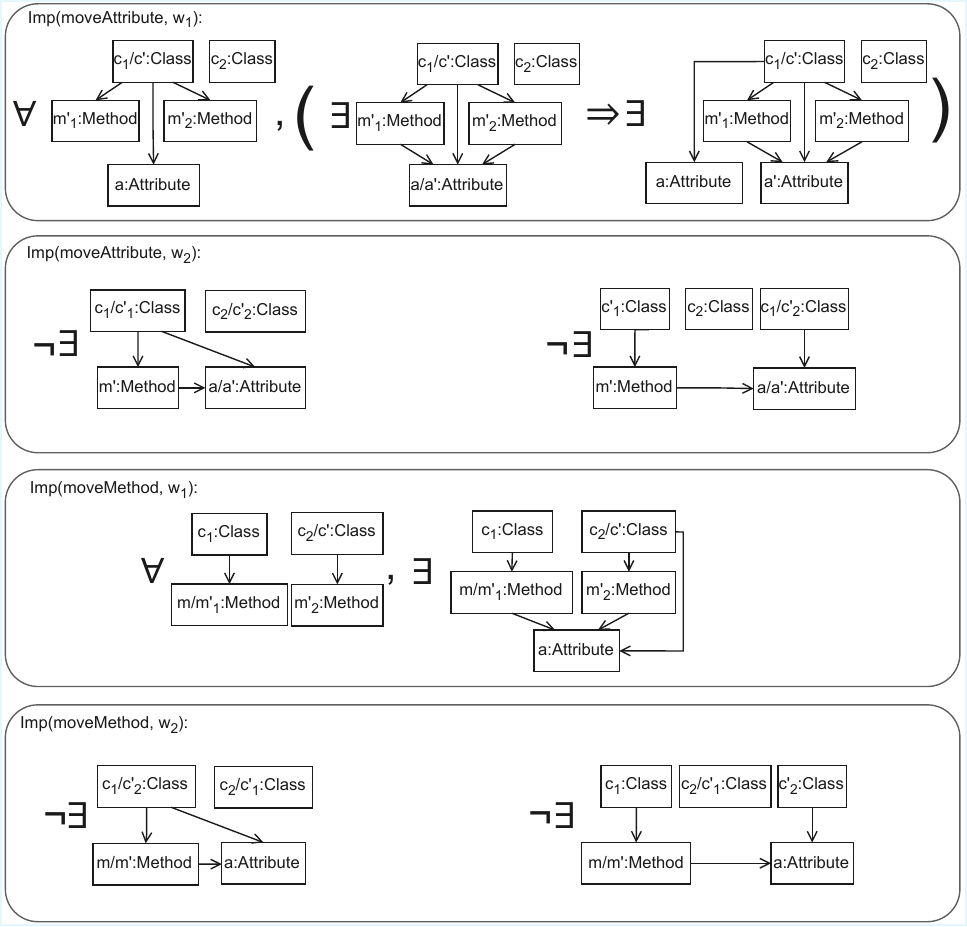}
	\caption{Impairment-indicating application conditions for the rules \textsf{moveAttribute} and \textsf{moveMethod} w.r.t.~the constraints $\textsf{w}_1$ and $\textsf{w}_2$}
	\label{fig:impair_indicating_application_conditions}
\end{figure}

\begin{exa}
	\Cref{fig:impair_indicating_application_conditions} shows the impairment-indicating application conditions for the rules \textsf{moveAttribute} and \textsf{moveMethod} w.r.t.~the constraints $\textsf{w}_1$ and $\textsf{w}_2$.
	The conditions contained in the sets $\violation(\textsf{moveAttribute}, \textsf{w}_2)$, $\violation(\textsf{moveMethod}, \textsf{w}_2)$, and $\violation(\textsf{moveMethod}, \textsf{w}_1)$ will detect impairments of the premise, while the application conditions contained in $\violation(\textsf{moveAttribute}, \textsf{w}_1)$ will detect impairments of the conclusion.
	Note that the right application conditions in $\violation(\textsf{moveMethod}, \textsf{w}_2)$ and $\violation(\textsf{moveAttribute}, \textsf{w}_2)$ are also contained in the sets $\repair(\textsf{moveMethod}, \textsf{w}_2)$ and $\repair(\textsf{moveAttribute},\textsf{w}_2)$, respectively.
	These four application conditions can be removed as their results will cancel each other out when evaluating the consistency gain, since each impairment detected by the right application condition in $\violation(\textsf{moveMethod},\textsf{w}_2)$ will also be detected as a repair by the right application condition in $\repair(\textsf{moveMethod}, \textsf{w}_2)$.
	Therefore, these application conditions cannot influence the consistency gain, either positively or negatively.
\end{exa}

Using the repair-indicating and impairment-indicating application conditions, we can now predict the change of consistency caused by a transformation.
To evaluate the number of impairments, we consider the number of violations of the impairment-indicating application conditions, i.e., each violation of an impairment-indicating application condition corresponds to an impairment of the constraint.
The total sum of these violations is the total number of impairments introduced by the transformation.
For the number of repairs, we consider the number of violations of the repair-indicating application conditions, i.e., each violation of a repair-indicating application condition corresponds to a repair of the constraint.
Again, the total sum of these violations is the total number of repairs introduced by the transformation.
Since each application condition is formulated over the \ac{LHS} of the rule, this gives a look-ahead on the exact gain without actually performing the transformation and reevaluating the number of violations in the resulting graph.

\begin{thm}\label{thm:main_theorem}
	Given a transformation $t\colon G \Longrightarrow_{\rho,m} H$ via a rule $\rho = \completeRle$ and a constraint $c$, then $$\nv_H(c) - \nv_G(c) = \sum_{ac \in \violation(\rho,c)} \nv_m(ac) - \sum_{ac \in \repair(\rho,c)} \nv_m(ac).$$
\end{thm}

\begin{figure}
	\includegraphics[scale=.8]{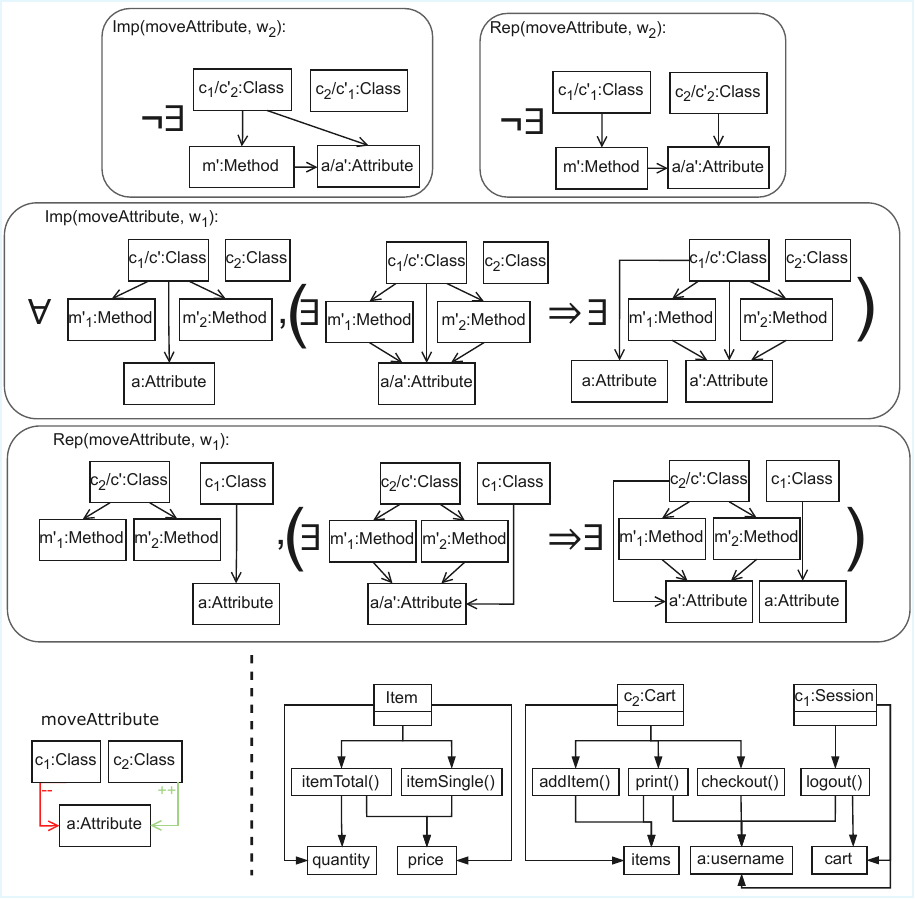}
	\caption{Repair- and impair-indicating application conditions for \textsf{moveAttribute} and $\textsf{w}_2$,  \textsf{moveAttribute} and $\textsf{w}_1$, and an application of \textsf{moveAttribute} to move \textsf{username} from \textsf{Session} to \textsf{Cart}}
	\label{fig:example_eval}
\end{figure}

\begin{exa}\label{example:overall}
	Let us consider the application of \textsf{moveAttribute} to move \textsf{username} from \textsf{Session} to \textsf{Cart}.
	The rule, match, and application conditions used throughout this example are shown in \cref{fig:example_eval}.
	When evaluating the application condition contained in $\repair(\textsf{moveAttribute}, \textsf{w}_1)$ (which signals repairs of $\textsf{w}_1$), we first check whether the match can be extended by the first graph of the application condition.
	There are six such extensions in total:
	\begin{enumerate}
		\item $\textsf{m}'_1$ is mapped to \textsf{addItem()} and $\textsf{m}'_2$ is mapped to \textsf{print()},
		\item $\textsf{m}'_1$ is mapped to \textsf{print()} and $\textsf{m}'_2$ is mapped to \textsf{addItem()},
		\item $\textsf{m}'_1$ is mapped to \textsf{addItem()} and $\textsf{m}'_2$ is mapped to \textsf{checkout()},
		\item $\textsf{m}'_1$ is mapped to \textsf{checkout()} and $\textsf{m}'_2$ is mapped to \textsf{addItem()},
		\item $\textsf{m}'_1$ is mapped to \textsf{checkout()} and $\textsf{m}'_2$ is mapped to \textsf{print()},
		\item $\textsf{m}'_1$ is mapped to \textsf{print()} and $\textsf{m}'_2$ is mapped to \textsf{checkout()}.
	\end{enumerate}
	In the next step, we need to check whether these extensions can be further extended by the left graph of the implication, but not by the right graph of the implication.
	If so, we have found a repair of $\textsf{w}_1$.
	For the left part of the implication, we need to check if both methods of the extension are connected to the attribute \textsf{username}.
	This only applies to extensions $5$ and $6$ in the list above since both methods  \textsf{print()} and \textsf{checkout()} depend on \textsf{username}.
	Extensions $1$--$4$ do not need to be considered anymore.

	In the last step, we need to check if extensions $5$ and $6$ can be further extended by the right graph of the implication, i.e., we need to check if there is an additional attribute contained in \textsf{Cart} that \textsf{checkout()} and \textsf{print()} are connected to.
	Since \textsf{checkout()} is not connected to any of the attributes contained in \textsf{Cart}, there is no such extension for $5$ and $6$.
	This means that $5$ and $6$ do not satisfy the implication and the application of \textsf{moveAttribute} will repair two violations of the constraint $\textsf{w}_1$, i.e., $\sum_{ac \in \repair(\rho, \textsf{w}_1)} \nv_m(ac) = 2$.
	When considering the impair-indicating application condition contained in $\violation(\textsf{moveAttribute}, \textsf{w}_1)$, the match cannot be extended by the first graph of the application condition, because \textsf{Session} only contains one method, i.e., the application condition is satisfied and $\sum_{ac \in \violation(\rho, \textsf{w}_1)} \nv_m(ac) = 0$.
	In total, we get a gain of $\sum_{ac \in \violation(\rho, \textsf{w}_1)} \nv_m(ac) - \sum_{ac \in \repair(\rho, \textsf{w}_1)} \nv_m(ac) = \nv_H(\textsf{w}_1) - \nv_G(\textsf{w}_1) = 0 -2 = -2$.

	Let us now consider the application conditions for \textsf{moveAttribute} and $\textsf{w}_2$.
	For the application condition contained in $\repair(\textsf{moveAttribute}, \textsf{w}_2)$, we need to check whether the match can by extended by the first graph of the application condition.
	There are two such extensions:
	\begin{enumerate}
		\item \textsf{m'} is mapped to \textsf{print()} and
		\item \textsf{m'} is mapped to \textsf{checkout()}.
	\end{enumerate}
	Both occurrences of the pattern in the application condition signal a repair of $\textsf{w}_2$, i.e.,
	$\sum_{ac \in \repair(\rho, \textsf{w}_2)} \nv_m(ac) = 2$.

	For the impairment-indicating application condition $\violation(\textsf{moveAttribute}, \textsf{w}_2)$, the rule match can only be extended by mapping \textsf{m'} to \textsf{logout()}.
	This occurrence signals an impairment of $\textsf{w}_2$, i.e., $\sum_{ac \in \violation(\rho, \textsf{w}_2)} \nv_m(ac) = 1$.
	In total, we get an overall gain of $\sum_{ac \in \violation(\rho, \textsf{w}_2)} \nv_m(ac) - \sum_{ac \in \repair(\rho, \textsf{w}_2)} \nv_m(ac) = \nv_H(\textsf{w}_2) - \nv_G(\textsf{w}_2) =1 - 2 = -1$ for $\textsf{w}_2$.
\end{exa}

In addition, we can predict the overall gain in consistency of a set of constraints and assign weights to each constraint to prioritise certain constraints over others if desired.

\begin{cor}
	Given a transformation $t \colon G \Longrightarrow_{\rho,m} H$ via some rule $\rho$ at match $m$, a finite set of constraints $\{c_1, \ldots, c_n\}$, and some weights $(w_1, \ldots, w_n) \in \mathbb{R}^n$, then $$\sum_{i = 1}^{n} w_i \big(\nv_H(c_i) - \nv_G(c_i)\big) = \sum_{i = 1}^{n} w_i \Big(\sum_{ac \in \violation(\rho,c_i)} \nv_m(ac) - \sum_{ac \in \repair(\rho,c_i)} \nv_m(ac)\Big).$$
\end{cor}

\begin{exa}
	Consider the set of constraints $\{\textsf{w}_1,\textsf{w}_2\}$, the weights $w = (1,1) \in \mathbb{R}^2$ and a transformation via \textsf{moveAttribute} to move \textsf{username} from \textsf{Session} to \textsf{Cart} (see \cref{fig:example_eval}).
	As discussed in \cref{example:overall}, we have $\nv_H(\textsf{w}_1) - \nv_G(\textsf{w}_1) = 2$ and $\nv_H(\textsf{w}_2) - \nv_G(\textsf{w}_2) = 1$.
	This means that the overall consistency of this transformation is increased by 3 since $H$ contains fewer violations than $G$:
	$$\sum_{i=1}^2 \nv_H(\textsf{w}_i) - \nv_G(\textsf{w}_i) =(-2) + (-1) = -3$$
\end{exa}

\subsection{Comparison to Consistency Sustaining and Improving Transformations}\label{sec:comparisson}

To conclude the theoretical part of this paper, we want to compare the newly introduced application conditions with the notions of consistency-sustaining and -improving transformations introduced in~\cite{KosiolSTZ22}.
As described in Section~\ref{sec:counting_constraint_violation}, the number of violations is an extension of the corresponding notion presented in~\cite{KosiolSTZ22}.
Next, we will recall the notions of consistency-sustaining and -improving transformations and compare them with our approach.
We will see that a transformation is consistency-sustaining if the gain in consistency is not negative, i.e., if the consistency does not decrease during a transformation.
It is consistency-improving if the gain is positive, i.e., if the consistency increases during a transformation.
\emph{Please note that for this comparison, we are assuming that the weight of each constraint is equal to 1.}

\begin{defi}[Consistency sustaining and improving transformation~\cite{KosiolSTZ22}]
	Given a transformation $t\colon G \Longrightarrow_{\rho, m} H$ and a constraint $c = \forall(e \colon \emptyset \inj P,d)$, the transformation $t$ is \emph{consistency sustaining w.r.t.~$c$} if
	$$\nv_H(c) \leq \nv_G(c).$$
	The transformation $t$ is called \emph{consistency improving w.r.t.~$c$} if $$\nv_H(c) < \nv_G(c).$$
\end{defi}
The result of \cref{thm:main_theorem} implies that we can use the application conditions to predict whether a transformation is consistency-sustaining (-improving) or not.

\begin{cor}
	Given a transformation $t\colon G \Longrightarrow_{\rho, m} H$ and a constraint $c = \forall(e \colon \emptyset \inj P,d)$, the transformation $t$ is consistency sustaining w.r.t.~$c$ if and only if
	$$\sum_{ac \in \violation(\rho,c)} \nv_m(ac) - \sum_{ac \in \repair(\rho,c)} \nv_m(ac) \leq 0.$$
	The transformation $t$ is consistency-improving w.r.t.~$c$ if
	$$\sum_{ac \in \violation(\rho,c)} \nv_m(ac) - \sum_{ac \in \repair(\rho,c)} \nv_m(ac) < 0.$$
\end{cor}

Kosiol et al.~have also introduced a stricter variant of consistency-sustaining and -improving transformations called \emph{directly consistency-sustaining and (-improving) transformations}.
A consistency-sustaining (-improving) transformation is directly consistency-sustaining (-improving) if the transformation does not introduce any violations, i.e., there is no impaired morphism w.r.t.~the transformation.
Thus, for a transformation to be directly consistency-sustaining it is sufficient that no new violation is introduced.

\begin{defi}[Directly consistency-sustaining transformation~\cite{KosiolSTZ22}]\label{def:direct_sustaining}
	Given a constraint $c = \forall(e \colon \emptyset \inj P, d)$, a transformation $t \colon G \Longrightarrow_{\rho,m} H$ is \emph{directly consistency-sustaining} w.r.t.~$c$ if
	\begin{align*}
		&\forall p \colon P \inj G((p \models d \wedge \track_t \circ p \text{ is total}) \implies \track_t \circ p \models d) \text{ and } \\
		&\forall p' \colon P \inj H(\neg \exists p \colon P \inj G(p' = \track_t \circ p) \implies p' \models d).
	\end{align*}
\end{defi}
Note that any morphism $\track_t \circ p$, where $p$ does not satisfy the first part of the condition, is an impairment of the conclusion according to \cref{def:impaired_morphism}, and any morphism $p'$ not satisfying the second part of the condition is an impairment of the premise.
For a directly consistency-improving transformation, it is sufficient that at least one violation of the constraint is repaired, i.e., there is at least one repaired morphism.
\begin{defi}[Directly consistency-improving transformation~\cite{KosiolSTZ22}]\label{def:direct_improving}
	Given a constraint $c = \forall (e\colon \emptyset \inj P,d)$, a transformation $t \colon G \Longrightarrow_{\rho,m} H$ is \emph{directly consistency-improving} w.r.t.~$c$ if $t$ is direct consistency-sustaining w.r.t.~$c$ and
	\begin{align*}
		&\exists p \colon P \inj G (p \not \models d \wedge p' := \track_t \circ p \text{ is total } \wedge p \models d) \text{ or }\\
		&\exists p \colon P \inj G (p \not \models d \wedge p' := \track_t \circ p \text{ is not total}).
	\end{align*}
\end{defi}
Again, any morphism $p$ that does not satisfy the condition is a repaired morphism according to \cref{def:repair_morphism}.

In~\cite{KosiolSTZ22}, application conditions were presented that are only sufficient, i.e., in some cases, transformations are prevented even if the transformation would be consistency-sustaining (consistency-improving).
Our construction of application conditions in Section~\ref{sec:construction_appl_cond} computes application conditions that prevent transformations if and only if the transformation would not be directly consistency-sustaining (-improving).
In particular, this implies that they are the weakest directly consistency-sustaining (-improving) application conditions.

\begin{thm}\label{thm:direct_sustaining_improving_ac}
	Given a constraint $c = \forall(e \colon \emptyset \inj P,d)$, a transformation $t \colon G \Longrightarrow_{\rho,m} H$ is directly consistency-sustaining w.r.t.~$c$ if and only if
	$$m \models \bigwedge_{ac \in \violation(\rho,c)} ac.$$
	The transformation $t$ is directly consistency-improving w.r.t.~$c$ if and only if
	$$m \models \bigwedge_{ac \in \violation(\rho,c)} ac \wedge \neg (\bigwedge_{ac \in \repair(\rho,c)} ac).$$
\end{thm}

\section{Evaluation}
\label{chapter:evaluation}
To evaluate the practical relevance of our approach, we implemented a greedy graph optimization algorithm for a variant of the \ac{CRA} case study (cf.~Section~\ref{chapter:example}).
This was done using the graph transformation tool eMoflon\footnote{\url{https://www.emoflon.org}}, which incrementally computes all matches to a given graph for rules and their application conditions.
This is a prerequisite for the efficient implementation of this approach, as calculating all matches from scratch after each rule application can easily become a serious performance bottleneck.
Based on the matches provided by eMoflon, our implementation then counts violations of derived application conditions and ranks rule matches w.r.t.~the number of constraint violations that are removed or added by the application of the considered rule at the considered match (see \cref{thm:main_theorem}).
Then, the rule application with the highest rank is greedily selected and applied, and the ranking of rule matches is updated.
Note that the application conditions are currently not automatically derived, but were designed and implemented by hand in eMoflon based on our formal construction.

Our evaluation was run on a Ryzen 7 3900x with \SI{64}{\giga\byte} RAM on Windows 11 23H2.
It is available as a virtual machine\footnote{\url{https://www.zenodo.org/records/10727438}} with a detailed description of how to reproduce our results.
In the following, we first investigate the scalability and effectiveness of our approach.
To assess the effectiveness of our approach, we compare it to an \ac{ILP}-based approach that can find the global optimum of our test scenario\footnote{This comparison is an extension of our earlier work.}.

\subsection{Scalability and effectiveness}
\label{sec:scalabilityAndEffectiveness}

Finding all the matches for rules and application conditions can become very expensive if there are many matches to find.
The \ac{CRA} case study is particularly challenging because a feature can be moved from one class to any other class, which means that the number of refactoring steps, and thus the number of possible matches, grows rapidly with an increasing number of classes and features.
Therefore, we pose the following research questions: {\em (RQ1) How does our approach scale with respect to the size of a processed class diagram (graph)?} and {\em (RQ2) Can our approach reduce the number of violations?}

To answer the first question, we need to examine the two phases of our approach, which are related to how eMoflon works and incrementally provide us with collected matches.
First, we measure the time it takes to compute the initial collection of all rule matches and application condition occurrences.
Then, we use the application condition occurrences to rank the rule matches.
Depending on the size of the model, the initialisation is expected to take longer than applying a rule, incrementally updating eMoflon's internal structures, and ranking the rule matches.
Second, we measure the time taken to perform 10 repair steps, where we have to judge which repair to apply next, based on an actually selected repair step.
For this, we use the most promising (i.e., highest ranked) rule application, where each repair and each impairment is uniformly weighted with a value of 1.

To investigate the scalability of our approach, we created synthetic class diagrams of varying sizes with increasing numbers of classes, where each class has five methods and five attributes.
Each method has two dependencies on attributes of the same class and another three dependencies on attributes of other classes.
Having more dependencies means that it is less likely to move features from one class to another, because more features would form a dependency clique within a class.

\begin{figure}
	\centering
	\includegraphics[width=1\textwidth, trim={0 1.13cm 0 0.3cm}, clip]{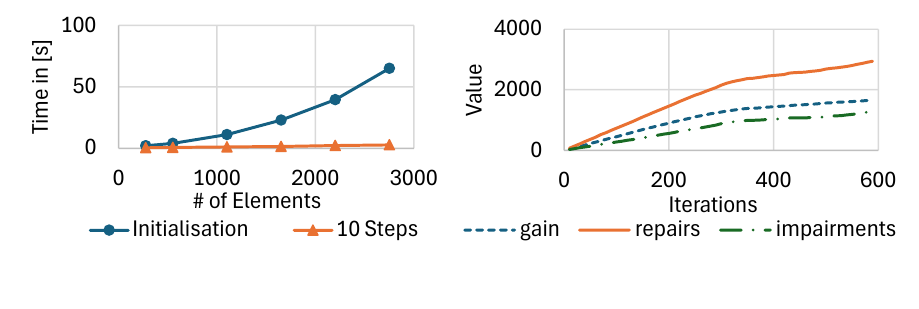}
	\caption{Performance measurements for models with increasing size (left) and consistency improvement of a model with 2201 nodes (right)}
	\label{fig:example:eval}
\end{figure}

\Cref{fig:example:eval} shows our evaluation results, where the left plot shows the time in seconds for processing models with up to 2,751 nodes.
Starting with 276 nodes, the initialisation time is \SI{2.2}{\second}, and performing 10 refactoring steps takes \SI{0.5}{\second}.
With 2,751 nodes, this time increases to \SI{65}{\second} for the initialisation and \SI{2.8}{\second} for the 10 steps.
Obviously, the initialisation takes 30 times longer for a 10 times larger model, while the incremental updates scale better and take only 5 times longer.
The reason for this non-linear increase lies in the structure of our rules, where the number of matches increases rapidly with each new class.
For a model with 276 nodes, we collected 600 rule matches and 36,000 application condition matches, while for the largest model with 2,751 nodes, we found 622,500 rule matches and 3,7 million condition matches.
So, while it takes 30 times as long to run a model 10 times the size, we found 118 times as many matches.
This shows that even for a challenging scenario like the \ac{CRA} use case, our approach scales reasonably well given the number of repair steps available (RQ1).

To answer RQ2, we took a model with 2,201 nodes and measured the aggregated number of impairments and repairs after n iterations (i.e., repair steps) along with their gains in consistency.
As before, in each iteration, we chose the rule application with the highest gain and continued until there was no rule application left with a positive gain, i.e., the application of a rule at that point would have no effect or cause more impairments than repairs. 
The results are shown in \cref{fig:example:eval} on the right.
After 589 iterations with a total gain of 1,661, the process terminated with a (local) optimum.
Thus, the resulting model contains 1,661 fewer violations than before, leading to a first answer to RQ2.

As shown, our approach can incrementally maintain the necessary ran\-king information for ru\-le-matches (RQ1) and improve consistency in a relatively small number of iterations (RQ2).
This is particularly interesting considering the fact that the search space of all rule matches can grow very rapidly, which is a particularly challenging scenario for our approach.

\subsubsection{Threats to validity}
\label{section:evaluation:threads}
Currently, we only investigate the \ac{CRA} case using synthetic class diagrams.
To evaluate the general scalability of our approach, more scenarios should be investigated, preferably using real-world data, e.g., extracted from public code repositories.
In addition, the application conditions are currently constructed by hand (following our formal construction process), which can be a source of error.
In addition, since we only have a look-ahead of 1, there may not always be a good next step to improve consistency, even though a better overall solution exists.
Therefore, future work should investigate how our approach performs when the greedy strategy is replaced by another strategy, such as simulated annealing.

\subsection{Comparison with an ILP-based approach}
\label{chapter:ext_evaluation}
In this subsection, we want to further investigate our approach and compare it with another solution that calculates an optimal solution using the \acf{GIPS} tool~\cite{gipsGCM2022}.
\ac{GIPS} is a framework that combines \ac{GT} and \ac{ILP} to define optimal model transformations.
We, thereby, want to show that our approach can indeed find the global optimum (at least for smaller models), with the benefit of scaling better for larger models than \ac{GIPS}.

As a basis for our comparison, we use the test data from the TTC16~\cite{DBLP:conf/staf/FleckTW16} \ac{CRA} case study, for which we have extended our example to also support dependencies between methods.
Note that these dependencies were omitted before (cf.~Section~\ref{chapter:example}).
Another difference between our previous variant and the original \ac{CRA} case study is that the latter consists only of features that need to be clustered into (newly created) classes.
Previously, we assumed a predefined class diagram, in which each feature was already assigned to a class.
Thus, the case study represents a more challenging scenario in which we now have to find out how many classes are needed to maximise the ratio between cohesion and coupling.
We bridge this gap similar to the \ac{CRA} solution of Georg Hinkel~\cite{Hinkel16} by creating a new class for each feature and then applying either our approach or \ac{GIPS} to the resulting model.
However, we do not compare our approach to any of the original \ac{CRA} competition entries because we are not actually implementing the \ac{CRA} index (the metric that was used in the TTC16), but rather a custom metric that is inspired by it.
To understand why, let us first introduce the metric itself.

Our custom metric is the sum of violations of the following two constraints:
First, features within the same class must have a dependency between each other, and, second, features within different classes must not have a dependency between each other.
The definition of a dependency is the same as in the TTC16~\cite{DBLP:conf/staf/FleckTW16} \ac{CRA} case study.
Thus, the resulting value of our metric represents the total number of violations of these two constraints within the system.
In contrast, the \ac{CRA} index is the difference between coupling and cohesion, where both are determined by the number of dependencies divided by the number of all possible dependencies.
While the results of the TTC16 show that improving the \ac{CRA} index does indeed lead to class diagrams with improved coupling and cohesion, it is difficult to determine the coupling's and cohesion's current state from the \ac{CRA} index alone, which is (usually) not an integer.
In general, our approach can only be combined with the \ac{CRA} index under severe drawbacks.
Using our custom metric, an impairment/repair will always have a constant value, e.g., $-1$ for an impairment and $+1$ for a repair.
While performing a refactoring step may inflict new impairments or lead to some repairs, we can track and maintain these changes incrementally to find the next best refactoring step in an efficient way.
The \ac{CRA} index, however, does not come with such a constant value for each impairment/repair, and depending on the previous former actions, an impairment/repair may have a different value due to the \ac{CRA} index' nature of incorporating all dependencies in a class diagram.
This means that, after each refactoring step, we would need to reevaluate all other refactoring steps anew to determine their actual effect on the \ac{CRA} index.
As this does not align with our concept that leverages incremental techniques, we chose this custom metric, which was designed to mimic the \ac{CRA} index to some degree.
However, there are also technical reasons related to \ac{GIPS} for not using the \ac{CRA} index, which will be discussed at the end of this chapter.

In the following subsections, we will discuss the modifications to support the TTC16 test data and briefly present the implementation based on \ac{GIPS}.
We then compare the performance of the two approaches in terms of runtime and quality, and discuss how our custom metric relates to the \ac{CRA} index.

\begin{figure}
	\centering
	\includegraphics[scale = 0.8]{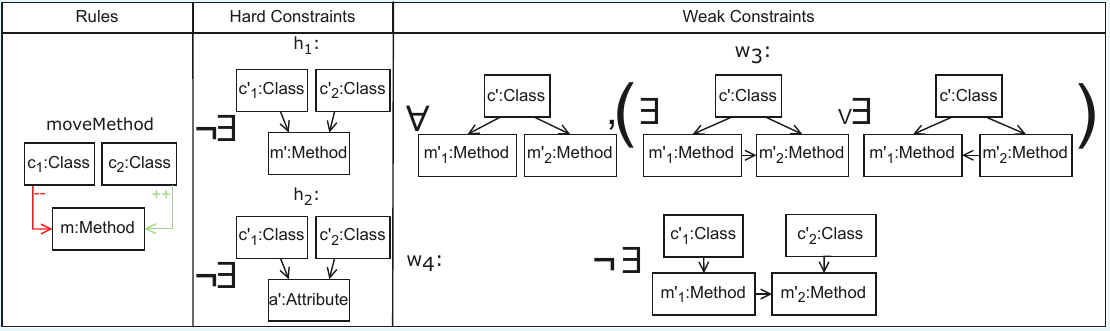}
	\caption{Rules and constraints for the extended evaluation}
	\label{fig:example:rule+constraints_extended_eval}
\end{figure}

\begin{figure}
	\includegraphics[scale = 0.8]{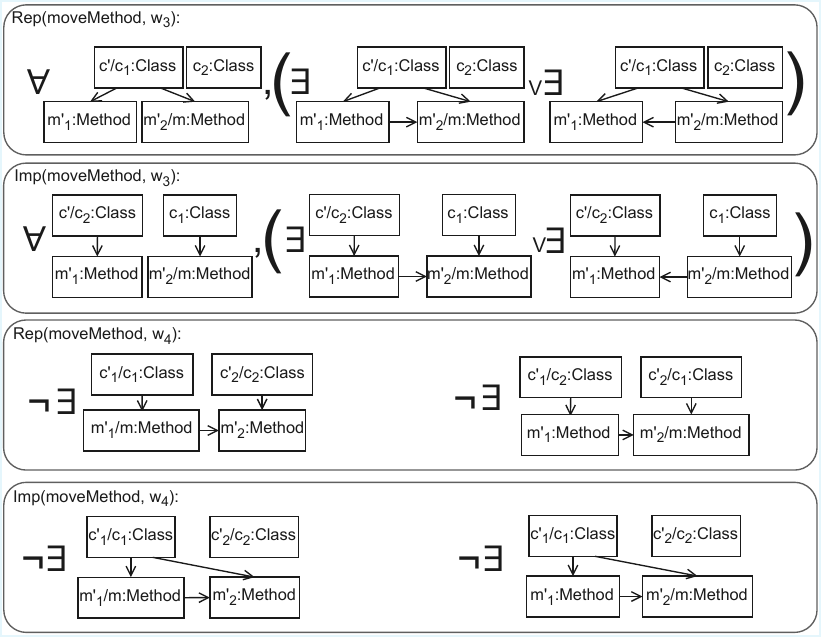}
	\caption{Repair- and impairment-indicating application conditions for the rule \textsf{moveMethod} w.r.t.~the constraints $\textsf{w}_3$ and $\textsf{w}_4$}\label{fig:acs_extended_eval}
\end{figure}

\subsubsection{Reducing violations with application conditions}
Compared to our evaluation in \cref{sec:scalabilityAndEffectiveness}, we have changed the constraints and thus the application conditions derived from them.
In particular, we use the refactoring rule \textsf{moveMethod}, assume that the hard constraint $\textsf{h}_1$ is always satisfied, and consider the weak constraints $\textsf{w}_3$ and $\textsf{w}_4$ shown in \cref{fig:example:rule+constraints_extended_eval}.
These constraints model good code patterns that should be satisfied in any well-designed class diagram.
Constraint $\textsf{w}_3$ requires that for any pair of methods contained in the same class, one method uses the other, or vice versa; $\textsf{w}_4$ requires that any pair of methods so that one of the methods uses the other, both are contained in the same class.
In total, we have four application conditions for constraint repair and another four for constraint impairment (cf.~\cref{fig:resulting_application_conditions}, \cref{fig:impair_indicating_application_conditions}, and \cref{fig:acs_extended_eval}).

The reason for this is that dependencies between two methods have a different edge type than those between a method and an attribute.

Another change to the experiments in Section~\ref{sec:scalabilityAndEffectiveness} is that we run the experiments multiple times, and additionally allow a certain amount of refactorings at the beginning of each run, to impair more than to actually repair.
In this way, we try to find different (local) optima and pick the best one.

Since there are no classes at the beginning, we create one for each feature.
We then start refactoring the model step by step, moving methods and attributes from one class to another until we can no longer improve with the next step.

\subsubsection{Finding an optimal solution with GIPS}

To find an optimal solution to our example problem, we use the \ac{GIPS} framework~\cite{gipsGCM2022}.
\ac{GIPS} combines the \ac{GT} framework with \ac{ILP} optimization techniques in order to optimize a given input (graph) model.
Therefore, a custom \ac{DSL} called \ac{GIPSL} has been developed that allows \ac{ILP} constraints to be expressed in a UML/OCL-like fashion and which extends the eMoflon::IBeX-GT language that provides the ability to specify \ac{GT} rules and graph patterns.
In essence, the \ac{GIPSL} specification contains a set of \ac{GT} rules with corresponding \ac{ILP} constraints and an objective, all of which refer to a given metamodel.
Based on this specification, the \ac{GIPS} framework generates an \ac{ILP} problem generator that is able to translate any (graph) model that conforms to the metamodel into an \ac{ILP} optimization problem.
For a given input model, the \ac{ILP} solver determines a subset of all \ac{GT} rule matches for which the corresponding \ac{GT} rules should be applied.
The application of all selected \ac{GT} rule matches then optimizes the model according to the objective goal defined in the \ac{GIPSL} specification.
Unlike using only \ac{GT}, this approach provides the ability to specify global constraints on the model and \ac{GT} rule applications.
This allows for inter-rule constraints, for example, if \textit{rule A} is executed on a specific match, \textit{rule B} must also be executed on the same (or another) match.

For our solution to the \ac{CRA} assignment problem, we generate a number of empty classes in the model to which the methods and attributes can be assigned.
We then use \ac{GIPSL}\footnote{\url{https://github.com/Echtzeitsysteme/gips-examples/blob/main/architecture.cra.gipssolution/src/architecture/cra/gipssolution/Model.gipsl}} to model the following behavior.
For the assignment itself, our specification contains a set of \ac{GT} rules that create an assignment edge between a method/attribute and a class.
In addition, we need a set of constraints to ensure that a computed solution is valid.
Therefore, our \ac{GIPSL} specification contains two constraints that ensure that each method and attribute is assigned to exactly one class.
For the optimization part, we count the number of \ac{CRA} violations that a specific solution would produce if all corresponding \ac{GT} rule matches were applied.
This sum of violations forms our objective function which motivates the solver to choose the set of assignments that minimizes the number of violations globally.
After the \ac{ILP} solver has determined the optimal subset of \ac{GT} rule matches to apply, we execute those rules and obtain the solution to the assignment problem.

\subsubsection{Comparison}
In the following, we will compare our approach with the \ac{GIPS} solution based on up to six input models of varying complexity degree taken from the TTC16 \ac{CRA} case study.
These models consist only of features (half methods and half attributes) with predefined dependencies between them.
From model A to F, they come with 9, 18, 35, 80, 160, and 400 features, respectively.
As a first research question, we want to know \emph{how our approach scales compared to another solution that solves the same problem using other optimization techniques} \textbf{(RQ1)}.
Second, while \ac{GIPS} will always find the global optimum as long as it terminates, this is not necessarily true for our approach.
Therefore, we want to investigate whether, \emph{for the models for which we have a solution from \ac{GIPS}, our approach is able to find the same solution} \textbf{(RQ2)} and \emph{how confident we can be in finding this solution with each run} \textbf{(RQ3)}.

To answer these three questions, we evaluated the \ac{CRA} scenarios on a workstation equipped with an AMD EPYC 7763 with 128 threads and \SI{256}{\giga\byte} RAM using Debian 12, OpenJDK 21.0.4, and Gurobi Optimizer 11.0.3.
For both approaches, we measured the total time to compute a solution by repeating each scenario 10 times with \ac{GIPS} and 200 times with our approach (AC).
\Cref{fig:example:ext_eval} shows the total time on the left, where we can immediately see that only the smaller three models could be processed with \ac{GIPS}, after which the memory consumption exceeded \SI{256}{\giga\byte}.
The time increases from \SI{2.4}{\second} for A, to \SI{21.3}{\second} for B, and to \SI{613.8}{\second} for C, with a standard deviation of \SI{0.04}{\second}, \SI{2.6}{\second}, and \SI{226.1}{\second}, respectively.
In comparison, using application conditions (AC), we see a less steep increase in the runtime, where it takes \SI{132}{\second} on average to find a solution with a standard deviation of less than \SI{2.5}{\percent} for all scenarios.
For \ac{GIPS}, doubling the number of features increases the time by about 25-30 times, while using ACs triples/quadruples the time \textbf{(RQ1)}.

\Cref{fig:example:ext_eval} shows the results of AC and GIPS on the right, showing the average number of violations for each scenario.
Obviously, \ac{GIPS} always returns the optimal solution to the problem, which is 7 remaining violations for A, 17 violations for B, and 73 for C.
Interestingly, using AC and picking the best solution from the 200 runs, we always got the globally optimal solution with the same number of remaining violations left \textbf{(RQ2)}.
Moreover, the standard deviation for the average number of violations never exceeded 2.3 violations, which means that even for larger \ac{CRA} models, we usually get solutions of a similar quality with respect to our custom metric \textbf{(RQ3)}.
However, we do not know whether our solutions for D, E, and F are actually globally optimal, because \ac{GIPS} could not compute this anymore due to the limited amount of memory of our system.

\begin{figure}
	\centering
	\includegraphics[width=1\textwidth, clip]{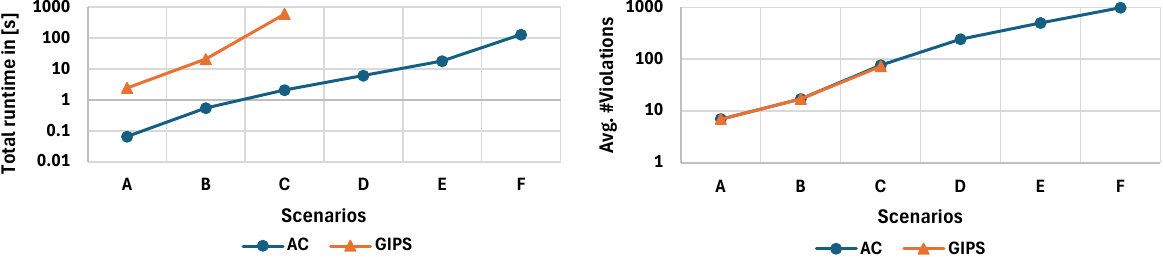}
	\caption{Average performance measurements for increasing model size (left) and average number of total violations (right)}
	\label{fig:example:ext_eval}
\end{figure}

\subsubsection{Threats to validity}
In general, the same threats to validity apply as in Section~\ref{section:evaluation:threads}.
However, we changed the greedy algorithm of picking the most promising rule in each run with an alternative that can also pick rules that repair less than they impair at the beginning.

In addition, we want to discuss our reasons for not using the \ac{CRA} index once more.
While it is true that our custom metric is more in line with the idea of using application conditions, we want to emphasise that it would be difficult to actually express incremental changes in the \ac{CRA} index using our approach.
The \ac{CRA} index is the sum of coupling and cohesion values that are expressed as the ratio of the number of dependencies of a class divided by the number of all possible dependencies.
This value, which is usually not an integer like our violation count, is harder to update incrementally for each refactoring step, and we would have to update it after each iteration for each available refactoring.
Also, we cannot express the \ac{CRA} index using \ac{GIPS} because the ratio contains divisors that depend on which classes have been assigned what features.
Since the divisor is not constant, but depends on the actual solution, it is not possible to express this as an \ac{ILP} problem.

\section{Related Work}
\label{chapter:relatedWork}
The following discusses related work on graph constraint checking and different approaches to graph repair. It also covers ranking of model repair steps and solving the CRA problem.

\paragraph{Graph constraints and application conditions.}
Habel and Pennmann~\cite{HabelP09} introduced the original process for generating application conditions from nested graph constraints,
which are consistency-guaranteeing, meaning that applying a rule with such conditions is guaranteed to produce a graph that is consistent with respect to the given (hard) constraints.
In our paper, we extend the binary case of satisfying constraints by ranking rule applications based on how many constraint violations they add or remove w.r.t.~a set of weak constraints.

Kosiol et al.~\cite{KosiolSTZ22} also count violations but only consider one type of constraint (hard constraints).
While they consider constraints in alternating quantifier normal form (i.e., constraints that do not use conjunction or disjunction), our approach supports arbitrary constraints.
In~\cite{KosiolSTZ22}, application conditions are constructed to make rule applications consistency-sustaining, but there is no such construction for consistency-improving rule applications.
Moreover, the application conditions are consistency-sustaining in the strict sense that no new constraint violations are allowed.

Since the resulting set of application conditions can be large and thus expensive to evaluate, Nassar et al.~\cite{NKAT20} showed that some subconditions can be filtered if they check for cases that cannot occur.
We also filter the resulting application conditions, but based on a set of additional hard constraints, e.g., by filtering out conditions that check for features contained in multiple classes simultaneously.

\emph{In contrast to all the literature mentioned above, we construct application conditions to rank rule applications along their potential to improve consistency.
This means that rule applications are not blocked, even if some application conditions are violated.}
This approach can be used effectively in a transformation engine like eMoflon, which supports the incremental computation of rule matches.

\paragraph{Incremental graph constraint checking.}
In~\cite{Abdeen14}, constraints are also checked by incremental graph pattern matching.
This paper distinguishes between well-formedness and ill-formedness constraints, the former of which are desirable and the latter of which are to be avoided.
Using genetic algorithms with a set of graph-modifying rules as mutators, a graph is searched for that maximises the number of occurrences of well-formedness constraints minus the number of occurrences of ill-formedness constraints.
\emph{Compared to our approach, the graph is changed to track consistency and detect violations and repairs, whereas we use a look-ahead to plan the next steps.}

\paragraph{Graph repair.}
Habel and Sandmann~\cite{DBLP:conf/staf/HabelS18, DBLP:journals/corr/abs-1912-09610, DBLP:journals/corr/abs-2012-01656, Sandmann21a} propose a graph repair approach based on constraints in alternating quantifier normal form (i.e., constraints that do not use conjunction or disjunction).
They showed that their graph repair algorithm produces a consistent graph for constraints with a nesting level less than or equal to $2$, or that end with a condition of the form $\exists(P, \true)$.
The algorithm derives rules from the constraints and reuses existing graph elements wherever possible.
For instance, a violation of $\textsf{w}_1$ (i.e., a pair of methods contained within the same class that do not depend on a common attribute within that class) would be repaired by creating a dependency between both methods and an arbitrary attribute contained within the same class.
However, this approach is not suitable for applications similar to the \ac{CRA} problem, since dependencies between features cannot be created or deleted without adapting the underlying system.
\emph{In contrast to our approach, theirs cannot work with arbitrary rules; instead, the provided rules must instead be equivalent to those derived from constraints.}

The graph repair approach presented in~\cite{DBLP:journals/isse/LauerKT24} also supports constraints in alternating quantifier normal form, but it uses user-specified rules.
Unlike the algorithm of Habel and Sandmann, this approach attempts to create missing structures rather than reusing already existing graph elements.
For example, a violation of $\textsf{w}_1$ would be repaired by creating a new attribute contained in the class, so that both methods depend on it.
To support multiple constraints, the authors introduce the concept of `circular conflict-freeness', where the algorithm searches for an ordering of the considered constraints, $c_1, c_2,\ldots, c_n$, such that repairing constraint $c_i$ does not impair any constraint $c_j$, with $j < i$.
They showed that a graph can be repaired if such an ordering can be found.
The algorithm then repairs the constraints in this order, eliminating the need for backtracking, i.e., without repairing constraints that have already been repaired.
However, the existence of such an ordering depends heavily on the provided repair rules, which restricts the applicability of this approach.
Such an ordering does not exist for the constraints $\textsf{w}_1$ and $\textsf{w}_2$ and the rules \textsf{moveAttribute} and \textsf{moveMethod} in our running example, since both rules can impair $\textsf{w}_1$ while repairing $\textsf{w}_2$ and vice versa.
\emph{Therefore, our approach is applicable to a wider range of scenarios, since it supports arbitrary constraints and repair rules.}

In~\cite{nassar2017rule}, the authors introduce a graph repair approach for multiplicities, which can be modelled as constraints of the form $\exists(P,\true)$ and $\forall(P, \false)$.
\emph{In contrast to our approach, which supports arbitrary conditions, the set of supported constraints is highly restricted.
This makes it unsuitable for complex settings, such as the \ac{CRA} problem.}

In~\cite{DBLP:journals/sttt/SchneiderLO21}, the authors introduce a logic-based, incremental approach to graph repair that can be used to repair arbitrary nested graph constraints.
Given an input model, the approach computes multiple least-changing graph repairs.
These repairs are least-changing in the sense that the most recent transformation that introduced an impairment cannot be reverted.
For example, if $\textsf{w}_2$ is impaired by moving a method that depends on an attribute contained in the same class to another class, this approach will not provide a repair that moves the method back into the original class.
In general, the approach computes the minimal graphs that satisfy the constraint and contain the input graph, preserving the actions of the most recent transformation that has introduced an impairment.
\emph{Compared to our approach, there is no support for using user-specified transformation rules, which can be used to restrict the ways in which a graph can be modelled.}
For example, applying the algorithm presented in~\cite{DBLP:journals/sttt/SchneiderLO21} to the \ac{CRA} problem may result in graphs where violations of $\textsf{w}_2$ have been repaired by removing an attribute from its class (i.e., by deleting the edge from a class to an attribute) without assigning it to another class.
While the result set may contain a graph $H$ that can be constructed by using the provided transformation rules, verifying the existence of a sequence of transformations $G \Rightarrow \ldots \Rightarrow H$ (where $G$ is the graph to be repaired) by applying only the provided rules is difficult.

\paragraph{Ranking in model repair.}
A similar ranking approach for model repair is presented in~\cite{KKE19}, where impairments are identified as negative side effects and repairs as positive side effects of model repair sequences.
All constraints have equal priority.
Given a model with a set of violations, all possible repairs are computed and ranked according to their side effects.
A repair consists of a sequence of repair actions of limited length, each of which repairs only a single model element or a single property of an element.
Rather than being determined in advance,  the ranking is established by first executing all computed repair action sequences and then observing their effects on the number of violated constraints.
{\em In contrast, our approach supports arbitrary repair rules and allows multiple repair actions to be performed in one transformation step.
Furthermore, we determine the positive and negative effects of all options for applying a given repair rule to all (weak) constraints statically, i.e. without altering the (model) graph for this purpose.}

\section{Conclusion}
\label{chapter:conclusion}
In this paper, we introduce a new dynamic analysis approach that ranks rule matches based on their potential for graph repair, which supports arbitrary nested graph constraints.
This potential is computed using application conditions that are automatically derived from a set of nested graph constraints.
While some of these conditions indicate repair steps, others detect impairments.
We formally show that the gain in consistency of a rule application can indeed be characterized by the difference between the number of violations of repair-indicating and impairment-indicating application conditions.
We illustrate and evaluate our approach in the context of the well-known \ac{CRA} problem and show that even for a worst-case scenario, the performance scales reasonably well.
An initial evaluation also shows that using our ranking approach within a greedy algorithm can yield well-acceptable optimization results.
In the future, we want to fully automate the ranking of graph transformations based on an automated construction of repair-indicating and impairment-indicating application conditions.
To this end, we plan to combine our approach with the work presented in~\cite{DBLP:journals/isse/LauerKT24} and~\cite{DBLP:conf/gg/LauerKT25}, with the goal of optimising the construction and use of application conditions even further.
In addition, we will investigate different strategies besides greedy-based algorithms for optimization of graphs in different scenarios.
To further strengthen our ranking approach, which has a look-ahead of 1, it may be advantageous to combine several repair rules into a larger one by composing concurrent rules~\cite{EH83,EEPT06}.

\bibliographystyle{alphaurl}
\bibliography{bibliography}

\newcommand{\etalchar}[1]{$^{#1}$}
\begin{thebibliography}{NKAT20}

\bibitem[AVS{\etalchar{+}}14]{Abdeen14}
Hani Abdeen, D{\'{a}}niel Varr{\'{o}}, Houari~A. Sahraoui,
  Andr{\'{a}}s~Szabolcs Nagy, Csaba Debreceni, {\'{A}}bel Heged{\"{u}}s, and
  {\'{A}}kos Horv{\'{a}}th.
\newblock Multi-objective optimization in rule-based design space exploration.
\newblock In Ivica Crnkovic, Marsha Chechik, and Paul Gr{\"{u}}nbacher,
  editors, {\em {ACM/IEEE} International Conference on Automated Software
  Engineering, {ASE} '14, Vasteras, Sweden - September 15 - 19, 2014}, pages
  289--300. {ACM}, 2014.
\newblock \href {https://doi.org/10.1145/2642937.2643005}
  {\path{doi:10.1145/2642937.2643005}}.

\bibitem[BBL10]{BBL10}
Michael Bowman, Lionel~C Briand, and Yvan Labiche.
\newblock Solving the class responsibility assignment problem in
  object-oriented analysis with multi-objective genetic algorithms.
\newblock {\em IEEE Transactions on Software Engineering}, 36(6):817--837,
  2010.
\newblock \href {https://doi.org/10.1109/TSE.2010.70}
  {\path{doi:10.1109/TSE.2010.70}}.

\bibitem[EEPT06]{EEPT06}
Hartmut Ehrig, Karsten Ehrig, Ulrike Prange, and Gabriele Taentzer.
\newblock {\em Fundamentals of Algebraic Graph Transformation}.
\newblock Monographs in Theoretical Computer Science. An {EATCS} Series.
  Springer, 2006.
\newblock \href {https://doi.org/10.1007/3-540-31188-2}
  {\path{doi:10.1007/3-540-31188-2}}.

\bibitem[EH83]{EH83}
Hartmut Ehrig and Annegret Habel.
\newblock Concurrent transformations of graphs and relational structures.
\newblock In Manfred Nagl and J{\"{u}}rgen Perl, editors, {\em Proceedings of
  the {WG} '83, International Workshop on Graphtheoretic Concepts in Computer
  Science}, pages 76--88. Universit{\"{a}}tsverlag Rudolf Trauner, Linz, 1983.

\bibitem[Ehr78]{Ehrig78}
Hartmut Ehrig.
\newblock Introduction to the algebraic theory of graph grammars {(A} survey).
\newblock In Volker Claus, Hartmut Ehrig, and Grzegorz Rozenberg, editors, {\em
  Graph-Grammars and Their Application to Computer Science and Biology,
  International Workshop, Bad Honnef, Germany, October 30 - November 3, 1978},
  volume~73 of {\em Lecture Notes in Computer Science}, pages 1--69. Springer,
  1978.
\newblock \href {https://doi.org/10.1007/BFB0025714}
  {\path{doi:10.1007/BFB0025714}}.

\bibitem[EKS22]{gipsGCM2022}
Sebastian Ehmes, Maximilian Kratz, and Andy Sch\"urr.
\newblock Graph-based specification and automated construction of ilp problems.
\newblock In {\em {\rm Proceedings of the Thirteenth International Workshop on}
  Graph Computation Models, {\rm Nantes, France, 6th July 2022}}, volume 374 of
  {\em Electronic Proceedings in Theoretical Computer Science}, pages 3--22.
  Open Publishing Association, 2022.
\newblock \href {https://doi.org/10.4204/EPTCS.374.3}
  {\path{doi:10.4204/EPTCS.374.3}}.

\bibitem[FLST24]{fritsche2024using}
Lars Fritsche, Alexander Lauer, Andy Sch{\"{u}}rr, and Gabriele Taentzer.
\newblock Using application conditions to rank graph transformations for graph
  repair.
\newblock In Russ Harmer and Jens Kosiol, editors, {\em Graph Transformation -
  17th International Conference, {ICGT} 2024, Held as Part of {STAF} 2024,
  Enschede, The Netherlands, July 10-11, 2024, Proceedings}, volume 14774 of
  {\em Lecture Notes in Computer Science}, pages 138--157. Springer, 2024.
\newblock \href {https://doi.org/10.1007/978-3-031-64285-2\_8}
  {\path{doi:10.1007/978-3-031-64285-2\_8}}.

\bibitem[FTW16]{DBLP:conf/staf/FleckTW16}
Martin Fleck, Javier Troya, and Manuel Wimmer.
\newblock The class responsibility assignment case.
\newblock In Antonio Garc{\'{\i}}a{-}Dom{\'{\i}}nguez, Filip Krikava, and
  Louis~M. Rose, editors, {\em Proceedings of the 9th Transformation Tool
  Contest, co-located with the 2016 Software Technologies: Applications and
  Foundations {(STAF} 2016), Vienna, Austria, July 8, 2016}, volume 1758 of
  {\em {CEUR} Workshop Proceedings}, pages 1--8. CEUR-WS.org, 2016.
\newblock URL: \url{https://ceur-ws.org/Vol-1758/paper1.pdf}.

\bibitem[Hin16]{Hinkel16}
Georg Hinkel.
\newblock An {NMF} solution to the class responsibility assignment case.
\newblock In Antonio Garc{\'{\i}}a{-}Dom{\'{\i}}nguez, Filip Krikava, and
  Louis~M. Rose, editors, {\em Proceedings of the 9th Transformation Tool
  Contest, co-located with the 2016 Software Technologies: Applications and
  Foundations {(STAF} 2016), Vienna, Austria, July 8, 2016}, volume 1758 of
  {\em {CEUR} Workshop Proceedings}, pages 15--20. CEUR-WS.org, 2016.
\newblock URL: \url{https://ceur-ws.org/Vol-1758/paper3.pdf}.

\bibitem[HP09]{HabelP09}
Annegret Habel and Karl{-}Heinz Pennemann.
\newblock Correctness of high-level transformation systems relative to nested
  conditions.
\newblock {\em Mathematical Structures in Computer Science}, 19(2):245--296,
  2009.
\newblock \href {https://doi.org/10.1017/S0960129508007202}
  {\path{doi:10.1017/S0960129508007202}}.

\bibitem[HS18]{DBLP:conf/staf/HabelS18}
Annegret Habel and Christian Sandmann.
\newblock Graph repair by graph programs.
\newblock In Manuel Mazzara, Iulian Ober, and Gwen Sala{\"{u}}n, editors, {\em
  Software Technologies: Applications and Foundations - {STAF} 2018 Collocated
  Workshops, Toulouse, France, June 25-29, 2018, Revised Selected Papers},
  volume 11176 of {\em Lecture Notes in Computer Science}, pages 431--446.
  Springer, 2018.
\newblock \href {https://doi.org/10.1007/978-3-030-04771-9_31}
  {\path{doi:10.1007/978-3-030-04771-9_31}}.

\bibitem[HT20]{HT20}
Reiko Heckel and Gabriele Taentzer.
\newblock {\em Graph Transformation for Software Engineers - With Applications
  to Model-Based Development and Domain-Specific Language Engineering}.
\newblock Springer, 2020.
\newblock \href {https://doi.org/10.1007/978-3-030-43916-3}
  {\path{doi:10.1007/978-3-030-43916-3}}.

\bibitem[KKE19]{KKE19}
Djamel~Eddine Khelladi, Roland Kretschmer, and Alexander Egyed.
\newblock Detecting and exploring side effects when repairing model
  inconsistencies.
\newblock In {\em Proceedings of the 12th ACM SIGPLAN international conference
  on software language engineering}, pages 113--126. {ACM}, 2019.
\newblock \href {https://doi.org/10.1145/3357766.3359546}
  {\path{doi:10.1145/3357766.3359546}}.

\bibitem[KSTZ22]{KosiolSTZ22}
Jens Kosiol, Daniel Str{\"{u}}ber, Gabriele Taentzer, and Steffen Zschaler.
\newblock Sustaining and improving graduated graph consistency: {A} static
  analysis of graph transformations.
\newblock {\em Science of Computer Programming}, 214:102729, 2022.
\newblock \href {https://doi.org/10.1016/J.SCICO.2021.102729}
  {\path{doi:10.1016/J.SCICO.2021.102729}}.

\bibitem[LKT24]{DBLP:journals/isse/LauerKT24}
Alexander Lauer, Jens Kosiol, and Gabriele Taentzer.
\newblock Empowering model repair: a rule-based approach to graph repair
  without side effects - extended version.
\newblock {\em Innovations in Systems and Software Engineering},
  20(4):597--618, 2024.
\newblock \href {https://doi.org/10.1007/S11334-024-00587-W}
  {\path{doi:10.1007/S11334-024-00587-W}}.

\bibitem[LKT25]{DBLP:conf/gg/LauerKT25}
Alexander Lauer, Jens Kosiol, and Gabriele Taentzer.
\newblock Granular conflict analysis for transformation rules with application
  conditions.
\newblock In J{\"{o}}rg Endrullis and Matthias Tichy, editors, {\em Graph
  Transformation - 18th International Conference, {ICGT} 2025, Held as Part of
  {STAF} 2025, Koblenz, Germany, June 11-12, 2025, Proceedings}, volume 15720
  of {\em Lecture Notes in Computer Science}, pages 63--90. Springer, 2025.
\newblock \href {https://doi.org/10.1007/978-3-031-94706-3_4}
  {\path{doi:10.1007/978-3-031-94706-3_4}}.

\bibitem[MJ14]{masoud2014clustering}
Hamid Masoud and Saeed Jalili.
\newblock A clustering-based model for class responsibility assignment problem
  in object-oriented analysis.
\newblock {\em Journal of Systems and Software}, 93:110--131, 2014.
\newblock \href {https://doi.org/10.1016/j.jss.2014.02.053}
  {\path{doi:10.1016/j.jss.2014.02.053}}.

\bibitem[NKAT20]{NKAT20}
Nebras Nassar, Jens Kosiol, Thorsten Arendt, and Gabriele Taentzer.
\newblock Constructing optimized constraint-preserving application conditions
  for model transformation rules.
\newblock {\em Journal of Logical and Algebraic Methods in Programming},
  114:100564, 2020.
\newblock \href {https://doi.org/10.1016/J.JLAMP.2020.100564}
  {\path{doi:10.1016/J.JLAMP.2020.100564}}.

\bibitem[NKR17]{nassar2017rule}
Nebras Nassar, Jens Kosiol, and Hendrik Radke.
\newblock Rule-based repair of emf models: formalization and correctness proof.
\newblock In {\em Electronic Pre-Proceedings International Workshop on Graph
  Computation Models}, 2017.

\bibitem[Plu05]{Plump05}
Detlef Plump.
\newblock Confluence of graph transformation revisited.
\newblock In Aart Middeldorp, Vincent van Oostrom, Femke van Raamsdonk, and
  Roel~C. de~Vrijer, editors, {\em Processes, Terms and Cycles: Steps on the
  Road to Infinity, Essays Dedicated to Jan Willem Klop, on the Occasion of His
  60th Birthday}, volume 3838 of {\em Lecture Notes in Computer Science}, pages
  280--308. Springer, 2005.
\newblock \href {https://doi.org/10.1007/11601548_16}
  {\path{doi:10.1007/11601548_16}}.

\bibitem[RAB{\etalchar{+}}18]{RadkeABHT18}
Hendrik Radke, Thorsten Arendt, Jan~Steffen Becker, Annegret Habel, and
  Gabriele Taentzer.
\newblock Translating essential {OCL} invariants to nested graph constraints
  for generating instances of meta-models.
\newblock {\em Science of Computer Programming}, 152:38--62, 2018.
\newblock \href {https://doi.org/10.1016/J.SCICO.2017.08.006}
  {\path{doi:10.1016/J.SCICO.2017.08.006}}.

\bibitem[Ren04]{Rensink04}
Arend Rensink.
\newblock Representing first-order logic using graphs.
\newblock In Hartmut Ehrig, Gregor Engels, Francesco Parisi{-}Presicce, and
  Grzegorz Rozenberg, editors, {\em Graph Transformations, Second International
  Conference, {ICGT} 2004, Rome, Italy, September 28 - October 2, 2004,
  Proceedings}, volume 3256 of {\em Lecture Notes in Computer Science}, pages
  319--335. Springer, 2004.
\newblock \href {https://doi.org/10.1007/978-3-540-30203-2_23}
  {\path{doi:10.1007/978-3-540-30203-2_23}}.

\bibitem[San20]{DBLP:journals/corr/abs-2012-01656}
Christian Sandmann.
\newblock Graph repair and its application to meta-modeling.
\newblock In Berthold Hoffmann and Mark Minas, editors, {\em Proceedings of the
  Eleventh International Workshop on Graph Computation Models, GCM@STAF 2020,
  Online-Workshop, 24th June 2020}, volume 330 of {\em {EPTCS}}, pages 13--34,
  2020.
\newblock \href {https://doi.org/10.4204/EPTCS.330.2}
  {\path{doi:10.4204/EPTCS.330.2}}.

\bibitem[San21]{Sandmann21a}
Christian Sandmann.
\newblock {\em A Theory on Graph Generation and Graph Repair with Application
  to Meta-Modeling}.
\newblock PhD thesis, University of Oldenburg, 2021.
\newblock URL:
  \url{http://uol.de/f/2/dept/informatik/download/Promotionen/Sandmann_Dissertation.pdf}.

\bibitem[SH19]{DBLP:journals/corr/abs-1912-09610}
Christian Sandmann and Annegret Habel.
\newblock Rule-based graph repair.
\newblock In Rachid Echahed and Detlef Plump, editors, {\em Proceedings Tenth
  International Workshop on Graph Computation Models, GCM@STAF 2019, Eindhoven,
  The Netherlands, 17th July 2019}, volume 309 of {\em {EPTCS}}, pages 87--104,
  2019.
\newblock \href {https://doi.org/10.4204/EPTCS.309.5}
  {\path{doi:10.4204/EPTCS.309.5}}.

\bibitem[SLO21]{DBLP:journals/sttt/SchneiderLO21}
Sven Schneider, Leen Lambers, and Fernando Orejas.
\newblock A logic-based incremental approach to graph repair featuring delta
  preservation.
\newblock {\em International Journal on Software Tools for Technology
  Transfer}, 23(3):369--410, 2021.
\newblock \href {https://doi.org/10.1007/S10009-020-00584-X}
  {\path{doi:10.1007/S10009-020-00584-X}}.

\end{thebibliography}
\newpage
\appendix
\section{Double-pushout approach}\label{app:dpo}

In this section, we briefly introduce the double-pushout approach, since we will use some properties of pushouts and pullbacks in our proofs~\cite{EEPT06}.
We start by introducing pushouts and pullbacks.
Intuitively, for graphs, a \emph{pushout} is the smallest gluing of two graphs $B$ and $C$ at some interface $A$, i.e., we construct a graph $D$ such that $A$, $B$ and $C$ are subgraphs of $D$ and the embeddings of $B$ and $C$ overlap only at the embedding of $A$.
Intuitively, when constructing a pullback of the graphs $B$, $C$, and $D$ so that $B$ and $C$ are subgraphs of $D$, we are searching for the largest overlap $A$ of $B$ and $C$, which is a subgraph of $D$.

\begin{figure}
	\centering
	\begin{tikzpicture}[scale = 1]
		\node(A1) at (0,0) {$A$};
		\node(B1) at (2,0) {$B$};
		\node(C1) at (0,-2) {$C$};
		\node(D1) at (2,-2) {$D$};
		\node(X1) at (3.5,-3.5) {$X$};
		\node(label1) at (1,-1) {$(PO)$};

		\draw [-stealth] (A1) edge node [fill=white] {$f$} (B1);
		\draw [-stealth] (A1) edge node [fill=white] {$g$} (C1);
		\draw [-stealth] (C1) edge node [fill=white] {$f'$} (D1);
		\draw [-stealth] (B1) edge node [fill=white] {$g'$} (D1);
		\draw [-stealth] (C1) edge node [fill=white] {$k$} (X1);
		\draw [-stealth] (D1) edge node [fill=white] {$x$} (X1);
		\draw [-stealth] (B1) edge node [fill=white] {$h$} (X1);

		\node(A2) at (5,0) {$A$};
		\node(B2) at (7,0) {$B$};
		\node(C2) at (5,-2) {$C$};
		\node(D2) at (7,-2) {$D$};
		\node(X2) at (3.5,1.5) {$X$};
		\node(label2) at (6,-1) {$(PB)$};

		\draw [-stealth] (A2) edge node [fill=white] {$f'$} (B2);
		\draw [-stealth] (A2) edge node [fill=white] {$g'$} (C2);
		\draw [-stealth] (C2) edge node [fill=white] {$f$} (D2);
		\draw [-stealth] (B2) edge node [fill=white] {$g$} (D2);
		\draw [-stealth] (X2) edge node [fill=white] {$k$} (C2);
		\draw [-stealth] (X2) edge node [fill=white] {$x$} (A2);
		\draw [-stealth] (X2) edge node [fill=white] {$h$} (B2);
	\end{tikzpicture}
	\caption{Pushout (on the left) and pullback (on the right)}
	\label{fig:pushout}
\end{figure}

\begin{defi}[Pushout and pullback~\cite{EEPT06}]
	Given the morphisms $f \colon A \to B$ and $g \colon A \to C$ as shown in the left square in \cref{fig:pushout}. The square $(PO)$ is a \emph{pushout} if it is commutative, i.e., $g' \circ f = f' \circ g$ so that the following universal property is satisfied: For all graphs $X$ and morphisms $h \colon B \to X$ and $k \colon C \to X$ with $k \circ g = h \circ f$, there is a unique morphism $x \colon D \to X$ such that $x \circ g' = h$ and $x \circ f' = k$.

	Given the morphisms $f \colon C \to D$ and $g \colon B \to D$ as shown in the right square in \cref{fig:pushout}. The square $(PB)$ is a \emph{pullback} if it is commutative, i.e., $g' \circ f = f' \circ g$ so that the following universal property is satisfied: For all graphs $X$ and morphisms $h \colon X \to B$ and $k \colon X \to C$ with $f \circ k = g \circ h$, there is a unique morphism $x \colon X \to A$ such that $f' \circ x = h$ and $g' \circ x = k$.
\end{defi}
Note that in the category graphs, a pushout is also a pullback.
Before continuing with the definition of transformations using pushouts, we state an important property of pushouts and pullbacks used in our proofs.

\begin{figure}
	\centering
	\begin{tikzpicture}[scale = 1]
		\node(A) at (0,0) {$A$};
		\node(B) at (2,0) {$B$};
		\node(C) at (0,-2) {$C$};
		\node(D) at (2,-2) {$D$};
		\node(E) at (4,0) {$E$};
		\node(F) at (4,-2) {$F$};
		\node(label1) at (1,-1) {$(1)$};
		\node(label2) at (3,-1) {$(2)$};

		\draw [-stealth] (A) edge node  {} (B);
		\draw [-stealth] (A) edge node  {} (C);
		\draw [-stealth] (C) edge node  {} (D);
		\draw [-stealth] (B) edge node  {} (D);
		\draw [-stealth] (B) edge node  {} (E);
		\draw [-stealth] (D) edge node  {} (F);
		\draw [-stealth] (E) edge node  {} (F);
	\end{tikzpicture}
	\caption{Pushout and Pullback decomposition}
	\label{fig:po_pb_decomposition}
\end{figure}

\begin{lem}[Properties of pushouts and pullbacks~\cite{EEPT06}]
	For pushouts, the following properties hold:
	\begin{enumerate}
		\item \emph{Uniqueness:} If the square $(1)$ shown in \cref{fig:po_pb_decomposition} is a pushout, then the graph $D$ is unique up to isomorphism.
		\item \emph{Pushout composition:} If $(1)$ and $(2)$ shown in \cref{fig:po_pb_decomposition} are pushouts, then $(1) + (2)$ is also a pushout.
		\item \emph{Pushout decomposition:} If $(1)$ and $(1) + (2)$ shown in \cref{fig:po_pb_decomposition} are pushouts, then $(2)$ is also a pushout.
	\end{enumerate}
	For pullbacks, the following properties hold:
	\begin{enumerate}
		\item \emph{Uniqueness:} If the square $(1)$ shown in \cref{fig:po_pb_decomposition} is a pullback, then the graph $A$ is unique up to isomorphism.
		\item \emph{Pullback composition:} If $(1)$ and $(2)$ shown in \cref{fig:po_pb_decomposition} are pullbacks, then $(1) + (2)$ is also a pullback.
		\item \emph{Pullback decomposition:} If $(2)$ and $(1) + (2)$ shown in \cref{fig:po_pb_decomposition} are pullbacks, then $(1)$ is also a pullback.
	\end{enumerate}
\end{lem}

\begin{figure}
	\centering
	\begin{tikzpicture}[scale = 1]
		\node(A) at (0,0) {$L$};
		\node(B) at (2,0) {$K$};
		\node(C) at (0,-2) {$G$};
		\node(D) at (2,-2) {$D$};
		\node(E) at (4,0) {$R$};
		\node(F) at (4,-2) {$H$};
		\node(label1) at (1,-1) {$(1)$};
		\node(label2) at (3,-1) {$(2)$};

		\draw [left hook-stealth] (B) edge node  [above]{$l$} (A);
		\draw [left hook-stealth] (A) edge node  [left]{$m$} (C);
		\draw [left hook-stealth] (D) edge node  [above]{$g$} (C);
		\draw [left hook-stealth] (B) edge node  {} (D);
		\draw [right hook-stealth] (B) edge node  [above]{$r$} (E);
		\draw [right hook-stealth] (D) edge node  [above]{$h$} (F);
		\draw [left hook-stealth] (E) edge node  [right]{$n$} (F);
	\end{tikzpicture}
	\caption{Graph transformation}
	\label{fig:graph_transformation2}
\end{figure}

\begin{defi}[Graph transformation~\cite{EEPT06}]
	Given a rule $\rho = \completeRle$ and a morphism $m \colon L \inj G$, called the \emph{match}, a \emph{transformation} $t \colon G \Longrightarrow_{\rho,m} H$ is given by the diagram shown in \cref{fig:graph_transformation2} if the squares $(1)$ and $(2)$ are pushouts.
	The rule $\rho$ is \emph{applicable at $m$} if $(1)$ is a pushout.
\end{defi}

\section{Additional formal results and proofs}\label{app:proofs}

This section contains additional formal results and proofs of the lemmas and theorems presented throughout the paper.
We start by stating some properties of the $\shift$ and $\leftshift$ operators as presented in Section~\ref{sec:preliminaries}.

\begin{lem}[Correctness of $\shift$ {\cite[Corollary 3]{HabelP09}}]\label{lemma:correctness_shift}
	Given a nested condition $c$ over a graph $C$ and a morphism $p' \colon C \inj C'$, for each morphism $p \colon C' \inj G$ it holds that
	$$p \models \shift(p',c) \iff p \circ p' \models c.$$
\end{lem}

\begin{lem}[Correctness of $\leftshift$ {\cite[Theorem 6]{HabelP09}}]\label{lemma:correctness_left}
	Given a plain rule $\rho = \completeRle$ and a condition $c$ over $R$, for each transformation $t \colon G \Longrightarrow_{\rho, m} H$ with comatch $n$ it holds that 
	$$n \models c \iff m \models \leftshift(\rho,c).$$
\end{lem}

The following lemma ensures that the $\leftshift$ operator is stable w.r.t.~the number of violations, i.e., if we have an application condition $ac$ formalised for the \ac{RHS} of a rule $\rho$, then the number of violations of that application condition w.r.t.~the comatch of a transformation is equal to the number of violations of $\leftshift(ac, \rho)$ w.r.t.~the match of this transformation. This ensures that \cref{thm:main_theorem} holds if we construct the impairment-indicating application conditions by computing the repair-indicating application conditions of the inverse rule and shift them to the \ac{LHS}.

\begin{lem}[Stability of $\leftshift$ w.r.t.~the number of violations]
	Given a graph $P$, a transformation $t \colon G \Longrightarrow_{\rho, m} H$ via a rule $\rho = \completeRle$ at a match $m \colon L \inj G$ with comatch $n \colon R \inj H$, and an application condition $ac = \forall(i_R \colon R \inj PR, d')$ constructed via some overlap $(i_R, i_P,PR)$, then
	$$\nv_n(ac) = \nv_m(\leftshift(\rho, ac)).$$
\end{lem}
\begin{proof}
	\begin{enumerate}
		\item First we want to show that $\nv_n(ac) \leq \nv_m(\leftshift(\rho, ac))$. Given a violation $n' \colon PR \inj H \in \sv_n(ac)$, i.e., $n = n' \circ i_R$ and $n' \not \models d'$. By the correctness of the $\leftshift$ operator we have $n' \not \models d' \iff m' \not\models \leftshift(\rho'^{-1}, d')$ where $\rho' = \rle{PR}{h}{D}{g}{PL}$ is the rule shown in \cref{fig:appendix_for_correctness_of_shift_and_co}. Additionally, by the definition of $\leftshift$ we have $\leftshift(\rho, ac) = \forall(i_L \colon L \inj PL, \leftshift(\rho'^{-1}, d'))$. By construction, we have $m = i' \circ i_L$ and therefore $m' \in \sv_m(\leftshift(\rho,ac))$. I.e., each violation $n' \colon PR \inj H \in \sv_n(ac)$ corresponds to a violation $m' \in \sv_m(\leftshift(\rho,ac))$.

		For two morphisms $m'_1, m_2' \in \sv_m(\leftshift(\rho, ac))$ it holds that $m_1' \circ i_L = m_2' \circ i_L$, i.e., $m'_1, m_2'$ are different in elements preserved by $t$, i.e., $m_1' \circ g \neq m_2' \circ g$. Then $\track_t \circ m_1' \circ g \neq \track_t \circ m_2' \circ g$ implies that $n_1' \circ h \neq n_2' \circ h$ for the corresponding morphisms $n'_1, n'_2 \in \sv_n(ac)$ and we have $n'_1 \neq n'_2$. So $\nv_n(ac) \leq \nv_m(\leftshift(\rho, ac)).$

		\item Analogously, we can show that each violation $m' \in \sv_m(\leftshift(\rho,ac))$ corresponds to a violation $n' \colon PR \inj H \in \sv_n(ac)$. \qedhere
	\end{enumerate}
\end{proof}

\subsection{Proofs of Section~\ref{sec_equiv_overlaps}}

\begin{figure}
	\centering
	\begin{tikzpicture}[scale = 1]
		\node(GH1) at (0,0){$GH'$};
		\node(GH2) at (0,3){$GH$};
		\node(X) at (0,1.5){$X$};

		\node(G) at (-3,1.5){$G$};
		\node(H) at (3,1.5){$H$};

		\draw[right hook-stealth](G) edge node [fill=white] {$i'_{G}$} (GH1);
		\draw[right hook-stealth](G) edge node [fill=white] {$i_{G}$} (GH2);

		\draw[left hook-stealth](H) edge node [fill = white] {$i'_{H}$} (GH1);
		\draw[left hook-stealth](H) edge node [fill = white] {$i_{H}$} (GH2);

		\draw[left hook-stealth](GH1) edge node [left]{$p'$}(X);
		\draw[right hook-stealth](GH2) edge node [left]{$p$}(X);
		\draw[right hook-stealth,bend left](GH2) edge node [right]{$\sim$}(GH1);
	\end{tikzpicture}
	\caption{Diagram for the characterisation of equivalence classes of overlaps}
	\label{fig:construction_isomorphism_2}
\end{figure}

\begin{proof}[Proof of Lemma~\ref{lem:equiv_iso_mor}]
	\begin{enumerate}
		\item Let $(i_G, i_H,GH) \cong (i'_G, i'_H,GH')$ and $p \colon GH \inj X$ be a morphism into some graph $X$, then there is an isomorphism $\sim \colon GH \inj GH'$ with $i_G' = \sim \circ i_G$ and $i_H' = \sim \circ i_H$ (see \cref{fig:construction_isomorphism_2}). Therefore, for the morphism $p':= p \circ \sim^{-1} \colon GH' \inj X$ we have
		\begin{align*}
			&p' \circ i'_G = p' \circ \sim \circ i_G = p \circ \sim^{-1} \circ \sim \circ i_G = p \circ \circ i_G \text{ and }\\
			&p' \circ i'_H = p' \circ \sim \circ i_H = p \circ \sim^{-1} \circ \sim \circ i_H = p \circ \circ i_H.
		\end{align*}
		\item Given two overlaps $(i_G, i_H,GH)$ and $(i'_G, i'_H,GH')$ and two morphisms $p\colon GH \inj X$ and $p'\colon GH' \inj X$ into some graph $X$ with
		\begin{align*}
			&p' \circ i'_G = p \circ i_G \text{ and }\\
			&p' \circ i'_H = p \circ i_H.
		\end{align*}
		We need to show that $(i_G, i_H, GH) \cong (i'_G,i'_H, GH')$, i.e., that there is an isomorphism $\sim \colon GH \inj GH'$ with $i'_G = \sim \circ i_G$ and $i'_H = \sim \circ i_H$. Since $p' \circ i'_G = p \circ i_G$ and $p' \circ i'_H = p \circ i_H$, we have $p(GH) = p'(GH')$. Therefore, the morphism $\sim := p'^{-1} \circ p \colon GH \inj GH'$ is total. Analogous, we obtain that  $ \sim^{-1} = p^{-1} \circ p' \colon GH' \inj GH$ is total and 
		\begin{align*}
			&\sim \circ \sim^{-1} = p'^{-1} \circ p \circ p^{-1} \circ p' = p'^{-1} \circ p' = \id_{GH'} \text{ and} \\
			&\sim^{-1} \circ \sim = p^{-1} \circ p' \circ p'^{-1} \circ p = p^{-1} \circ p = \id_{GH}.
		\end{align*}
		Hence, $\sim$ is an isomorphism and we have
		\begin{align*}
			p' \circ i'_G &= p \circ i_G &&\iff \\
			i'_G &= p'^{-1} \circ p \circ i_G = \sim \circ i_G &&\text{and }\\
			p' \circ i'_H &= p \circ i_H &&\iff \\
			i'_H &= p'^{-1} \circ p \circ i_H = \sim \circ i_H.
		\end{align*}
		Therefore, $(i_G, i_H,GH) \cong (i'_G, i'_H,GH')$. \qedhere
	\end{enumerate}
\end{proof}

\subsection{Proofs of Section~\ref{sec:construction_appl_cond}}

\begin{proof}[Proof of Lemma~\ref{lem:relation_impair-repair}]
	Given a transformation $t \colon G \Longrightarrow_{\rho, m} H$, its inverse transformation $t^{-1}\colon H \Longrightarrow_{\rho^{-1}, n} G$, and a constraint $c = \forall(e \colon \emptyset \inj P,d).$
	\begin{enumerate}
		\item Let $p \colon P \inj H$ be an impaired morphism w.r.t.~$t$ (i.e., $p \not \models d)$, then $p$ is either an impair of the premise or an impair of the conclusion:
		\begin{enumerate}
			\item\label{proof:equiv:1} If $p$ is an impair of the premise, there is no morphism $p' \colon P \inj G$ with $p = \track_t \circ p'$, i.e., $p$ was introduced by $t$ and is therefore destroyed by $t^{-1}$. So $\track_{t^{-1}} \circ p$ is not total and $p$ is a repaired morphism w.r.t.~$t^{-1 }$.
			\item\label{proof:equiv:2} If $p$ is an impair of the conclusion, there is a morphism $p' \colon P \inj G $ with $p = \track_t \circ p'$ and $p' \models d$, i.e., $\track_{t^{-1}} \circ p$ is total. So $p$ is a repaired morphism w.r.t.~$t^{-1 }$ since $\track_{t^{-1}} \circ p = p' \models d$.
		\end{enumerate}
		\item Let $p \colon P \inj G$ be a repaired morphism w.r.t.~$t$ (i.e., $p \not \models d)$, then $p$ is either a repair of the premise or a repair of the conclusion.
		\begin{enumerate}
			\item If $p$ is a repair of the premise, the proof is analogous to~\cref{proof:equiv:1}.
			\item If $p$ is a repair of the conclusion, the proof is analogous to~\cref{proof:equiv:2}. \qedhere
		\end{enumerate}
	\end{enumerate}
\end{proof}

\begin{proof}[Proof of Lemma~\ref{lemma:naive_filter}]
	We assume that $G \models c$ and that there is a morphism $p \colon P' \inj G$ with $p \models c'$. I.e., there is a morphism $q \colon Q' \inj G$ with $p = q \circ e'$ and $q \models d$. Since $Q' \not \models c$, there is a morphism $p' \colon P \inj Q'$ and hence, there is a morphism $q \circ p' \colon P \inj G$. It follows that $G \not \models c$, this is a contradiction.
\end{proof}

\begin{proof}[Proof of \cref{cor:naive_filter}]
	This Lemma follows immediately with the fact that $(a \vee b) \implies (a\vee c) \equiv b \implies(a \vee c)$.
\end{proof}

\subsection{Proof of Theorem~\ref{thm:main_theorem}}
In this section, we provide the proof of \cref{thm:main_theorem} and other theoretical results that are needed for the proof.
We begin by showing that the \emph{induced pre-} and \emph{post-conditions} are correct in the sense that they predict whether an occurrence of the premise satisfies the conclusion before the transformation (correctness of the pre-condition) and after the transformation (correctness of the post-condition).
For this, we first need the following result.

\begin{figure}
	\centering
	\begin{tikzpicture}[scale = 1]
		\node (G) at (-1,0){$PL$};
		\node (D1) at (-3,0){$D$};
		\node (H1) at (-5,0){$PR$};

		\draw[right hook-stealth](D1) edge node [above] {$g$}(G);
		\draw[left hook-stealth](D1) edge node [above] {$h$}(H1);

		\node (K1) at (-3,2) {$K$};
		\node (L1) at (-1,2) {$L$};
		\node (R1) at (-5,2) {$R$};

		\draw[right hook-stealth](K1) edge node [above] {$l$}(L1);
		\draw[left hook-stealth](K1) edge node [above] {$r$}(R1);

		\node (G2) at (-1,-2){$G$};
		\node (D2) at (-3,-2){$D'$};
		\node (H2) at (-5,-2){$H$};

		\draw[left hook-stealth] (D2) edge node [above]{$h'$}(H2);
		\draw[right hook-stealth] (D2) edge node [above]{$g'$}(G2);

		\node (L2) at (1,2) {$P$};

		\draw[right hook-stealth] (G) edge node [fill=white]{$m'$}(G2);
		\draw[right hook-stealth] (D1) edge node [fill=white]{$k'$}(D2);
		\draw[right hook-stealth] (H1) edge node [fill=white]{$n'$}(H2);

		\draw[dashed, -stealth, bend right=15] (L2) edge node [label={[xshift=-3em, yshift=-1.9em]$x$}]{}(D1);
		\draw[dashed, -stealth, bend right=15, gray] (L2) edge node [label={[xshift=-3em, yshift=-3.9em]$x'$}]{}(D2);

		\draw[left hook-stealth] (R1) edge node [fill=white]{$i_R$}(H1);
		\draw[left hook-stealth] (K1) edge node [fill=white]{$k$}(D1);

		\draw[right hook-stealth] (L1) edge node [fill=white]{$i_L$}(G);
		\draw[left hook-stealth] (L2) edge node [fill=white]{$i_P$}(G);

		\draw[left hook-stealth, bend right] (R1) edge node [fill=white]{$n$}(H2);
		\draw[right hook-stealth, bend left] (L1) edge node [fill=white]{$m$}(G2);

		\draw[left hook-stealth, bend left](L2) edge node [fill=white]{$p$}(G2);
		\draw[dashed, -stealth, bend left](G2) edge node [below]{$\track_t$}(H2);
	\end{tikzpicture}
	\caption{Construction of impairment-indicating and repair-indicating application conditions and shift of overlaps}
	\label{fig:appendix_for_correctness_of_shift_and_co}
\end{figure}

\begin{figure}
	\centering
	\begin{tikzpicture}[scale = 1]
		\node (G) at (-1,0){$PL$};
		\node (D1) at (-3,0){$D$};

		\draw[right hook-stealth](D1) edge node [above] {$g$}(G);

		\node (K1) at (-3,2) {$K$};
		\node (L1) at (-1,2) {$L$};

		\draw[right hook-stealth](K1) edge node [above] {$l$}(L1);

		\node (G2) at (-1,-2){$G$};
		\node (D2) at (-3,-2){$D'$};

		\node (Label2) at (-2,-1){$(2)$};

		\draw[right hook-stealth] (D2) edge node [above]{$g'$}(G2);

		\node (L2) at (-5,2) {$P$};

		\draw[right hook-stealth] (G) edge node [fill=white]{$m'$}(G2);
		\draw[right hook-stealth] (D1) edge node [fill=white]{$k'$}(D2);

		\draw[dashed, -stealth] (L2) edge node [fill=white]{$x$}(D1);
		\draw[dashed, -stealth, bend right=15] (L2) edge node [fill=white]{$x'$}(D2);

		\draw[left hook-stealth] (K1) edge node [fill=white]{$k$}(D1);

		\draw[right hook-stealth] (L1) edge node [fill=white]{$i_L$}(G);
		\draw[right hook-stealth, bend left=10] (L2) edge node [near end,fill=white]{$i_P$}(G);

		\draw[right hook-stealth, bend left] (L1) edge node [fill=white]{$m$}(G2);
		\node (Label1) at (-2,1.3){$(1)$};
	\end{tikzpicture}
	\caption{Construction of impairment-indicating and repair-indicating application conditions and shift of overlaps}
	\label{fig:appendix_for_correctness_of_shift_and_co_second_morphism}
\end{figure}

\begin{lem}\label{lem:totality_track}
	Given a transformation $t \colon G \Longrightarrow_{m, \rho} H$ via some rule $\rho = \completeRle$ at match $m$, and a morphism $p \colon P \inj G$, as shown in \cref{fig:appendix_for_correctness_of_shift_and_co}, then $\track_t \circ p$ is total if and only if there is a morphism $x' \colon P \inj D'$ with $p = g' \circ x'$.
\end{lem}
\begin{proof}
	\begin{enumerate}
		\item If $\track_t \circ p$ is total, $g'^{-1} \circ p$ is total, i.e., for each $e \in P$ there is an element $e' \in D'$ with $g(e') = p(e)$. Therefore $x' := g'^{-1} \circ p$ satisfies $g \circ x' = p$.
		\item If there is a morphism $x' \colon P \inj D'$ with $p = g' \circ p$, the morphism $g'^{-1} \circ p$ is total, since for each $e \in P$ there is an element $e' \in D'$ with $p(e) = g'(e')$. This implies that $\track_t \circ p$ is also total. \qedhere
	\end{enumerate}
\end{proof}

\begin{lem}[Correctness of $\pre$ and $\post$]\label{lemma:correctness_pre_post}
	Given a rule $\rho = \completeRle$, a transformation $t \colon G \Longrightarrow_{\rho,m} H$ via $\rho$ at match $m$, a condition $d$ over some graph $P$, and an overlap $(i_L, i_P,PL) \in \Over(\rho,L)$, for each morphism $m' \colon PL \inj G$ with $m = m' \circ i_L$, we have
	\begin{enumerate}
		\item $m' \models \pre_{\rho}((i_L,i_P,PL),d)$ if and only if $m' \circ i_P \models d$, and
		\item $m' \models \post_{\rho}((i_L,i_P,PL),d)$ if and only if ($\track_t \circ m' \circ i_P$ is total and $\track_t \circ m' \circ i_P \models d$) or $\track_t \circ m' \circ i_P$ is not total.
	\end{enumerate}
\end{lem}

\begin{proof}
	\begin{enumerate}
		\item As $\pre_{\rho}((i_L,i_P,PL),d) = \shift(i_P,d)$, we have $m' \models \pre_{\rho}((i_L,i_P,PL),d)$ if and only if $m' \circ i_P \models d$ (\cref{lemma:correctness_shift}).
		\item First, we assume that $m' \models \post_{\rho}((i_L,i_P,PL),d)$, then the induced transformations of $(i_L, i_P,PL)$ are either parallel dependent or parallel independent:
		\begin{enumerate}
			\item If the induced transformations of $(i_L, i_P,PL)$ are parallel dependent, the overlap $(i_L, i_P,PL)$ is contained in $\OverPre(\rho,P)$, i.e., there is no morphism $x \colon P \inj D$ with $g \circ x = i_P$ (\cref{fig:appendix_for_correctness_of_shift_and_co}). To show that $\track_t \circ m' \circ i_P$ is not total it suffices to show that there is no morphism $x' \colon P \inj D'$ with $m' \circ i_P = g' \circ x'$ (\cref{fig:appendix_for_correctness_of_shift_and_co_second_morphism}) as this implies that $g'^{-1} \circ m' \circ i_P$ and therefore especially $\track_t \circ m' \circ i_P$ is not total (\cref{lem:totality_track}). Since the squares (1) + (2) and (2) in \cref{fig:appendix_for_correctness_of_shift_and_co_second_morphism} are pushouts ((1) + (2) is part of the transformation $t$ and (1) is part of the induced transformation of $(i_L,i_P,PL)$), the square (2) is also a pushout and a pullback as we are in the category of graphs and $g$ and $k'$ are injective. If we assume that there is a morphism $x' \colon P \inj D'$ with $m' \circ i_P = g' \circ x'$, there is also a morphism $x \colon P \inj D$ such that $x' = k' \circ x$ and especially $i_P = g \circ x$ (universal property of pullbacks).
			This is a contradiction as the induced transformations of $(i_L,i_P,PL)$ are parallel dependent.
			\item If the induced transformations of $(i_L, i_P,PL)$ are parallel independent, i.e., there is a morphism $x \colon P \inj D$ with $i_P = g \circ x$ and $\post_{\rho}((i_L,i_P,PL),d) = \leftshift(\rho', \shift(h \circ x, d))$ where $\rho' = \rle{PR}{h}{D}{g}{PL}$ is the induced rule and $g, h$ and $x$ are the morphisms as shown in \cref{fig:appendix_for_correctness_of_shift_and_co}, the morphism $m'$ satisfies $\leftshift(\rho', \shift(h \circ x, d))$ if and only if $n' \colon PR \inj H$ satisfies $\shift(h \circ x, d)$ (\cref{lemma:correctness_left}). \Cref{lemma:correctness_shift} implies that $n' \models \shift(h \circ x, d)$ if and only if $n' \circ h \circ x \colon P \inj H$ satisfies $d$. Therefore, it remains to show that $\track_t \circ m' \circ i_P$ is total and that $\track_t \circ m' \circ i_P = n' \circ h \circ x$.
			As there is the morphism $x \colon P \inj D$ with $i_P = g \circ x$ and the square on the bottom right of \cref{fig:appendix_for_correctness_of_shift_and_co} is a pullback, there is also a morphism $x' \colon P \inj D'$ with $m' \circ i_P = g' \circ x'$ (universal property of pullbacks) and \cref{lem:totality_track} implies that $\track_t \circ m' \circ i_P$ is total. We also have
			\begin{align*}
				\track_t \circ m' \circ i_P &= \track_t \circ g' \circ k' \circ x \\
				&= h' \circ g'^{-1} \circ g' \circ k' \circ x \\
				&= h' \circ k' \circ x \\
				&= n' \circ h \circ x.
			\end{align*}
		\end{enumerate}
		\item Now, we assume that either $\track_t \circ m' \circ i_P$ is total and $\track_t \circ m' \circ i_P \models d$ or $\track_t \circ m' \circ i_P$ is not total implies that $m' \models \post_{\rho} ((i_L,i_P,PL),d)$.
		\begin{enumerate}
			\item If $\track_t \circ m' \circ i_P$ is not total, there is no morphism $x' \colon P \inj D'$ with $m' \circ i_P = x' \circ g'$ (\cref{lem:totality_track}) and there is no morphism $x \colon P \inj D$ with $i_P = g \circ x$
			(universal property of pullbacks), i.e., the induced transformations of $(i_L, i_P,PL)$ are parallel dependent and $\post_{\rho} ((i_L,i_P,PL),d) = \true$. So $m' \models \post_{\rho} ((i_L,i_P,PL),d)$.
			\item If $\track_t \circ m' \circ i_P$ is total and $\track_t \circ m' \circ i_P \models d$, there is a morphism $x' \colon P \inj D'$ with $m' \circ i_P = g' \circ x'$ and a morphism $x \colon P \inj D$ with $i_P = g \circ x$ (\cref{lem:totality_track} and the universal property of pullbacks, i.e., the induced transformations of $(i_L, i_P,PL)$ are parallel independent and $\post_{\rho}((i_L,i_P,PL),d) = \leftshift(\rho', \shift(h \circ x, d))$. The morphism $m'$ satisfies $\leftshift(\rho', \shift(h \circ x, d))$ if and only if $n' \models \shift(h \circ x, d)$ (\cref{lemma:correctness_left}). Additionally, $n' \models \shift(h \circ x, d)$ if and only if $n' \circ h \circ x \models d$ (\cref{lemma:correctness_shift}). In the previous part of the proof, we have shown that $n' \circ h \circ x = \track_t \circ m' \circ i_P$. As $\track_t \circ m' \circ i_P \models d$ we have $n' \circ h \circ c \models d$, $n' \models \shift(h \circ x, d)$ and $m' \models \leftshift(\rho', \shift(h \circ x, d))$ (\cref{lemma:correctness_shift} \cref{lemma:correctness_left}). \qedhere
		\end{enumerate}
	\end{enumerate}
\end{proof}

The following lemma ensures that every violation of a repair-indicating application condition contains a repaired morphism w.r.t.~a transformation, and that every repaired morphism w.r.t.~a transformation contained in a violation of exactly one repair-indicating application condition.

\begin{lem}
	Given a constraint $c = \forall(e \colon \emptyset \inj P,d)$, a rule $\rho = \completeRle$, a transformation $t \colon G \Longrightarrow_{\rho,m} H$ via $\rho$ at match $m$, then
	\begin{enumerate}
		\item for each application condition $ac \in \repair(\rho,c)$, constructed via an overlap $[(i_L,i_P,PL)] \in \OverEq(\rho,P)$ and for each morphism $m' \colon PL \inj G$ contained in $\sv_m(ac)$, $m' \circ i_P$ is a repaired morphism w.r.t.~$t$, and
		\item for each repaired morphism $p \colon P \inj G$ w.r.t.~$t$, there is exactly one application condition $ac \in \repair(\rho,c)$ constructed via an overlap $[(i_L,i_P,PL)] \in \OverEq(\rho,P)$ so that there is a violation $m' \in \sv_m(ac)$ with $p = m' \circ i_P$.
	\end{enumerate}
\end{lem}
\begin{proof}
	\begin{enumerate}
		\item Given a repair indicating application condition $ac\in \repair(\rho,c)$ constructed via an overlap $[(i_L,i_P,PL)] \in \OverEq(\rho,P)$ and a violation $p \colon PL \inj G \in \sv_m(ac)$. We need to show that $m' \circ i_P$ is a repaired morphism w.r.t.~$t$. As $m' \not \models \post_{\rho}((i_L,i_P,PL),d) \implies \pre_{\rho}((i_L,i_P,PL),d)$, we have $m' \models \post_{\rho}((i_L,i_P,PL),d)$ and $m' \not \models \pre_{\rho}((i_L,i_P,PL),d)$. Let us first consider the second part of the implication: By correctness of the $\pre$-operator (\cref{lemma:correctness_pre_post}), we have $m' \circ i_P \not \models d$.
		When considering $\post_{\rho}((i_L,i_P,PL),d)$ we need to consider whether the induced transformations of $(i_L,i_P,PL)$ are parallel dependent or not:
		\begin{enumerate}
			\item If the induced transformations of $(i_L,i_P,PL)$ are parallel dependent, there is no morphism $x \colon P \inj D$ with $i_P = x \circ g$ (\cref{fig:appendix_for_correctness_of_shift_and_co}), i.e., $\post_{\rho}((i_L,i_P,PL),d) = \true$, so $m' \models \post_{\rho}((i_L,i_P,PL),d)$. It remains to show that $\track_t \circ p$ is not total:
			The universal property of pullbacks implies that there is no morphism $x' \colon P \inj D'$ with $ m' \circ i_P = g' \circ x'$ (compare \cref{fig:appendix_for_correctness_of_shift_and_co_second_morphism}). So $\track_t \circ p$ is not total (\cref{lem:totality_track}) and $p \circ i_P$ is a repair of the premise w.r.t.~$t$.
			\item If the induced transformations of $(i_L,i_P,PL)$ are parallel independent, we have $\post_{\rho}((i_L,i_P,PL),d) = \leftshift(\shift(h \circ x,d))$, where $h$ and $x$ are the morphisms shown in \cref{fig:appendix_for_correctness_of_shift_and_co}. The existence of $x$ implies that $\track_t \circ p \circ i_P$ is total (universal property of pullbacks and \cref{lem:totality_track}). As $m' \models \post_{\rho}((i_L,i_P,PL),d)$, \cref{lemma:correctness_pre_post} implies that $\track_t \circ m' \circ i_P\models d$.
			In total, $m'$ is a repair of the conclusion w.r.t.~$t$.
		\end{enumerate}
		\item Given a repaired morphism $p \colon P \inj G$ w.r.t.~$t$, then $p$ is either a repair of the premise or a repair of the conclusion w.r.t.~$t$ and there is an overlap $(i_L, i_P,PL) \in \Over(\rho,P)$ and a morphism $m' \colon PL \inj P$ with $m = m' \circ i_L$ and $p = m' \circ i_P$. We need to show that $m' \in sv_m(ac)$, i.e., that $m' \models \post_{\rho}((i_L,i_P,PL),d)$ and $m' \not \models \pre_{\rho}((i_L,i_P,PL),d)$. As $p$ is a repair, we have $p = m' \circ i_P \not \models d$ and \cref{lemma:correctness_pre_post} implies that $m' \models \pre_{\rho}((i_L,i_P,PL),d)$. When considering $\post_{\rho}((i_L,i_P,PL),d)$, we need to consider whether $\track_t \circ p$ is total or not:
		\begin{enumerate}
			\item If $\track_t \circ p$ is not total, \cref{lem:totality_track} implies that there is no morphism $x' \colon P \inj D'$ with $p = g' \circ x'$ (compare \cref{fig:appendix_for_correctness_of_shift_and_co}). Therefore, there is also no morphism $x \colon P \inj D$ with $i_P = g \circ x$
			(universal property of pullbacks), i.e., $\post_{\rho}((i_L,i_P,PL),d) = \true$ and $m' \models \post_{\rho}((i_L,i_P,PL),d)$. So $m' \in \sv_m(ac)$.
			\item If $\track_t \circ p$ is total, there is a morphism $x' \colon P \inj D'$ with $p = x' \circ g'$ (\cref{lem:totality_track}) and there is a morphism $x \colon P \inj D$ with $i_P = g \circ x$ (universal property of pullbacks). So $\post_{\rho}((i_L,i_P,PL),d) = \leftshift(\rho', \shift(h \circ x, d))$. As $p$ is a repaired morphism and $\track_t \circ p$ is total, $\track_t \circ p = \track_t \circ m' \circ i_P \models d$. \Cref{lemma:correctness_pre_post} implies that $m' \models \post_{\rho}((i_L,i_P,PL),d)$ and $m' \in \sv_m(ac)$.
		\end{enumerate}
		Lastly we have to show that for each repaired morphism $p \colon P \inj G$ there is exactly one application condition $ac \in \repair(\rho,c)$ constructed over an overlap $[(i_L, i_P,PL)] \in \OverEq(\rho,P)$ so that there is a violation $m' \in \sv_m(ac)$ with $p = m' \circ i_P$. We assume that there are two such application condition $ac, ac' \in \repair(\rho, c)$ constructed over overlaps $[(i_L, i_P,PL)]$ and $[(i'_L,i'_P,PL')] \in \OverEQ(\rho,c)$. There are morphisms $m' \colon PL \inj G$ and $m''\colon PL' \inj G$ with
		\begin{align*}
			p &= m' \circ i_P = m'' \circ i'_P \text{ and }\\
			m &= m' \circ i_L = m'' \circ i'_L.
		\end{align*}
		\Cref{lem:equiv_iso_mor} implies that $[(i_L, i_P,PL)] \cong [(i'_L,i'_P,PL')]$ and we have $ac = ac'$. \qedhere
	\end{enumerate}
\end{proof}
With the following results, we can extract sets of repaired and impaired morphisms by restricting the violations of each repair- (impairment-)indicating application condition to the embedding of the premise into the overlap graph, this application condition was computed with.

\begin{cor}\label{lem:all_repaired_morphisms_found}
	Given a constraint $c = \forall(e \colon \emptyset \inj P,d)$, a rule $\rho = \completeRle$, a transformation $t \colon G \Longrightarrow_{\rho,m} H$, then
	$$\dot{\bigcup}_{ac \in \repair(\rho,c)} \{m' \circ i_P \mid m' \in \sv_m(ac)\} = \{p \colon P \inj G \mid p \text{ is a repaired morphism w.r.t.~$t$}\},$$
	where $ac = \forall(i_L \colon L \inj PL, d')$ is the application condition constructed with the overlap $(i_L,i_P, PL) \in \OverEq(\rho,P)$.
\end{cor}

\begin{cor}\label{cor:all_impaired_morphisms_found}
	Given a constraint $c = \forall(e \colon \emptyset \inj P,d)$, a rule $\rho = \completeRle$, a transformation $t \colon G \Longrightarrow_{\rho,m} H$, then
	$$\dot{\bigcup}_{ac \in \violation(\rho,c)} \{m' \circ i_P \mid m' \in \sv_m(ac)\} = \{p \colon P \inj G \mid p \text{ is an impaired morphism w.r.t.~$t$}\},$$
	where $ac = \forall(i_L \colon L \inj PL, d')$ is the application condition constructed with the overlap $(i_L,i_P, PL) \in \OverEq(\rho,P)$.
\end{cor}
It remains to show that the number of violations of repair-indicating application conditions is equal to the number of repaired morphisms w.r.t.~a transformation and, therefore, the number of violations of impairment-indicating application conditions is equal to the number of impaired morphisms w.r.t.~a transformation.
\begin{lem}\label{lem:filter_does_not_change}
	Given a constraint $c = \forall(e \colon \emptyset \inj P,d)$, a rule $\rho = \completeRle$, a transformation $t \colon G \Longrightarrow_{\rho,m} H$ via $\rho$ at match $m$, and an application condition $ac = \forall(i_L \colon L \inj PL, d') \in \repair(\rho, c)$ constructed using the overlap $(i_L, i_P,PL)$, then
	$$\nv_m(ac) = |\{p \circ i_P \colon P \inj G \mid p \in \sv_m(ac)\}|$$
\end{lem}

\begin{proof}
	\begin{enumerate}
		\item [\enquote{$\leq$}:] Given two morphisms $p,p'\colon PL \inj G \in \sv_m(ac)$ with $p \neq p'$. Since $i_L$ and $i_P$ are jointly surjective and $p \circ i_L = m = p' \circ i_L$, we have $p \circ i_P \neq p' \circ i_P$ and
		$$\nv_m(ac) \leq |\{p \circ i_P \colon P \inj G \mid p \in \sv_m(ac)\}|.$$
		\item [\enquote{$\geq$}:] Given two morphisms $p, p' \in |\{pl \circ i_P \colon P \inj G \mid pl \in \sv_m(ac)\}|$ with $p \neq p'$, then there are two morphisms $pl, pl' \colon PL \inj G \in \sv_m(c)$ with $m = pl \circ i_L = pl' \circ i_L$, $p = pl \circ i_P$ and $p' = pl' \circ i_P$. Since $p = pl \circ i_P \neq pl' \circ i_P = p'$ it follows that $pl \neq pl'$ and we have
		\begin{equation*}
			\nv_m(ac) \geq |\{pl \circ i_P \colon P \inj G \mid pl \in \sv_m(ac)\}|. \qedhere
		\end{equation*}
	\end{enumerate}
\end{proof}

We can evaluate the change in consistency by only considering impaired and repaired morphisms.
\begin{lem}\label{lem:correct_counting}
	Given a constraint $c = \forall(e \colon \emptyset \inj P, d)$ and a transformation $t \colon G \Longrightarrow_{\rho,m} H$, then
	\begin{align*}
		\nv_H(c) - \nv_G(c) = &|\{p \colon P \inj H \mid p \text{ is impaired w.r.t.~$t$}\}| - \\
		&|\{p \colon P \inj G \mid p \text{ is repaired w.r.t.~$t$}\}|.
	\end{align*}
\end{lem}

\begin{proof}
	We can rewrite the set of violations of $H$ as
	\begin{align*}
		\sv_H(c) &= \{ p \colon P \inj H \mid \text{$p$ is an impaired morphism w.r.t.~$t$}\}\\&\cup\{\track_t \circ p \mid p \in \sv_G(c) \text{ and $p$ is not a repaired morphism w.r.t.~$t$}\}.
	\end{align*}
	So we can evaluate $\nv_H(c)$ as
	\begin{align*}
		\nv_H(c) = \nv_G(c) &+ |\{p \mid \text{$p$ is an impaired morphism w.r.t.~$t$}\}|\\ &-|\{p \mid \text{$p$ is a repaired morphism w.r.t.~$t$}\}|
	\end{align*}
	and the statement follows immediately.
\end{proof}

With these prerequisites in place, we can now prove \cref{thm:main_theorem}.
\begin{proof}[Proof of Theorem~\ref{thm:main_theorem}]
	Given a transformation $t \colon G \Longrightarrow_{\rho,m} H$ and a constraint $c = \forall(e\colon \emptyset \inj P,d)$, with \cref{lem:correct_counting} the difference $\nv_H(c) - \nv_G(c)$ can be evaluated by only considering impaired and repaired morphisms. \Cref{lem:all_repaired_morphisms_found} and \Cref{cor:all_impaired_morphisms_found} show that we can find repaired and impaired morphisms by evaluating the application conditions contained in the sets $\repair(\rho,c)$ and $\violation(\rho,c)$, respectively.
	By using \cref{lem:all_repaired_morphisms_found} and \cref{lem:filter_does_not_change} we get

	\begin{align*}
		|\{p \colon P \inj G \mid p \text{ is a repaired morphism w.r.t.~$t$}\}| &= |\dot{\bigcup_{ac \in \repair(\rho,c)}} \{p \circ i_P \mid p \in \sv_m(ac)\}| \\&=
		\sum_{ac \in \repair(\rho,c)} |\{p \circ i_P \mid p \in \sv_m(ac)\}| \\&= \sum_{ac \in \repair(\rho,c)} \nv_m(ac)
	\end{align*}
	When considering $|\{p \colon P \inj G \mid p \text{ is an impaired morphism w.r.t.~$t$}\}|$ we utilize the fact each impaired morphism is a repaired morphism of the inverse transformation $t^{-1} \colon H \Longrightarrow_{\rho^{-1}, n} G$ (\cref{lem:relation_impair-repair}):
	\begin{align*}
		&|\{p \colon P \inj H \mid p \text{ is an impaired morphism w.r.t.~$t$}\}| \\= &|\{p \colon P \inj H \mid p \text{ is a repaired morphism w.r.t.~$t^{-1}$}\}|
	\end{align*}
	Analogous to the first part of the proof, we can show that
	\begin{align*}
		|\{p \colon P \inj H \mid p \text{ is an impaired morphism w.r.t.~$t$}\}| &= \sum_{ac \in \repair(\rho^{-1},c)} \nv_n(ac) \\
		&= \sum_{ac \in \violation(\rho,c)} \nv_m(ac)
	\end{align*}
	In total, we obtain
	\begin{equation*}
		\nv_H(c) - \nv_G(c) = \sum_{ac \in \violation(\rho,c)} \nv_m(ac) - \sum_{ac \in \repair(\rho,c)} \nv_m(ac) \qedhere
	\end{equation*}
\end{proof}

\subsection{Proofs of Section~\ref{sec:comparisson}}

\begin{proof}[Proof of Theorem~\ref{thm:direct_sustaining_improving_ac}]
	Given a transformation $t \colon G \Longrightarrow_{\rho,m} H$ and a constraint $c = \forall(e \colon \emptyset \inj P,d)$,
	\begin{enumerate}
		\item we start by showing that $t$ is direct consistency-sustaining w.r.t.~$c$ if and only if $m \models \bigwedge_{ac \in \violation(\rho,c)} ac$.
		\begin{enumerate}
			\item [\enquote{$\Longrightarrow$}:] Let $t$ be direct consistency-sustaining w.r.t.~$c$, i.e., there is no impaired morphism w.r.t.~$t$. As \cref{cor:all_impaired_morphisms_found} implies that each violation of an application condition $ac \in \violation(\rho,c)$ contains an impaired morphism, there is no violation for any $ac$ and $m \models \bigwedge_{ac \in \violation(\rho,c)}ac$.
			\item[\enquote{$\Longleftarrow$}:] Let $m \models \bigwedge_{ac \in \violation(\rho,c)}ac$, then \cref{cor:all_impaired_morphisms_found} implies that there are no impaired morphisms w.r.t.~$t$, i.e., $t$ is direct consistency-sustaining w.r.t.~$c$
		\end{enumerate}
		\item Now, we show that $t$ is direct consistency-improving w.r.t.~$c$ if and only if $m \models \bigwedge_{ac \in \violation(\rho,c)} ac \wedge \neg (\bigwedge_{ac \in \repair(\rho,c)} ac)$.
		\begin{enumerate}
			\item[\enquote{$\Longrightarrow$}:] Let $t$ be direct consistency-improving w.r.t.~$c$. Then $t$ is also directly consistency-sustaining w.r.t.~$c$ and the first part of \cref{thm:direct_sustaining_improving_ac} implies that $m \models \bigwedge_{ac \in \violation(\rho,c)} ac$. Also, there is at least one repaired morphism w.r.t.~$t$ and \cref{lem:all_repaired_morphisms_found} implies that this repaired morphism is a violation of an application condition $ac \in \repair(\rho,c)$. So $m \models \neg (\bigwedge_{ac \in \repair(\rho,c)} ac)$.
			\item[\enquote{$\Longleftarrow$}:] Let $m \models \bigwedge_{ac \in \violation(\rho,c)} ac \wedge \neg (\bigwedge_{ac \in \repair(\rho,c)} ac)$, the first part of \cref{thm:direct_sustaining_improving_ac} implies that $m \models \bigwedge_{ac \in \violation(\rho,c)} ac$. At least one of the application conditions contained $\repair(\rho,c)$ is violated. Since each violation of an application condition $ac \in \repair(\rho,c)$ contains a repaired morphism w.r.t.~$t$ (\cref{lem:all_repaired_morphisms_found}) $t$ is a direct consistency-improving transformation w.r.t.~$c$. \qedhere
		\end{enumerate}
	\end{enumerate}
\end{proof}

\end{document}